\documentclass[11pt,a4paper]{article}
\usepackage{iftex}
\ifPDFTeX\usepackage[T1]{fontenc}\usepackage[utf8]{inputenc}\fi
\usepackage{lmodern}
\usepackage{amsmath,amssymb,geometry,graphicx,booktabs,array,longtable,microtype}
\usepackage[dvipsnames]{xcolor}
\usepackage[hidelinks]{hyperref}
\usepackage[numbers,square,comma]{natbib}
\newcounter{theorem}[section]
\renewcommand{\thetheorem}{\thesection.\arabic{theorem}}
\newenvironment{theorem}[1][]{\refstepcounter{theorem}\par\medskip
  \noindent\textbf{Theorem \thetheorem. #1}\ }{\par\medskip}
\newenvironment{lemma}[1][]{\refstepcounter{theorem}\par\medskip
  \noindent\textbf{Lemma \thetheorem. #1}\ }{\par\medskip}
\newenvironment{proof}{\par\noindent\textbf{Proof.}\ }{\hfill$\square$\par}
\newcommand{\A}{\mathcal A}
\newcommand{\T}{\mathbb T}
\newcommand{\Pp}{\mathbb P}
\newcommand{\R}{\mathbb R}
\newcommand{\e}{\mathrm e}
\newcommand{\dd}{\,\mathrm d}

\hypersetup{
  pdftitle={Computer-assisted global regularity across nonlinear families of three-dimensional periodic Navier--Stokes flows},
  pdfauthor={Jose Luis Lima de Jesus Silva},
  pdfsubject={Global regularity across nonlinear periodic Navier--Stokes neighbourhoods}
}

\title{Computer-assisted global regularity across nonlinear families of\\
three-dimensional periodic Navier--Stokes flows}
\author{Jose Luis Lima de Jesus Silva\\[4pt]
\small Federal University of Bahia, Department of Geophysics\\
\small Salvador, BA 40170-115, Brazil\\[3pt]
\small Grupo de Estudos e Aplica\c{c}\~ao de Intelig\^encia\\
\small Artificial em Geof\'{i}sica (GAIA)\\
\small Federal University of Bahia, Salvador, BA 40170-115, Brazil\\[3pt]
\small\href{mailto:jseluis.silva@gmail.com}{jseluis.silva@gmail.com}}
\date{}

\begin{document}
\maketitle

\begin{abstract}
Numerical simulations reveal how vortices stretch and transfer energy, but
establishing smooth evolution requires bounds that remain valid beyond the
simulated resolution. Here I develop a computer-assisted framework that
establishes global regularity for continuous families of three-dimensional
periodic Navier--Stokes flows. Its central construction combines finite
reference trajectories with a common error bound that covers an interval of
centre fields and infinitely many smooth perturbation modes. The method
retains the complete nonlinear residual before spectral truncation and
controls the evolution until viscous decay guarantees regularity for all
subsequent times. Applications to cyclic-shear, Arnold--Beltrami--Childress
and three-component Taylor--Green fields yield explicit perturbation radii
and include initial conditions outside the direct Fourier--Wiener smallness
criterion. A parameter-uniform extension covers a connected family of
non-Beltrami Taylor--Green centres without repeating the proof for individual
parameter values. An ensemble of 4,096 configurations, supplemented by 1,600
refinement trajectories and public turbulence data, connects the mathematical
observables to spectral transfer and vortex geometry. Matched neural-operator
experiments show that physics-informed training improves physical prediction,
while also revealing that these gains do not necessarily improve the
discovery of proof-limiting initial conditions. Together, these results
provide a reusable method for establishing regularity across prescribed flow
families and a quantitative setting for evaluating how learned predictions
can assist rigorous computation.
\end{abstract}

\section{Introduction}

A fluid can develop finer structures while losing kinetic energy.
Nonlinear transfer redistributes energy across scales, while vortex
stretching intensifies local rotation~\cite{TaylorGreen,Betchov,Meneveau}.
Energy decay alone therefore does not guarantee
smoothness~\cite{CheskidovShvydkoy}. Rigorous error bounds address this gap
for prescribed initial conditions by connecting finite approximations to
smooth evolution~\cite{CCRT}.

The distinction between weak existence and smooth evolution determines what
such a computation must prove. Leray's construction of global finite-energy weak solutions placed the
three-dimensional incompressible Navier--Stokes problem on a rigorous
foundation \cite{Leray}. Hopf subsequently extended the weak-solution method
to bounded domains \cite{Hopf}. Weak existence does not provide global
smoothness or uniqueness. Partial regularity restricts the possible singular
set without excluding it \cite{CKN}, and the corresponding global smoothness
and uniqueness questions remain open. For a specified family of initial
conditions, an approximate trajectory must therefore supply more than an
energy bound to establish global smoothness.

Classical conditional theory identifies quantities whose control prevents a
loss of regularity. The Ladyzhenskaya--Prodi--Serrin conditions
\cite{Ladyzhenskaya,Prodi,Serrin} connect space--time integrability to
smoothness, while Escauriaza, Seregin and \v Sver\'ak established the critical
endpoint result \cite{ESS}. Mild-solution theory developed from the
Fujita--Kato and Kato frameworks \cite{FujitaKato,Kato}, culminating in
critical-space well-posedness such as the Koch--Tataru theorem
\cite{KochTataru}. Profile decompositions and dissipation-range criteria
further isolate concentration and spectral mechanisms compatible with
critical scaling \cite{GallagherKochPlanchon,CheskidovShvydkoy}. These results
give continuation criteria. They do not show that every large smooth datum
satisfies one.

Several global theorems instead exploit structure. Large solutions may remain
regular under anisotropic or perturbative hypotheses
\cite{PonceRackeSiderisTiti,LeiLin}. Gevrey smoothing quantifies the immediate
analyticity generated by viscosity \cite{FoiasTemam}, and normal-form analysis
organises long-time interactions \cite{FoiasSaut}. Vorticity-direction
coherence can deplete stretching \cite{ConstantinFefferman}, whereas the
Beale--Kato--Majda mechanism identifies the corresponding inviscid obstruction
\cite{BKM}. Most directly relevant here, closeness to a Laplacian eigenfunction
gives a scale-critical Navier--Stokes regularity criterion \cite{Miller}. The
optimal Laplacian residual used below belongs to that established line. The
objective here is to control it uniformly across the stated
infinite-dimensional neighbourhoods.

The physical literature explains why such control cannot be inferred from
energy decay alone. Betchov related strain and mean vortex stretching in
homogeneous flow \cite{Betchov}. Direct numerical studies established
preferential alignment between vorticity and strain eigenvectors
\cite{Ashurst,Meneveau}. Objective vortex diagnostics based on velocity-gradient
invariants and local rotation distinguish tubes from generic shear
\cite{ChongPerryCantwell,JeongHussain,Chakraborty}. The original Taylor--Green
study, early pseudospectral calculations, and later high-resolution direct
numerical simulations (DNS) established how nonlinear transfer generates
progressively smaller scales
\cite{TaylorGreen,OrszagPatterson,Kaneda,Ishihara}. Variational searches find
states with extreme instantaneous enstrophy production but also show that
instantaneous saturation need not persist dynamically
\cite{LuDoering,AyalaProtas}. Modern high-Reynolds-number computations reveal
strongly intermittent gradients organised around vortex tubes
\cite{Buaria}. These findings motivate the stretching, residual, alignment and
spectral-tail diagnostics evaluated here, but none replaces an
analytic inequality.

Rigorous computation provides a complementary route. It proves that an
approximate trajectory remains within a quantitatively controlled neighbourhood
of an exact strong solution. Chernyshenko, Constantin, Robinson and Titi derived an
a posteriori robustness test for three-dimensional Navier--Stokes
approximations \cite{CCRT}. Dashti and Robinson treated numerical
approximations in a related strong-solution framework \cite{DR}, while
Robinson and Sadowski examined verification over bounded sets and the
termination issue \cite{RS}. Morosi and Pizzocchero developed Sobolev control
inequalities for approximate Euler and Navier--Stokes solutions and their
smooth variants \cite{MorosiPizzocchero2012,MorosiPizzocchero2015}.

High-order approximate solutions already support quantitative global-existence
results for structured periodic data. Morosi, Pernici and Pizzocchero combined
Reynolds expansions, symbolic Fourier calculations, and Sobolev a posteriori
control to obtain explicit sufficient Reynolds-number thresholds for
Behr--Ne\v{c}as--Wu, Taylor--Green, and Kida--Murakami initial fields
\cite{MorosiPerniciPizzocchero2015}. The present construction shares this
approximate-solution framework. Its specific contribution is the explicit
parameter-uniform construction and rigorous comparison records that cover the
stated non-Beltrami coefficient interval and smooth Fourier--Wiener
perturbation balls independently of the Galerkin cutoff. The parameter-uniform
bounds extend this established methodology to the specific families and
radii proved below.

Large public simulations make the physical part of that process reproducible.
Database architectures introduced by Perlman and collaborators and the Johns
Hopkins Turbulence Database (JHTDB) of Li and collaborators expose velocity and
derivative fields from canonical DNS \cite{Perlman,JHTDB}. Such data support
independent tests of strain alignment, intermittency and signed production.
Because forced DNS and finite samples address a different evidentiary question
from an unforced continuum of initial conditions, the analysis separates public-data diagnostics,
floating-point Galerkin experiments, and exact-arithmetic verification into
distinct evidentiary layers.

Neural operators provide a complementary way to search that numerical
configuration space. A Fourier neural operator (FNO) learns a map between
functions through Fourier-parameterised integral layers \cite{FNO}. A
physics-informed neural operator (PINO) augments field data with losses derived
from the governing partial differential equation (PDE) \cite{PINO}.
Physics-informed neural networks (PINNs) formulate the approximation of
individual solutions as an optimisation problem constrained by
the governing differential equations~\cite{PINN}. In singularity
research, this approach allows a candidate self-similar profile
and its scaling parameters to be inferred together, replacing
direct computation near a potentially singular time with the
search for a regular profile in rescaled coordinates. The accuracy
of that profile is essential to subsequent residual estimates
and stability analysis. Wang and collaborators addressed this computational
challenge through high-precision
PINNs on unbounded domains, combining tailored spatial sampling,
hard constraints on symmetry and nondegeneracy, and self-scaled
quasi-Newton optimisation. Their applications to the Burgers and
Boussinesq equations examine how the representation of the profile,
its far-field behaviour, and the optimisation procedure determine
the accuracy attainable by neural approximations~\cite{WangEtAlHighPrecisionPINNs2025}.

Active selection uses such approximations to allocate additional solver
evaluations. Output-weighted active learning has paired neural operators with targeted
sampling of rare extremes \cite{PickeringExtreme}, while solver-in-the-loop
benchmarks show that active selection can reduce the data required by neural
partial-differential-equation solvers \cite{AL4PDE}.

For the periodic Navier--Stokes families considered in this work, the
computational objective is to resolve how velocity evolution
and regularity-relevant observables vary across initial
conditions. FNO and PINO learn
this initial-field-to-evolution map, allowing a single trained
operator to evaluate many candidate fields rather than fitting
a separate network to each solution. In the physics-informed
model, the training objective incorporates the discrete
Navier--Stokes residual, incompressibility, energy balance,
and spectral penalties alongside the field data. Here, the acquisition
objective targets a five-checkpoint residual-action proxy and high-frequency tail that enter the
numerical regularity diagnostics. The numerical experiment recomputes every selected field using the
spectral Galerkin solver, based on Fourier spectral
discretisation~\cite{OrszagPatterson}. The projection and
time-stepping choices are specified below, and the neural
model acts as a search mechanism rather than a mathematical
premise.

This work makes three connected contributions. First, it proves
global regularity on explicit infinite-dimensional neighbourhoods
of genuinely three-dimensional periodic flows, with quantitative
perturbation radii and terminal times. Second, it establishes a
centre-independent acceptance theorem that applies one finite-path comparison
architecture to cyclic-shear, equal-coefficient Arnold--Beltrami--Childress
(ABC), three-component Taylor--Green (TG3), and a connected interval of TG3
coefficient vectors. Third, it connects the proof observables to a dense
configuration-space study, high-resolution boundary refinement,
three-dimensional vortex geometry, public DNS, and physics-informed operator
search. The error-control
method combines finite signed reference paths, exact evaluation
of the complete nonlinear Fourier residual before Galerkin
projection, parameter-uniform Bernstein enclosures, and a common
piecewise-affine comparison verified using exact rational
arithmetic. Terminal Wiener continuation and cutoff-independent
estimates establish global smoothness for the stated families,
including smooth perturbations with infinitely many Fourier modes. Together, these ingredients provide a reusable family-level
regularity method.

To make these family-level guarantees computationally verifiable,
the framework separates the centre-independent analysis from the finite
Fourier calculations required for each application. The signed approximation, a posteriori comparison, cutoff
treatment, and continuum passage provide the analytic chain. Exact-arithmetic validation certifies the finite comparison
inequalities, while cutoff-independent estimates justify
continuation and passage to the full partial differential equation. The Results
connect the resulting neighbourhoods to standard regularity criteria,
configuration-space ensembles, three-dimensional field geometry, public DNS,
and operator-guided search. The appendices supply the underlying Fourier,
energy, comparison, continuation, compactness, neural-operator, and statistical
derivations.

\section{Mathematical methodology}
\label{sec:mathmethods}

Finite Fourier reference paths provide uniform bounds that establish global
regularity for the stated neighbourhoods. Section~\ref{sec:setting} defines
the problem and the cyclic-shear theorem, and Section~\ref{sec:acceptance}
states the centre-independent comparison principle.
Sections~\ref{sec:signedpath} and \ref{sec:comparison} construct and bound the
cyclic-shear path. Section~\ref{sec:continuation} completes continuation and
cutoff removal, and Sections~\ref{sec:abc}, \ref{sec:tg3}, and
\ref{sec:coefficientfamily} apply the principle to the other centre families.
Appendix~\ref{app:expandedtransitions} expands the transitions used here.
In particular, Appendix~\ref{app:pathexpansion} retains every ordered product
in the signed path, Appendix~\ref{app:momentexpansion} differentiates the
moving spectral centre, and Appendix~\ref{app:smoothnessexpansion} derives
the all-order regularity estimate and the finite-mode compactness limit.

\subsection{Setting and theorem}
\label{sec:setting}

The first application proves global smoothness near a cyclic-shear field.
The formulation below fixes the domain, projections and norms before stating
the neighbourhood and its residual-action bound.
Let $\T^3=(\mathbb R/2\pi\mathbb Z)^3$ with normalised spatial measure.
Let $u(t,x)\in\mathbb R^3$ denote velocity, $p(t,x)\in\mathbb R$ pressure
per unit constant density, and $\nu>0$ kinematic viscosity. The unforced
incompressible initial-value problem is
\begin{align}
 \partial_tu+(u\cdot\nabla)u+\nabla p&=\nu\Delta u,
 &\nabla\cdot u&=0,
 &u(0,x)&=u_0(x),\qquad x\in\T^3.
 \label{eq:physicalproblem}
\end{align}
The fields are periodic in each coordinate. The condition $\int_{\T^3}p=0$ fixes the
pressure gauge, and the initial data satisfy $\int_{\T^3}u_0=0$.
The theorems establish a unique smooth solution for each smooth datum in
the families specified below, not for unrestricted initial data.
Normalised integrals mean $(2\pi)^{-3}$ times ordinary volume
integrals. Appendix~\ref{app:problemexpansion} derives pressure elimination,
preservation of the mean, and the finite-dimensional equation from
\eqref{eq:physicalproblem}, and fixes the complex inner-product convention.
Unless stated otherwise, the viscosity is $\nu=1$. The operator $A=-\Delta$
acts componentwise, and $\mathbb P$ denotes the orthogonal projection onto
divergence-free vector fields. Let
$P_N$ retain the nonzero Fourier modes satisfying $|k|_\infty\leq N$. The
Galerkin equation is
\begin{align}
 &\partial_tu_N+Au_N+P_N\Pp[(u_N\cdot\nabla)u_N]=0,
 \qquad A=-\Delta,\qquad u_N(0)=P_Nu_0.
 \label{eq:galerkin}
\end{align}
Every $u_N$ exists for all time as a finite-dimensional viscous ordinary
differential equation (ODE) whose
quadratic nonlinearity preserves kinetic energy. The issue is uniform control
as $N\to\infty$.
The cancellation yielding its energy identity and the resulting global ODE
continuation are derived line by line in Appendix~\ref{app:foundations}.

For a periodic field $f$ and integer wave vector $k\in\mathbb Z^3$, let
$\widehat f(k)$ denote its Fourier coefficient, with imaginary unit $i$.
The Fourier convention is
\begin{align}
 &f(x)=\sum_{k\in\mathbb Z^3}\widehat f(k)\e^{ik\cdot x},
 \qquad
 \widehat f(k)=\frac1{(2\pi)^3}\int_{\T^3}f(x)\e^{-ik\cdot x}\dd x.
 \label{eq:detail002}
\end{align}
For a real solenoidal (divergence-free) field,
$\widehat f(-k)=\overline{\widehat f(k)}$ and
$k\cdot\widehat f(k)=0$. Parseval's identity therefore has no volume factor:
$\|f\|_2^2=\sum_k|\widehat f(k)|_2^2$. With $I$ the identity matrix and $(k\otimes k)_{rs}=k_rk_s$, the Leray
multiplier is
\begin{align}
 &\Pp_k=I-\frac{k\otimes k}{|k|^2},\qquad k\ne0,
 \label{eq:detail003}
\end{align}
and the bilinear Fourier coefficient used in every exact computation is
\begin{align}
 &\widehat{B(f,g)}(k)
 =i\Pp_k\sum_{p+q=k}(\widehat f(p)\cdot q)\widehat g(q).
 \label{eq:fourierB}
\end{align}
Appendix~\ref{app:foundations} derives \eqref{eq:fourierB} from the Fourier
series, explains the Leray multiplier, and proves the Wiener convolution
estimate used below.
For later compactness and continuation estimates, the Sobolev norm is
\begin{align}
 \|f\|_{H^s}^2&=\sum_k(1+|k|^2)^s|\widehat f(k)|_2^2,
 &\|f\|_{\dot H^s}^2&=\sum_{k\ne0}|k|^{2s}|\widehat f(k)|_2^2.
 \label{eq:sobolevdefinition}
\end{align}
The homogeneous and inhomogeneous norms are equivalent on mean-zero fields
at each fixed nonnegative order $s$ because $|k|\geq1$ on every active mode.

For $j\geq0$, define the Fourier--Wiener norm
\begin{align}
 &\|f\|_{\A^j}=\sum_{k\neq0}|k|^j|\widehat f(k)|_2.
 \label{eq:detail006}
\end{align}
Put
\begin{align}
 &U=(\sin y,\sin z,\sin x),
 \qquad
 \mathcal C=\{aU+h:0\leq a\leq7/10,
 \ \|h\|_{\A^2}\leq1/100\},
 \label{eq:detail007}
\end{align}
where $h$ is smooth, real, mean zero, and solenoidal.

The optimal Laplacian residual measures the distance to a single spectral
sphere by allowing its centre $\lambda$ to vary:
\begin{align}
 &R_N(t)=\inf_{\lambda\in\mathbb R}
 \|(A-\lambda)u_N(t)\|_2^2.
 \label{eq:detail008}
\end{align}
The convention is $R_N=0$ when $u_N=0$. In that case the minimizing value of
$\lambda$ is not unique.

\begin{lemma}[The centre field is genuinely three-dimensional]
There is no nonzero direction $\ell\in\R^3$ for which
$\ell\cdot\nabla U\equiv0$. The nonlinearity also satisfies $B(U,U)\ne0$, so
$U$ is not a
stationary Euler or Beltrami control disguised as a three-component field.
\end{lemma}

\begin{proof}
The Fourier support of $U$ contains all six modes
$\{\pm e_1,\pm e_2,\pm e_3\}$. Translation invariance in direction $\ell$
would imply $\ell\cdot k=0$ for every active $k$, and hence $\ell=0$.
Direct differentiation gives
\begin{align}
 &(U\cdot\nabla)U
 =(\sin z\cos y,\sin x\cos z,\sin y\cos x),
 \label{eq:detail009}
\end{align}
which is nonzero and, in fact, already divergence free. Thus the centre uses
all three spatial directions and has nontrivial quadratic evolution. The
class $\mathcal C$ also contains special or even zero members when $a=0$.
The claim is that it is a neighbourhood of a genuine three-dimensional centre,
not that every member has the same symmetry type.
\end{proof}

\begin{theorem}[Regularity neighbourhood]
\label{thm:main}
For every $u_0\in\mathcal C$,
\begin{align}
 &\sup_{N\geq1}\int_0^\infty R_N(t)^{2/3}\,dt<22000.
 \label{eq:action}
\end{align}
Consequently the corresponding periodic Navier--Stokes solution is global
and smooth.
\end{theorem}

The perturbation may have infinitely many Fourier modes, so Theorem
\ref{thm:main} covers an infinite-dimensional neighbourhood of the prescribed
centre rather than only a finite set of Fourier coefficients.

The proof has four dependencies. First, the signed path in
Equation~\eqref{eq:path} supplies an explicit approximate solution and a
complete residual. Second, Equations~\eqref{eq:dini0}--\eqref{eq:dini}
convert that residual into a scalar comparison uniform in the Galerkin cutoff.
Third, the exact slab induction reaches the strict terminal Wiener bound
\eqref{eq:terminal}. Fourth, the exceptional-cutoff, continuation, and
compactness arguments convert the finite-time comparison into the all-time
action bound and a smooth PDE solution. Each
dependency is proved below and
expanded in the corresponding appendix.

The residual measures spectral spread through three quadratic moments.
Define $E_N=\|u_N\|_2^2$, $X_N=\langle Au_N,u_N\rangle$, and
$Z_N=\|Au_N\|_2^2$, so that $E_N$ is twice the kinetic energy and $X_N$
is the squared vorticity norm. If $E_N>0$, minimization over $\lambda$ gives
\begin{align}
 &\mu_N=\frac{X_N}{E_N},\qquad
 R_N=Z_N-\frac{X_N^2}{E_N}.
 \label{eq:Rmoments}
\end{align}
If $\rho_k=|\widehat u_N(k)|^2/E_N$, then
\begin{align}
 &R_N=E_N\operatorname{Var}_{\rho}(|k|^2).
 \label{eq:spectralvariance}
\end{align}
The weights $\rho_k$ and the quotient $\mu_N=X_N/E_N$ are defined only when
$E_N>0$, whereas $R_N$ extends continuously to the zero field by the
convention above. Thus $R_N$ measures the energetic spread of the field across Laplacian
eigenvalues. It vanishes exactly when the active energy lies on one spectral
sphere. This interpretation is used in the exceptional-cutoff estimate and
in the physical diagnostics of Section~\ref{sec:physics}.
The minimization in \eqref{eq:Rmoments}, its probabilistic rewriting in
\eqref{eq:spectralvariance}, and the finite-spectrum variance bound used for
$N=1$ are proved in Appendix~\ref{app:observable}.

\subsection{Centre-independent finite-path acceptance}
\label{sec:acceptance}

The three verified centres below are instances of one analytic acceptance
principle. The principle separates the universal comparison argument from
the centre-specific Fourier calculations. Let $\Theta$ denote a parameter set and let
$u_{0,\theta}$ be smooth, real, mean-zero, solenoidal data for
$\theta\in\Theta$.

\begin{theorem}[Finite-reference-path acceptance]
\label{thm:acceptance}
Fix $\nu>0$ and $T>0$. Suppose a real, mean-zero, solenoidal finite-Fourier
path $v_\theta\in C^1([0,T])$, supported in $|k|_\infty\leq N_*$, has full residual
\begin{align}
 &r_\theta=\partial_tv_\theta+\nu Av_\theta+B(v_\theta,v_\theta)
 \label{eq:detail013}
\end{align}
and uniform nonnegative, bounded, piecewise-continuous bounds
\begin{align}
 &\|v_\theta\|_{\A^0}\leq A_0,\qquad
 \|v_\theta\|_{\A^1}\leq A_1,\qquad
 \|r_\theta\|_{\A^0}\leq d.
 \label{eq:detail014}
\end{align}
Assume, for every $N\geq N_*$, cutoff compatibility
$P_Nv_\theta=v_\theta$, projected residual $P_Nr_\theta$, and initial
mismatch at most $\varepsilon$ in $\A^0$. If a nonnegative absolutely
continuous function $y$ satisfies $y(0)\geq\varepsilon$ and
\begin{align}
 &y'\geq d+\left(A_1-\frac\nu2+\frac{A_0^2}{2\nu}\right)y
       +\frac{A_0}{\nu}y^2+\frac1{2\nu}y^3
 \quad\text{a.e. on }[0,T],
 \label{eq:detail015}
\end{align}
and if $A_0(T)+y(T)\leq\eta<\nu$, then
\begin{align}
 &\|u_N(t)-v_\theta(t)\|_{\A^0}\leq y(t),\qquad 0\leq t\leq T.
 \label{eq:detail016}
\end{align}
If, in addition, the initial $H^1$ norms are uniformly bounded,
$\int_0^T(A_0+y)^2\dd t<\infty$, and every cutoff $N<N_*$ is covered by a
separate global energy argument, then every $u_{0,\theta}$ generates a unique
global smooth periodic solution.
\end{theorem}

\begin{proof}
For $w_N=u_N-v_\theta$, cutoff compatibility gives the exact error equation
\begin{align}
 &\partial_tw_N+\nu Aw_N+P_N\{B(v_\theta,w_N)+B(w_N,v_\theta)
 +B(w_N,w_N)\}=-P_Nr_\theta.
 \label{eq:detail017}
\end{align}
The Wiener convolution estimate of Appendix~\ref{app:foundations} and the
optimization in Appendix~\ref{app:comparison} give precisely the displayed
scalar inequality. Appendix~\ref{app:diniexpansion} treats complex coefficient
norms, including their zeros, before this scalar comparison. Comparison yields
the error bound. The bridge estimate in
Appendix~\ref{app:continuation} supplies uniform $H^1$ and integrated $H^2$
control. Terminal smallness $\|u_N(T)\|_{\A^0}\leq\eta<\nu$ activates the
Wiener barrier, and Appendix~\ref{app:pdepassage} gives compactness,
smoothness, and uniqueness. Lower cutoffs are finite-dimensional global
energy solutions and do not affect the continuum limit.
\end{proof}

An action estimate is an additional instance-specific conclusion. It requires
explicit bridge, tail, and exceptional-cutoff bounds and does not follow from
path algebra alone. Appendix~\ref{app:continuation} gives the generic bridge
constant and the cyclic-shear instance supplies the all-cutoff action bound in
Theorem~\ref{thm:main}.

\subsection{The signed degree-three comparison path}
\label{sec:signedpath}

Three amplitude orders produce a reference path whose complete residual can
be bounded explicitly. The construction first identifies its spatial fields,
then evaluates the heat convolutions and collects the remaining nonlinear
terms. Define the bilinear operator and its first spatial outputs by
\begin{align}
 &B(f,g)=\Pp[(f\cdot\nabla)g],
 \qquad V=B(U,U),
 \qquad C=B(U,V)+B(V,U),
 \qquad H=C+\tfrac12U.
 \label{eq:detail018}
\end{align}
Exact Fourier algebra gives
\begin{align}
 &AU=U,\qquad AV=2V,\qquad
 V=(\sin z\cos y,\sin x\cos z,\sin y\cos x).
 \label{eq:detail019}
\end{align}
The second identity is
\begin{align}
 &(V\cdot\nabla)U=\nabla(-\cos x\cos y\cos z),
 \qquad B(V,U)=0.
 \label{eq:detail020}
\end{align}
The Leray projection therefore removes the gradient before the two cubic
outputs are combined, rather than cancelling two nonzero projected terms.

Substituting Duhamel's formula into itself through amplitude degree three
gives the signed path. The linear term is
$ae^{-t}U$. Since the quadratic output $V$ has eigenvalue two, its resonant
time integral is
\begin{align}
 &-a^2\int_0^t e^{-2(t-s)}e^{-2s}\dd s\,V
 =-a^2te^{-2t}V.
 \label{eq:detail021}
\end{align}
The next forcing is $a^3te^{-3t}C$. The field $C$ splits into its
eigenvalue-one and eigenvalue-five pieces as $C=H-U/2$. Consequently the two
scalar convolutions are
\begin{align}
 K_1(t)&=\int_0^te^{-(t-s)}s e^{-3s}\dd s,
 &K_5(t)&=\int_0^te^{-5(t-s)}s e^{-3s}\dd s,
 \label{eq:detail022}
\end{align}
which evaluate to the formulas below. This derivation keeps the confluent
case together with its numerator. No spurious heat pole is introduced.

Define the confluent heat kernels
\begin{align}
 &K_1(t)=\frac{e^{-t}-(1+2t)e^{-3t}}4,
 \qquad
 K_5(t)=\frac{(2t-1)e^{-3t}+e^{-5t}}4.
 \label{eq:detail023}
\end{align}
They satisfy $K_\ell'+\ell K_\ell=te^{-3t}$, vanish at zero, and are
nonnegative. For
\begin{align}
 &\alpha=ae^{-t},\quad \beta=a^2te^{-2t},\quad
 \gamma=a^3K_1(t),\quad \delta=a^3K_5(t),
 \label{eq:detail024}
\end{align}
set
\begin{align}
 &v_a(t)=(\alpha-\gamma/2)U-\beta V+\delta H.
 \label{eq:path}
\end{align}

Let
\begin{align}
 &D=B(V,V),\quad E=B(U,H)+B(H,U),\quad
 F=B(V,H)+B(H,V),\quad G=B(H,H).
 \label{eq:detail026}
\end{align}
Direct substitution into the unprojected equation gives the complete
residual
\begin{align}
 \mathcal R_3={}&(-\alpha\gamma+\gamma^2/4)V
 +(\alpha\delta-\gamma\delta/2)E+\beta^2D \notag\\
 &+(\beta\gamma/2)C-\beta\delta F+\delta^2G.
 \label{eq:residual}
\end{align}
No nonlinear output of \eqref{eq:residual} is discarded before the norm is
estimated.
The mode-by-mode Fourier reconstruction and every radical majorization used
in this step are detailed in Appendix~\ref{app:algebra}.

To see the bookkeeping explicitly, write $c=\alpha-\gamma/2$. The linear
defects of $cU$, $-\beta V$, and $\delta H$ cancel the prescribed degree-one,
degree-two, and degree-three Duhamel forcing because
\begin{align}
 &\alpha'+\alpha=0,
 \quad \beta'+2\beta=\alpha^2,
 \quad \gamma'+\gamma=a^3te^{-3t},
 \quad \delta'+5\delta=a^3te^{-3t}.
 \label{eq:detail028}
\end{align}
Expanding $B(cU-\beta V+\delta H,cU-\beta V+\delta H)$, using
$B(V,U)=0$, and collecting the uncancelled degrees gives
\eqref{eq:residual}. The coefficients can also be checked from
\begin{align}
 &c^2-\alpha^2=-\alpha\gamma+\gamma^2/4,
 \quad \alpha\beta-c\beta=\beta\gamma/2,
 \quad c\delta=\alpha\delta-\gamma\delta/2.
 \label{eq:detail029}
\end{align}

The exact spatial norms entering the comparison are
\begin{align}
 &\|V\|_{\A^0}=\|C\|_{\A^0}=3,\qquad
 \|D\|_{\A^0}=\sqrt3,\qquad \|H\|_{\A^0}=\frac32,
 \label{eq:detail030}
\end{align}
\begin{align}
 &\|E\|_{\A^0}=\frac{3+\sqrt{21}}2<4,\qquad
 \|G\|_{\A^0}=\frac{2\sqrt{21}}7<\frac43,
 \label{eq:detail031}
\end{align}
\begin{align}
 &\|F\|_{\A^0}=
 \frac{9\sqrt{22}}{44}+\frac{9\sqrt5}{20}+\frac{\sqrt{29}}4
 <\frac72.
 \label{eq:detail032}
\end{align}
These norms give the following rational majorants:
\begin{align}
 A_0&=3\alpha+3\beta+\tfrac32(\gamma+\delta),\label{eq:A0}\\
 A_1&=3\alpha+\tfrac92\beta+\tfrac32\gamma+
       \tfrac{27}{8}\delta,\label{eq:A1}\\
 r&=3(\alpha\gamma+\gamma^2/4)
 +4(\alpha\delta+\gamma\delta/2)+\tfrac74\beta^2\notag\\
 &\hspace{8mm}+\tfrac32\beta\gamma+
 \tfrac72\beta\delta+\tfrac43\delta^2.
 \label{eq:r}
\end{align}
These positive majorants are nondecreasing in $a$. An exact comparison evaluated at
$a=7/10$ therefore covers the full amplitude interval, without assuming that
the exact solutions themselves are monotone in amplitude.

\subsection{Uniform a posteriori comparison}
\label{sec:comparison}

The residual bound controls the error through a scalar differential
inequality independent of the cutoff. The following derivation reduces the
Wiener error to that inequality, and Section~\ref{sec:slabs} specifies the
finite tests that verify an affine supersolution on every time slab.

For $N\geq2$, the support of $v_a$ lies inside the cutoff, so
$P_Nv_a=v_a$ and its Galerkin residual is $P_N\mathcal R_3$. Projection does
not increase the Wiener norm. Let $e_N=u_N-v_a$ and
$Y_j=\|e_N\|_{\A^j}$, $Y=Y_0$. Coefficientwise differentiation and the
Wiener product estimate give the upper-Dini inequality
\begin{align}
 &D^+Y+Y_2\leq A_0Y_1+A_1Y+YY_1+r.
 \label{eq:dini0}
\end{align}
Using $Y_1\leq\sqrt{YY_2}$, $Y_2\geq Y$, and Young's inequality yields
\begin{align}
 &D^+Y\leq
 r+(A_1-\tfrac12+A_0^2/2)Y+A_0Y^2+\tfrac12Y^3.
 \label{eq:dini}
\end{align}

From \eqref{eq:fourierB}, the Leray
contraction and discrete convolution give
\begin{align}
 &\|B(f,g)\|_{\A^0}\leq\|f\|_{\A^0}\|g\|_{\A^1}.
 \label{eq:detail036}
\end{align}
Applying this to
\begin{align}
 &e_t+Ae+P_N\{B(v_a,e)+B(e,v_a)+B(e,e)\}=-P_N\mathcal R_3
 \label{eq:detail037}
\end{align}
gives \eqref{eq:dini0}. Counting-measure Cauchy--Schwarz gives
\begin{align}
 &Y_1\leq
 \left(\sum_k|\widehat e(k)|\right)^{1/2}
 \left(\sum_k|k|^2|\widehat e(k)|\right)^{1/2}
 =\sqrt{YY_2}.
 \label{eq:detail038}
\end{align}
For $Y>0$, put $z=Y_2/Y\geq1$, $w=\sqrt z\geq1$, and
$b=A_0+Y$. Then \eqref{eq:dini0} implies
\begin{align}
 &D^+Y\leq r+A_1Y+Y(bw-w^2).
 \label{eq:detail039}
\end{align}
The elementary constrained maximum satisfies
\begin{align}
 &\sup_{w\geq1}(bw-w^2)\leq\frac{b^2-1}{2}.
 \label{eq:detail040}
\end{align}
Indeed, for $b\leq2$ the maximum is $b-1$ at $w=1$, and
$(b^2-1)/2-(b-1)=(b-1)^2/2$. For $b\geq2$ the maximum is $b^2/4$ at
$w=b/2$, which is at most $(b^2-1)/2$. Expanding
$(b^2-1)Y/2$ gives exactly \eqref{eq:dini}. The case $Y=0$ follows by
the upper-Dini limit. No future norm of the unknown solution is used.
The initial uncertainty obeys
\begin{align}
 &Y(0)=\|P_Nh\|_{\A^0}\leq\|h\|_{\A^2}\leq1/100
 \label{eq:detail041}
\end{align}
uniformly in $N$. Thus Fourier tails of $h$ are included analytically.

\begin{lemma}[Piecewise-affine exact comparison]
Let $0=t_0<t_1<\ldots<t_m=T_*$ and let $y$ be continuous and affine on each
$[t_i,t_{i+1}]$, with $y(0)\geq1/100$ and $y\geq0$. Suppose exact upper
bounds $A_{0,i},A_{1,i},r_i$ dominate \eqref{eq:A0}--\eqref{eq:r} on the
whole slab, and
\begin{align}
 &\frac{y(t_{i+1})-y(t_i)}{t_{i+1}-t_i}
 \geq \max_{q\in\{t_i,t_{i+1}\}}
 \left[r_i+(A_{1,i}-\tfrac12+A_{0,i}^2/2)y(q)
 +A_{0,i}y(q)^2+\tfrac12y(q)^3\right].
 \label{eq:detail042}
\end{align}
Then $Y(t)\leq y(t)$ on $[0,T_*]$.
\end{lemma}

\begin{proof}
The scalar polynomial in $y$ has second derivative
$2A_{0,i}+3y\geq0$, so its maximum along the affine segment occurs at an
endpoint. Equation \eqref{eq:dini} and the standard first-contact comparison
argument for upper-Dini derivatives give $Y\leq y$ on the first slab.
Continuity and induction give the result on all slabs.
The positive-part argument that handles tangential contact is written out
in Appendix~\ref{app:comparison}.
\end{proof}

The exact comparison has 854 slabs, terminal time
\begin{align}
 &T_*=427/256<27/16,
 \label{eq:detail043}
\end{align}
and verifies using exact fraction arithmetic
\begin{align}
 &\sup_{0\leq t\leq T_*}\|u_N(t)\|_{\A^0}<17/8,
 \qquad
 \|u_N(T_*)\|_{\A^0}<99/100
 \label{eq:terminal}
\end{align}
for every $N\geq2$ and every $u_0\in\mathcal C$.
Appendix~\ref{app:slabalgorithm} specifies the rounding, exponential bounds,
trial slopes, and step-halving rule that generate these comparison values.
It also gives the corresponding parameters for the other three applications.

Alternating Taylor bounds enclose the exponentials independently of the
comparison construction. Exact rational operations perform all subsequent
comparisons. This separation provides an independent check of the exponential
enclosures.

\subsubsection{Exact slab construction and acceptance rule}
\label{sec:slabs}

For a slab $I_i=[t_i,t_{i+1}]$, the construction uses eight exact
strings:
\begin{align}
 &t_i,\quad t_{i+1},\quad A_{0,i},\quad A_{1,i},\quad r_i,\quad
 y_i,\quad y_{i+1},\quad \sigma_i.
 \label{eq:detail045}
\end{align}
The last value is the minimum endpoint slack
\begin{align}
 &\sigma_i=\min_{q\in\{y_i,y_{i+1}\}}
 \left\{
 \frac{y_{i+1}-y_i}{t_{i+1}-t_i}
 -r_i-\left(A_{1,i}-\frac12+\frac{A_{0,i}^2}{2}\right)q
 -A_{0,i}q^2-\frac12q^3
 \right\}.
 \label{eq:slack}
\end{align}
Acceptance requires $\sigma_i\geq0$, exact compatibility of adjacent time
and error endpoints, $y_0=1/100$, and the terminal gates in
\eqref{eq:terminal}. A whole-slab upper bound for $t^pe^{-qt}$ is evaluated
at both endpoints and at the unique critical point $p/q$ when it lies in the
slab. The calculation retains the signed numerator of each $K_\ell$ until it
forms the enclosing interval.

The primary and independent implementations deliberately use different
exponential enclosures. The primary implementation bounds $e^x$ by a positive
Taylor sum with a geometric tail and then reciprocates. The independent
implementation encloses
$e^{-x}$ between alternating Taylor sums of degrees 128 and 129. After this
step, every operation in \eqref{eq:slack} is exact integer arithmetic on
fractions.

\subsection{Exceptional cutoff and terminal continuation}
\label{sec:continuation}

The finite-time comparison closes the all-time action estimate once the
exceptional cutoff and the terminal tail are controlled. The estimates below
treat those two contributions, Section~\ref{sec:limit} passes to the PDE,
and Section~\ref{sec:scaling} restores viscosity and domain size.

The approximant \eqref{eq:path} occupies the cube $|k|_\infty\leq2$.
Projecting it to $N=1$ would not commute with the quadratic residual, so the
proof treats that cutoff directly. The higher-cutoff comparison therefore
remains unchanged.

For a finite spectrum with $1\leq|k|^2\leq\Lambda_N=3N^2$, the probability
weights $\rho_k=|\widehat u_k|^2/E_N$ give
\begin{align}
 &R_N=E_N\operatorname{Var}_{\rho}(|k|^2)
 \leq E_N\frac{(\Lambda_N-1)^2}{4}.
 \label{eq:detail047}
\end{align}
For $N=1$, $R_1\leq E_1$. The energy identity and the
modewise inequality $X_1\geq E_1$ give
$E_1'=-2X_1\leq-2E_1$. Since the stated data class satisfies $E_1(0)<1$,
Gronwall's inequality yields $E_1(t)<e^{-2t}$. Hence
\begin{align}
 &\int_0^\infty R_1(t)^{2/3}\,dt<\frac34.
 \label{eq:lowcutoff}
\end{align}

For $N\geq2$, put
$X_N=\|\nabla u_N\|_2^2$ and $Z_N=\|Au_N\|_2^2$. The Wiener estimate gives
\begin{align}
 &X_N'+Z_N\leq\|u_N\|_{\A^0}^2X_N,
 \qquad X_N\leq Z_N.
 \label{eq:enstrophy}
\end{align}
Appendix~\ref{app:continuation} derives \eqref{eq:enstrophy} directly from
the Galerkin equation and expands the bridge and terminal weighted-Hölder
calculations without suppressing intermediate factors.
The initial enstrophy is below one. From \eqref{eq:terminal},
\begin{align}
 &X_N(T_*)<e^8<3000,\qquad
 \int_0^{T_*}Z_N(t)\,dt<3000.
 \label{eq:detail050}
\end{align}
Since $R_N\leq Z_N$, H\"older's inequality bounds the bridge action by $418$.

The constants in this step are fully explicit. Indeed,
\begin{align}
 &\|\nabla u_0\|_2
 \leq a\|\nabla U\|_2+\|h\|_{\dot H^1}
 \leq\frac7{10}\sqrt{\frac32}+\frac1{100}<1.
 \label{eq:detail051}
\end{align}
On the bridge, $\|u_N\|_{\A^0}<17/8$ and
$T_*<27/16$, so
\begin{align}
 &\int_0^{T_*}\|u_N\|_{\A^0}^2\dd t<8.
 \label{eq:detail052}
\end{align}
Multiplying \eqref{eq:enstrophy} by the corresponding integrating factor
gives both $X_N(T_*)<e^8$ and
$\int_0^{T_*}Z_N<e^8$. Since $e^8<3000$ and
$3000^{2/3}<209$,
\begin{align}
 &\int_0^{T_*}R_N^{2/3}\dd t
 \leq T_*^{1/3}\left(\int_0^{T_*}Z_N\dd t\right)^{2/3}
 <2\cdot209=418.
 \label{eq:detail053}
\end{align}

After $T_*$, the strict Wiener barrier $\|u_N\|_{\A^0}\leq c=99/100$ is
preserved. Equation \eqref{eq:enstrophy} improves to
\begin{align}
 &X_N'+dZ_N\leq0,\qquad d=1-c^2=199/10000.
 \label{eq:detail054}
\end{align}
Multiplication by $e^{d(t-T_*)/2}$ and weighted H\"older inequality give a tail action
below $418/d$. Exact fraction arithmetic verifies
\begin{align}
 &418+\frac{418}{d}<22000.
 \label{eq:detail055}
\end{align}
Together with \eqref{eq:lowcutoff}, this proves \eqref{eq:action}.

For reference, the weighted calculation is
\begin{align}
 &\int_{T_*}^{\infty}e^{d(t-T_*)/2}Z_N(t)\dd t
 \leq\frac{2X_N(T_*)}{d},
 \label{eq:detail056}
\end{align}
and H\"older's inequality with exponents $3/2$ and $3$ yields
\begin{align}
 &\int_{T_*}^{\infty}Z_N^{2/3}\dd t
 \leq\frac{2^{2/3}X_N(T_*)^{2/3}}d<\frac{418}{d}.
 \label{eq:detail057}
\end{align}

\subsubsection{Passage from Galerkin trajectories to the PDE}
\label{sec:limit}

The cutoff-independent bounds yield a smooth solution of the PDE through
compactness and continuation. On every finite interval they give
\begin{align}
 &\sup_N\|u_N\|_{L^\infty_tH^1_x}<\infty,
 \qquad
 \sup_N\|u_N\|_{L^2_tH^2_x}<\infty.
 \label{eq:detail058}
\end{align}
Sobolev embedding and the product estimate
$\|(u_N\cdot\nabla)u_N\|_2
 \leq C\|u_N\|_{H^2}\|u_N\|_{H^1}$
bound $\partial_tu_N$ in $L^2_tL^2_x$. Aubin--Lions compactness gives,
after passing to a subsequence,
\begin{align}
 &u_N\to u\quad\hbox{strongly in }L^2(0,T,H^1),
 \label{eq:detail059}
\end{align}
with weak convergence in the two bounded spaces. This is sufficient to pass
the quadratic term and obtain a strong periodic solution with initial datum
$u_0$. Lower semicontinuity transfers the $H^1$ and integrated $H^2$ bounds.
At $T_*$, coefficientwise convergence and Fatou transfer
$\|u(T_*)\|_{\A^0}\leq99/100$. The Wiener barrier propagates this strict
smallness for all later times. The all-order Fourier estimate in
Appendix~\ref{app:smoothnessexpansion} propagates every spatial derivative of
each fixed smooth datum and gives a global smooth solution. Weak--strong uniqueness removes
subsequence dependence.
Appendix~\ref{app:pdepassage} gives the exponent bookkeeping for the nonlinear
term, the precise compactness spaces, the passage to the limit, and the
regularity continuation argument.

\subsubsection{Scaling in viscosity and torus size}
\label{sec:scaling}

The normalisation $\nu=1$ loses no dimensional information. On the same
$2\pi$ torus, if $v$ is a unit-viscosity solution in the proved class, then
\begin{align}
 &u(t,x)=\nu v(\nu t,x)
 \label{eq:detail060}
\end{align}
solves the equation at viscosity $\nu$. Hence the corresponding initial class is
$u_0=\nu(aU+h)$ with the same dimensionless bounds, and
\begin{align}
 &\int_0^\infty R_u(t)^{2/3}\dd t
 =\nu^{1/3}\int_0^\infty R_v(\tau)^{2/3}\dd\tau.
 \label{eq:detail061}
\end{align}
On a torus of side $2\pi L$, the scaling
$u(t,x)=(\nu/L)v(\nu t/L^2,x/L)$ gives the corresponding action factor
$\nu^{1/3}/L^2$. These are rescalings of the same proved dimensionless
family, not coverage of arbitrary viscosity-independent data.
Both changes of variables, including the scaling of $R$, are calculated in
Appendix~\ref{app:pdepassage}.

\begin{proof}
For Theorem~\ref{thm:main}, the exceptional cutoff $N=1$ satisfies
\eqref{eq:lowcutoff}. For every $N\geq2$, the exact comparison lemma and its
854 accepted slabs give \eqref{eq:terminal}. The bridge estimate then bounds
the action on $[0,T_*]$ by $418$, while the terminal Wiener barrier and the
weighted estimate bound the remaining action by $418/d$, where
$d=199/10000$. Thus
\begin{align}
 &\int_0^\infty R_N(t)^{2/3}\dd t
 <418+\frac{418}{d}<22000
 \label{eq:detail062}
\end{align}
uniformly for $N\geq2$. Together with \eqref{eq:lowcutoff}, this proves
\eqref{eq:action} for every cutoff. Appendix~\ref{app:pdepassage} removes the
cutoff, propagates the terminal Wiener bound, and proves smoothness and
uniqueness of the resulting periodic solution. These steps establish both
conclusions of the theorem.
\end{proof}

\subsection{Second centre: the equal-coefficient ABC family}
\label{sec:abc}

The ABC application tests the comparison principle on a centre with zero
projected nonlinearity. Its heat path is exact, so the scalar comparison
controls only the perturbation. Define the equal-coefficient ABC field
\begin{align}
 &U_{\rm ABC}=(\sin z+\cos y,\ \sin x+\cos z,\ \sin y+\cos x).
 \label{eq:detail063}
\end{align}

\begin{theorem}[Equal-ABC neighbourhood]
\label{thm:abc}
For every viscosity $\nu>0$, every $0\leq a\leq17\nu/25$, and every smooth,
real, mean-zero, solenoidal field $h$ satisfying
\begin{align}
 &\|h\|_{\A^2}\leq\frac{\nu}{100},
 \label{eq:detail064}
\end{align}
the periodic datum $u_0=aU_{\rm ABC}+h$ generates a unique global smooth
solution. At unit viscosity, every cubical Galerkin cutoff $N\geq1$ satisfies
\begin{align}
 &T_{\rm ABC}=\frac{1195}{512},\qquad
 \|u_N(T_{\rm ABC})\|_{\A^0}
 \leq
 \frac{598396355886649288745723}{604462909807314587353088}
 <\frac{99}{100}.
 \label{eq:detail065}
\end{align}
For viscosity $\nu$, the corresponding terminal time is
$T_{\rm ABC}/\nu$.
\end{theorem}

\begin{proof}
Appendix~\ref{app:abcalgebra} verifies
$\nabla\times U_{\rm ABC}=U_{\rm ABC}$,
$AU_{\rm ABC}=U_{\rm ABC}$,
$B(U_{\rm ABC},U_{\rm ABC})=0$, and
$\|U_{\rm ABC}\|_{\A^0}=\|U_{\rm ABC}\|_{\A^1}=3\sqrt2$.
Consequently $v(t)=ae^{-t}U_{\rm ABC}$ is an exact unit-viscosity path whose
projected residual vanishes for every $N\geq1$. With
$A(t)=(297/70)ae^{-t}$, the error comparison is the specialization of
\eqref{eq:dini} with $A_0=A_1=A$ and $r=0$:
\begin{align}
 &D^+Y\leq
 \left(A-\frac12+\frac{A^2}{2}\right)Y
 +AY^2+\frac12Y^3.
 \label{eq:detail066}
\end{align}
Appendix~\ref{app:abccomparison} derives this specialization and the convex
endpoint test. The 2390 exact slabs start from $Y(0)=1/100$ and give the
displayed terminal bound, together with the bridge estimate
\begin{align}
 &\sup_{0\leq t\leq T_{\rm ABC}}\|u_N(t)\|_{\A^0}
 \leq
 \frac{437511399315879405468719}{151115727451828646838272}.
 \label{eq:detail067}
\end{align}
The terminal Wiener barrier and the compactness argument of
Appendix~\ref{app:pdepassage} then yield a unique global smooth solution.
Finally, $u(t,x)=\nu w(\nu t,x)$ maps the unit-viscosity bounds
$a/\nu\leq17/25$ and $\|h\|_{\A^2}/\nu\leq1/100$ to the stated
viscosity-$\nu$ bounds and scales the terminal time by $1/\nu$.
\end{proof}

Theorem~\ref{thm:abc} establishes a second centre within the same comparison
architecture. Its Beltrami geometry and vanishing nonlinear centre residual
provide a structurally distinct test of the centre-independent theorem. The
non-Beltrami family-level extension follows below.

\subsection{Third centre: the nonlinear TG3 family}
\label{sec:tg3}

The TG3 application retains a nonzero quadratic residual even for a heat
reference path. Its single-shell support makes that residual explicit,
allowing the same comparison principle to cover all cutoffs. Define
\begin{align}
 &U_{\rm TG3}=(\sin x\cos y\cos z,\ \sin y\cos z\cos x,
              -2\sin z\cos x\cos y).
 \label{eq:detail068}
\end{align}

\begin{theorem}[TG3 neighbourhood]
\label{thm:tg3}
For every $\nu>0$, every $0\leq a\leq41\nu/40$, and every smooth, real,
mean-zero, solenoidal $h$ with $\|h\|_{\A^2}\leq\nu/40$, the datum
$u_0=aU_{\rm TG3}+h$ generates a unique global smooth periodic solution.
The bounds are uniform over all cubical Galerkin cutoffs $N\geq1$. At unit
viscosity,
\begin{align}
 &\left\|u_N\!\left(\frac{873}{1024}\right)\right\|_{\A^0}
 \leq
 \frac{1196766436682595646613963}{1208925819614629174706176}
 <\frac{99}{100}<1.
 \label{eq:detail069}
\end{align}
For viscosity $\nu$, the terminal time is $873/(1024\nu)$.
\end{theorem}

\begin{proof}
Appendix~\ref{app:tg3algebra} verifies that $U_{\rm TG3}$ is solenoidal, occupies
eight modes on $|k|^2=3$, and satisfies
\begin{align}
 &AU_{\rm TG3}=3U_{\rm TG3},\qquad
 \|U_{\rm TG3}\|_{\A^0}=\sqrt6,\qquad
 \|U_{\rm TG3}\|_{\A^1}=3\sqrt2.
 \label{eq:detail070}
\end{align}
Unlike the ABC centre, its projected quadratic output
$V_{\rm TG3}=B(U_{\rm TG3},U_{\rm TG3})$ is nonzero, with
$\|V_{\rm TG3}\|_{\A^0}=3\sqrt2/4$. The heat path
$v(t)=ae^{-3t}U_{\rm TG3}$ therefore has the full, nonzero residual
$r(t)=a^2e^{-6t}V_{\rm TG3}$. This output is bounded before projection and is
not discarded at $N=1$.

Appendix~\ref{app:tg3comparison} substitutes the exact spatial constants into
Theorem~\ref{thm:acceptance}. Starting from $y(0)=1/40$, 873 exact slabs on
$[0,873/1024]$ establish the displayed terminal inequality and the bridge
bound
\begin{align}
 &\sup_{0\leq t\leq873/1024}\|u_N(t)\|_{\A^0}
 \leq
 \frac{1533837106283467124791861}{604462909807314587353088}.
 \label{eq:detail071}
\end{align}
Theorem~\ref{thm:acceptance} then yields the unit-viscosity conclusion for
every $N\geq1$, and the scaling in Appendix~\ref{app:pdepassage} yields the
stated viscosity-$\nu$ family.
\end{proof}

The TG3 centre is non-Beltrami and has a retained nonzero centre residual.
Together, Theorems~\ref{thm:main}, \ref{thm:abc}, and \ref{thm:tg3} are three
verified instances of Theorem~\ref{thm:acceptance}. The next subsection
strengthens the TG3 instance to a uniform coefficient interval while retaining
its fixed Fourier support.

\subsection{A uniform nonlinear coefficient family}
\label{sec:coefficientfamily}

One common comparison also covers an interval of non-Beltrami centre shapes.
The construction varies two component amplitudes while preserving
solenoidality, then bounds the full residual uniformly over that interval.
For $\theta\in[-1/10,1/10]$, define the fixed-support family
\begin{align}
 &U_\theta=\bigl((1+\theta)\sin x\cos y\cos z,
 (1-\theta)\sin y\cos z\cos x,
 -2\sin z\cos x\cos y\bigr).
 \label{eq:tg3family}
\end{align}
Thus $U_0=U_{\rm TG3}$, while nonzero $\theta$ changes the relative strength
of the first two velocity components without destroying solenoidality.

\begin{theorem}[Uniform TG3 coefficient-family neighbourhood]
\label{thm:tg3box}
For every $\nu>0$, $|\theta|\leq1/10$, $0\leq a\leq\nu$, and every smooth,
real, mean-zero, solenoidal $h$ satisfying
\begin{align}
 &\|h\|_{\A^2}\leq\frac{\nu}{40},
 \label{eq:detail073}
\end{align}
the datum $u_0=aU_\theta+h$ generates a unique global smooth periodic
solution. The conclusion is uniform in $\theta$, $a$, $h$, and every cubical
Galerkin cutoff $N\geq1$. At unit viscosity, the common terminal time is
$T_\Box=31/32$, and
\begin{align}
 &\sup_{\substack{|\theta|\leq1/10,\ 0\leq a\leq1\\
                   \|h\|_{\A^2}\leq1/40}}
 \sup_{N\geq1}\|u_N(T_\Box)\|_{\A^0}
 \leq
 \frac{1196626932805349354526769}{1208925819614629174706176}
 <\frac{99}{100}.
 \label{eq:tg3boxterminal}
\end{align}
For general viscosity, the terminal time is $31/(32\nu)$.
\end{theorem}

\begin{proof}
Appendix~\ref{app:boxalgebra} derives the exact identities
\begin{align}
 &\|U_\theta\|_{\A^0}^2=6+2\theta^2,
 \qquad \|U_\theta\|_{\A^1}^2=18+6\theta^2,
 \label{eq:detail075}
\end{align}
and decomposes the complete quadratic output as
\begin{align}
 &B(U_\theta,U_\theta)=V_0+\theta V_1+\theta^2V_2.
 \label{eq:detail076}
\end{align}
The three grouped fields occupy a common 12-mode union on $|k|^2=8$ and
satisfy
\begin{align}
 &\|V_0\|_{\A^0}=\frac{3\sqrt2}{4},\qquad
 \|V_1\|_{\A^0}=\sqrt2,\qquad
 \|V_2\|_{\A^0}=\frac{\sqrt2}{4}.
 \label{eq:detail077}
\end{align}
Bernstein convex-hull bounds on the full interval $|\theta|\leq1/10$ give
\begin{align}
 &\|U_\theta\|_{\A^0}<\frac{123}{50},\qquad
 \|U_\theta\|_{\A^1}<\frac{17}{4},
 \label{eq:detail078}
\end{align}
and
\begin{align}
 &\|B(U_\theta,U_\theta)\|_{\A^0}
 \leq\frac{99}{70}\frac{341}{400}.
 \label{eq:detail079}
\end{align}
These are interval bounds, not evaluations of sampled centres. For the heat
path $v(t)=ae^{-3t}U_\theta$, the last inequality bounds the complete residual
$a^2e^{-6t}B(U_\theta,U_\theta)$ before projection. Every path mode lies in
cube one, and projection can only decrease the residual Wiener norm, so all
$N\geq1$ are covered without an exceptional cutoff. Substitution into
Theorem~\ref{thm:acceptance}, followed by the 992 common exact slabs derived
in Appendix~\ref{app:boxcomparison}, gives \eqref{eq:tg3boxterminal} and the
uniform bridge bound. Terminal Wiener continuation, Galerkin compactness, and
viscosity scaling then give the stated conclusion.
\end{proof}

Theorem~\ref{thm:tg3box} is uniform over both a positive-dimensional centre
parameter and an infinite-dimensional perturbation ball. It therefore extends
the same comparison argument continuously across non-Beltrami centre shapes
while retaining a fixed Fourier support.
Appendix~\ref{app:bernsteindetail} derives the Bernstein basis and the exact
coefficient conversion used in this theorem. Appendix~\ref{app:slabalgorithm}
then converts these uniform bounds into the common slab comparison.

\section{Computational methodology}
\label{sec:computationalmethods}

Exact arithmetic resolves the finite comparison inequalities that connect the
analytic estimates to each explicit centre family. The implementation first
constructs the comparison bounds and tests their defining conditions with
deliberately altered inputs. Separate calculations treat the cyclic-shear,
equal-ABC, TG3, and coefficient-family cases.
Section~\ref{sec:exactmethods} describes the exact checks, and
Section~\ref{sec:pinomethod} specifies the neural experiment. The executable
construction is derived in Appendix~\ref{app:slabalgorithm}. Appendix~\ref{app:numerics}
constructs the Fourier initial fields, and Appendix~\ref{app:timestepping}
derives the numerical update and balance diagnostics. Appendix~\ref{app:pino}
specifies the neural experiment, with each discrete loss derived in
Appendix~\ref{app:neuralobjective}.
Appendix~\ref{app:discreteexpansion} expands discrete Fourier normalisation,
aliasing and the Runge--Kutta stages. Appendix~\ref{app:neuralexpansion}
specifies network contractions and loss factors, and
Appendix~\ref{app:statisticsexpansion} derives the quadrature and statistical
operations used to summarise the trajectories.

\subsection{Independent exact-arithmetic verification}
\label{sec:exactmethods}

The primary implementation constructs exact slab bounds and an affine
supersolution, while an independently written implementation reconstructs
$U,V,C,H,E,F,G$ through a separate sparse Fourier convolution,
checks solenoidality and reality mode by mode, recomputes all spatial norms,
uses the independent exponential enclosure described above, and evaluates
all 854 instances of \eqref{eq:slack}.

Four negative controls independently alter the residual majorant, amplitude
interval, terminal upper bound and cutoff range. Specifically, they consist of
\begin{enumerate}
\item setting one residual majorant to zero,
\item enlarging the amplitude interval from $[0,7/10]$ to $[0,4/5]$,
\item understating the terminal upper bound as $1/2$,
\item falsely claiming cutoff coverage beginning at $N=1$.
\end{enumerate}
The four deliberately altered inputs fail the corresponding checks for
residual domination, continuum parameter coverage, terminal smallness and
cutoff compatibility.

Independent evaluation verifies all 854 slab inequalities over the stated
amplitude interval and perturbation radius, together with the
exceptional-cutoff and all-cutoff action bounds. A further 13 scope checks
and eight integration checks test consistency between the analytic
hypotheses, the finite comparison and the continuum conclusion.

For Theorem~\ref{thm:abc}, the independent implementation reconstructs the
ABC Fourier modes, curl and Laplacian identities, zero projected
nonlinearity, and exact Wiener norms. It independently checks all 2390 slab
inequalities and the terminal and bridge bounds. Enlarging the amplitude
interval, assigning a negative residual, understating the terminal value, or
declaring an incompatible cutoff condition causes rejection.

For Theorem~\ref{thm:tg3}, the independent implementation reconstructs all
eight TG3 modes, verifies solenoidality and the eigenvalue-three shell,
retains the nonzero projected quadratic output, and recomputes its exact
Wiener norm. It checks all 873 slab inequalities, the terminal and bridge
bounds, and all-cutoff compatibility. Enlarging either the amplitude or
perturbation radius, deleting the nonlinear residual, understating the
terminal value, or changing the cutoff scope causes rejection.

For Theorem~\ref{thm:tg3box}, a separate interval implementation reconstructs
$U_\theta$, the three residual coefficients $V_0,V_1,V_2$, both exact
norm-squared polynomials, their Bernstein representations, and the complete
12-mode residual union. It verifies all 992 common slabs over the entire
parameter interval. Six negative controls enlarge the coefficient interval
or perturbation radius, delete the residual, understate its Bernstein
maximum, understate the terminal value, or alter the cutoff scope, each of
which violates the corresponding comparison condition.

\subsection{Physics-informed neural-operator search}
\label{sec:pinomethod}

The neural-operator experiment compares a supervised FNO with a PINO of identical architecture.
Both maps take the initial velocity and a two-component family indicator as
input and return three velocity components at five times
$t\in\{0,0.025,0.05,0.1,0.2\}$. Three spectral layers retain three modes per
coordinate and propagate eight latent channels on a $12^3$ grid, with group
normalisation \cite{GroupNorm} and Gaussian error linear units
\cite{GELU}. Training used
127 fields, model selection used 45 validation fields, and an
out-of-distribution (OOD) set contained 212 fields selected by parameter-space
conditions that do not enter training. Appendix~\ref{app:pino} gives the split,
architecture, loss definitions, and extended diagnostics.

Let $\mathcal G_\vartheta(u_0,f)_j$ denote the velocity predicted by parameters
$\vartheta$ for initial field $u_0$, family indicator $f$, and checkpoint
$t_j$. The FNO minimises the field loss and an initial-condition anchor. The
PINO uses the same terms and adds the weighted objective
\begin{align}
 \mathcal L_{\rm PINO}={}&\mathcal L_{\rm field}+2\mathcal L_{\rm initial}
 +0.002\mathcal L_{\rm NS}+0.05\mathcal L_{\rm div}
 +0.01\mathcal L_{\rm energy}\notag\\
 &+0.05\mathcal L_{\rm dealias}+0.05\mathcal L_{\rm spectral}
 +0.05\mathcal L_{\rm tail}.
 \label{eq:pinoloss}
\end{align}
Here $\mathcal L_{\rm NS}$ measures the cutoff-two projected midpoint
Navier--Stokes residual on the $12^3$ training grid, $\mathcal L_{\rm div}$ measures incompressibility,
$\mathcal L_{\rm energy}$ measures the discrete energy balance, and the final
three terms penalise native-grid cutoff sensitivity, component-averaged
Fourier-weighted field error, and predicted outer-band spectral mass.
Appendix~\ref{app:neuralobjective} derives their averaging factors,
projection conventions and midpoint discretization from
Equation~\eqref{eq:pinoloss}. Adam optimisation used at most 18 epochs,
batches of four, and learning rate $0.0015$~\cite{Adam}. Training-only means
and standard deviations normalised each velocity component.

Gradient search optimised a composite predicted residual-action and spectral-
tail score, followed by ranking of the endpoint pool by the predicted
five-checkpoint action proxy alone. The random and Sobol pools underwent
the same PINO-proxy ranking before
solver evaluation. They are therefore PINO-screened proposal baselines rather
than unguided solver baselines. FNO, PINO, random, and Sobol acquisition each
received 512 candidate evaluations and six solver rechecks in each of the
cyclic-shear and TG3 families. A normalised separation of $0.12$ prevents duplicate selections.
The spectral solver recomputed all 48 selected fields at
$(16^3,0.005)$, $(20^3,0.0025)$, and $(24^3,0.00125)$, where each pair gives
the grid and time step. Training, ranking, and these rechecks used the
five-checkpoint trapezoidal action proxy evaluated at the network output
times. A separate dense quadrature recorded the action integrand at every
classical fourth-order Runge--Kutta (RK4) step for all 48 selected candidates. Appendix~\ref{app:parameters} gives the input maps and seeds. Appendix~\ref{app:networkdetail} specifies the network operations. Equations~\eqref{eq:pinofivepointproxy} and \eqref{eq:densequadraturedetail} distinguish the action quadratures. Neural predictions, floating-point checks, and the
exact regularity theorem remain separate levels of evidence.

\section{Results and discussion}
\label{sec:results}

The proved neighbourhoods occupy a precise position relative to established
regularity criteria. Their analytic observables also define the geometric and
spectral quantities measured in the ensemble calculations and public
turbulence data. Section~\ref{sec:normposition} compares the verified families with explicit
small-data thresholds. Section~\ref{sec:physics} connects their observables
to numerical ensembles, local field geometry and public turbulence samples.

\subsection{Position relative to standard small-data norms}
\label{sec:normposition}

The present theorem defines an explicit computer-assisted regularity class and
can be located quantitatively relative to the critical-space traditions of
Kato \cite{Kato} and Koch--Tataru \cite{KochTataru}. Since
all modes of $U$ have $|k|=1$,
\begin{align}
 &\|aU\|_{\A^0}=\|aU\|_{\A^2}=3a,
 \qquad
 \|aU\|_{\dot H^{1/2}}=a\sqrt{\frac32},
 \qquad
 \|aU\|_\infty=a\sqrt3.
 \label{eq:detail081}
\end{align}
The perturbation hypothesis implies
\begin{align}
 &\|h\|_{\A^0},\ \|h\|_{\dot H^{1/2}},\ \|h\|_\infty
 \leq\|h\|_{\A^2}\leq\frac1{100}.
 \label{eq:detail082}
\end{align}
At the largest amplitude this gives
\begin{align}
 \|u_0\|_{\A^0}&\leq2.11,\notag\\
 \|u_0\|_{\dot H^{1/2}}
 &\leq0.7\sqrt{3/2}+0.01<0.868,\notag\\
 \|u_0\|_{L^3}&\leq\|u_0\|_\infty
 \leq0.7\sqrt3+0.01<1.223.
 \label{eq:normposition}
\end{align}
The six Fourier coefficients of $U$ and each norm comparison in
\eqref{eq:normposition} are expanded in
Appendix~\ref{app:physicalidentities}.

At $a=0.7$, the cyclic-shear centre has $\|aU\|_{\A^0}=2.1>1$,
so the proved neighbourhood extends beyond the displayed unit Wiener
smallness threshold. This comparison concerns that sufficient condition,
not every small-data criterion \cite{LeiLin}. For the critical
$\dot H^{1/2}$ and $L^3$ frameworks \cite{FujitaKato,Kato},
Equation~\eqref{eq:normposition} supplies explicit norm bounds, while the
present theorem supplies family-specific constants and terminal times.
Smooth periodic mean-zero data also belong to the relevant inverse-derivative
space of bounded mean oscillation ($BMO^{-1}$) and
critical Besov spaces, but membership alone does not establish the smallness
required by a critical-space theorem \cite{KochTataru}. The finite-path
criterion provides a separate quantitative verification mechanism.

The centre has no invariant spatial direction, so its three-dimensional
geometry complements directionally structured large-data constructions.
Within the established a posteriori regularity framework
\cite{CCRT,DR,RS,MorosiPizzocchero2012,MorosiPizzocchero2015}, the
contribution is the explicit verification of a continuum of initial data
with computable perturbation radii, terminal times, and bounds uniform in
the Galerkin cutoff.

The equal-ABC result provides a second, structurally distinct centre at unit
viscosity. At its upper endpoint, the reverse triangle inequality and
$\|h\|_{\A^0}\leq\|h\|_{\A^2}$ give
\begin{align}
 &\|aU_{\rm ABC}+h\|_{\A^0}
 \geq 3\sqrt2\,\frac{17}{25}-\frac1{100}>1.
 \label{eq:detail084}
\end{align}
Thus the entire admitted perturbation ball at that endpoint lies outside the
displayed direct condition $\|u_0\|_{\A^0}<1$. The exact ABC norm is derived in
Appendix~\ref{app:algebra}. The non-Beltrami coefficient family then shows that
the shared architecture is not tied to vanishing centre residuals.
Table~\ref{tab:certifiedcentres} summarises the parameter ranges, perturbation
radii, and cutoff coverage of the four verified instances.

\begin{table}[tbp]
\centering
\caption{Verified centre families and admissible parameter ranges at
viscosity $\nu>0$. The perturbation column gives an upper bound for
$\|h\|_{\A^2}/\nu$. The four rows correspond, respectively, to
Theorems~\ref{thm:main}, \ref{thm:abc}, \ref{thm:tg3}, and
\ref{thm:tg3box}. Each theorem states the assumptions on $h$.}
\label{tab:certifiedcentres}
\small
\begin{tabular}{@{}>{\raggedright\arraybackslash}p{0.19\linewidth}
                    >{\raggedright\arraybackslash}p{0.26\linewidth}
                    >{\centering\arraybackslash}p{0.11\linewidth}
                    >{\centering\arraybackslash}p{0.13\linewidth}
                    >{\raggedright\arraybackslash}p{0.16\linewidth}@{}}
\toprule
Centre & Reference path and residual & $a/\nu$ & Radius bound & Cutoffs\\
\midrule
Cyclic shear & Degree-three path, nonzero residual & $[0,7/10]$ & $1/100$ & $N\geq2$, with $N=1$ treated separately\\
Equal ABC & Beltrami heat path, zero residual & $[0,17/25]$ & $1/100$ & All $N\geq1$\\
TG3 & Heat path, nonzero residual retained & $[0,41/40]$ & $1/40$ & All $N\geq1$\\
TG3 coefficient family, $|\theta|\leq1/10$ & Heat path, quadratic residual bounded uniformly in $\theta$ & $[0,1]$ & $1/40$ & All $N\geq1$\\
\bottomrule
\end{tabular}
\end{table}

At the TG3 upper endpoint, the same reverse-triangle argument gives
\begin{align}
 &\left\|\frac{41}{40}U_{\rm TG3}+h\right\|_{\A^0}
 \geq\frac{41}{40}\sqrt6-\frac1{40}>1.
 \label{eq:detail085}
\end{align}
Theorems~\ref{thm:main}, \ref{thm:abc}, and \ref{thm:tg3} provide three
independently verified instances of the common acceptance theorem, including two
nonzero-residual centres. Theorem~\ref{thm:tg3box} further supplies one
positive-dimensional interval of centre shapes with a uniform radius. This
connected multi-centre result is the family-level contribution of the method.

\subsection{Numerical experiment design and physical results}
\label{sec:physics}

Local stretching and global spectral spread describe different aspects of
the same evolution. Sections~\ref{sec:initialgeometry} and
\ref{sec:multigeometry} compare analytic centre fields, while
Sections~\ref{sec:ensemble} and \ref{sec:refinement} locate the numerical
resolution boundary. Section~\ref{sec:pinoresults} evaluates operator-guided
search, and Section~\ref{sec:dynamicgeometry} interprets the selected evolving
fields. Sections~\ref{sec:eigentransients}, \ref{sec:fixedstress}, and
\ref{sec:publicdns} test scaling, discretisation and cross-flow geometry,
respectively. Appendix~\ref{app:tensorexpansion} derives the local tensor
identities, and Appendix~\ref{app:momentexpansion} derives the distinct
global spectral balance.

\subsubsection{Exact geometry of the reference centre}
\label{sec:initialgeometry}

Analytic evaluation of the centre on a $64^3$ grid resolves its initial
geometry without time integration, so the grid affects the rendering rather
than the underlying field. The derivatives below determine the displayed
velocity, vorticity, strain and stretching directly from the centre formula.

For $u=aU$, let $\omega=\nabla\times u$ denote vorticity and
$S=(\nabla u+\nabla u^\top)/2$ the strain tensor, with $|S|=|S|_F$ its
Frobenius norm. Direct differentiation gives
\begin{align}
 \omega&=\nabla\times u=-a(\cos z,\cos x,\cos y),\label{eq:omegaU}\\
 S_{12}&=\frac a2\cos y,\quad
 S_{13}=\frac a2\cos x,\quad
 S_{23}=\frac a2\cos z,\quad S_{ii}=0,\label{eq:strainU}\\
 |S|^2&=\frac{a^2}{2}(\cos^2x+\cos^2y+\cos^2z)
       =\frac12|\omega|^2,\label{eq:strainomega}\\
 \omega\cdot S\omega&=3a^3\cos x\cos y\cos z,\label{eq:stretchU}\\
 (u\cdot\nabla)u&=a^2(\sin z\cos y,\sin x\cos z,\sin y\cos x).
 \label{eq:nonlinearU}
\end{align}
Each derivative, contraction, and trigonometric identity in this display is
expanded in Appendix~\ref{app:physicalidentities}. The grid is used only for
rendering the already analytic fields.
Equations \eqref{eq:omegaU}--\eqref{eq:nonlinearU} are the formulas used to
write the Visualization Toolkit (VTK) arrays and to produce
Figure~\ref{fig:centerpanels}.

\begin{figure}[t]
\centering
\includegraphics[width=0.94\linewidth,height=0.43\textheight,keepaspectratio]{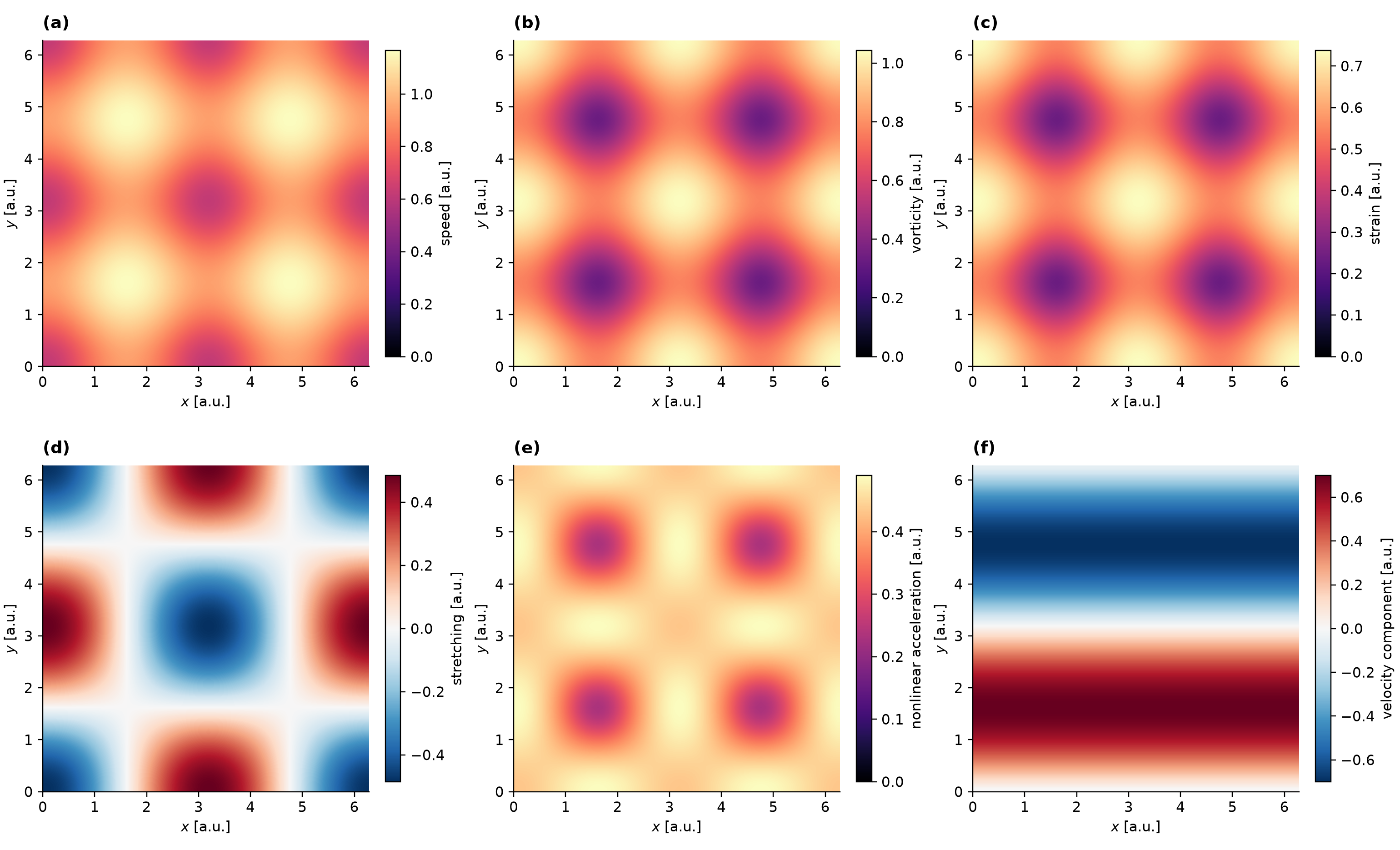}
\caption{Reference-centre slice at $z=21(2\pi/64)$: (a) $|u|$, (b)
$|\omega|$, (c) $|S|$, (d) $\omega\cdot S\omega$, (e)
$|(u\cdot\nabla)u|$, (f) $u_1$. Color limits are absolute and shared only
where the displayed quantity is the same.}
\label{fig:centerpanels}
\end{figure}

Figure~\ref{fig:centerpanels}(a) shows that speed is spatially modulated even
though every component begins on the same Laplacian shell. The matching
geometry in panels (b) and (c) is not a visual coincidence: it is exactly the
factor $|S|=|\omega|/\sqrt2$ in \eqref{eq:strainomega}. Panel (d) uses a
diverging scale centred at zero. Red regions have
$\omega\cdot S\omega>0$, blue regions have
$\omega\cdot S\omega<0$, and the pale interfaces are near-zero stretching.
Equation~\eqref{eq:stretchU} explains the alternating cells and their zero
spatial mean. Panel (e) shows the quadratic-acceleration magnitude. This
centre's acceleration is already divergence free, as established in the first
lemma, so its nonzero magnitude implies a nonzero projected nonlinearity.
Panel (f) is a
component-level validation of the amplitude and phase convention used to
construct the Fourier seed.

\begin{figure}[t]
\centering
\includegraphics[width=0.6\linewidth,height=0.38\textheight,keepaspectratio]{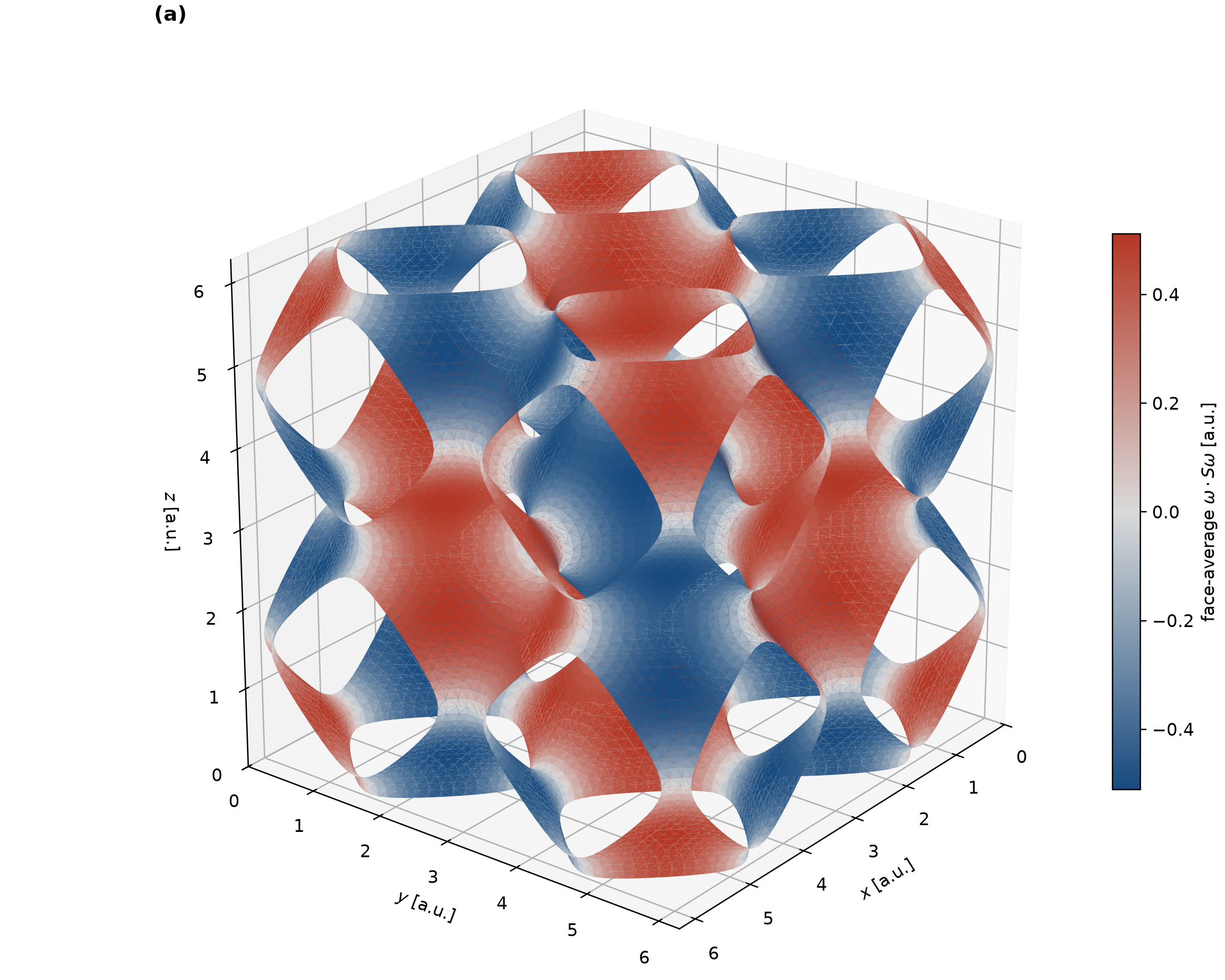}
\caption{(a) $|\omega|=0.9614$ VTK isosurface of the analytic $64^3$
centre field, coloured by face-averaged $\omega\cdot S\omega$, 53,864
triangles.}
\label{fig:center3d}
\end{figure}

Figure~\ref{fig:center3d} adds topology to the planar slices. Red and blue
patches identify the two stretching signs on one fixed-vorticity surface,
so colour encodes stretching rather than vorticity magnitude. Their interlocking arrangement
shows why a scalar magnitude or a positivity hypothesis would discard relevant
phase information. The plot evaluates the exact field at $t=0$ and exposes
its initial stretching geometry. The field begins on one
Laplacian shell with $R(0)=0$, while the nonzero quadratic acceleration in
Figure~\ref{fig:centerpanels}(e) creates eigenvalue-two content. The proof
therefore controls the full signed Fourier residual rather than assuming a
global stretching sign.

\subsubsection{Matched geometry of the three proved centres}
\label{sec:multigeometry}

The geometric comparison evaluates the three theorem centres at their unit-viscosity upper amplitudes
on the same $64^3$ grid and renders them with identical cameras,
periodic streamline integration, deterministic seeding, and dimensionless
colour ranges. The matched settings expose geometric differences without a
viewpoint or normalisation change. Because all panels represent initial fields, the comparison describes
geometry rather than temporal evolution. Vorticity-line colour uses
$\sigma_\omega=(\omega\cdot S\omega)/(|\omega|^2\|S\|_F)$, with zero
assigned when the denominator vanishes. Directional alignment is
$|\widehat\omega\cdot e_{\max}|$, where $\widehat\omega=\omega/|\omega|$
and $e_{\max}$ is a unit eigenvector of the largest strain eigenvalue.
Appendix~\ref{app:physicalidentities} gives the bounds and the convention
for repeated eigenvalues.

\begin{figure}[t]
\centering
\includegraphics[width=\linewidth,height=0.54\textheight,keepaspectratio]
{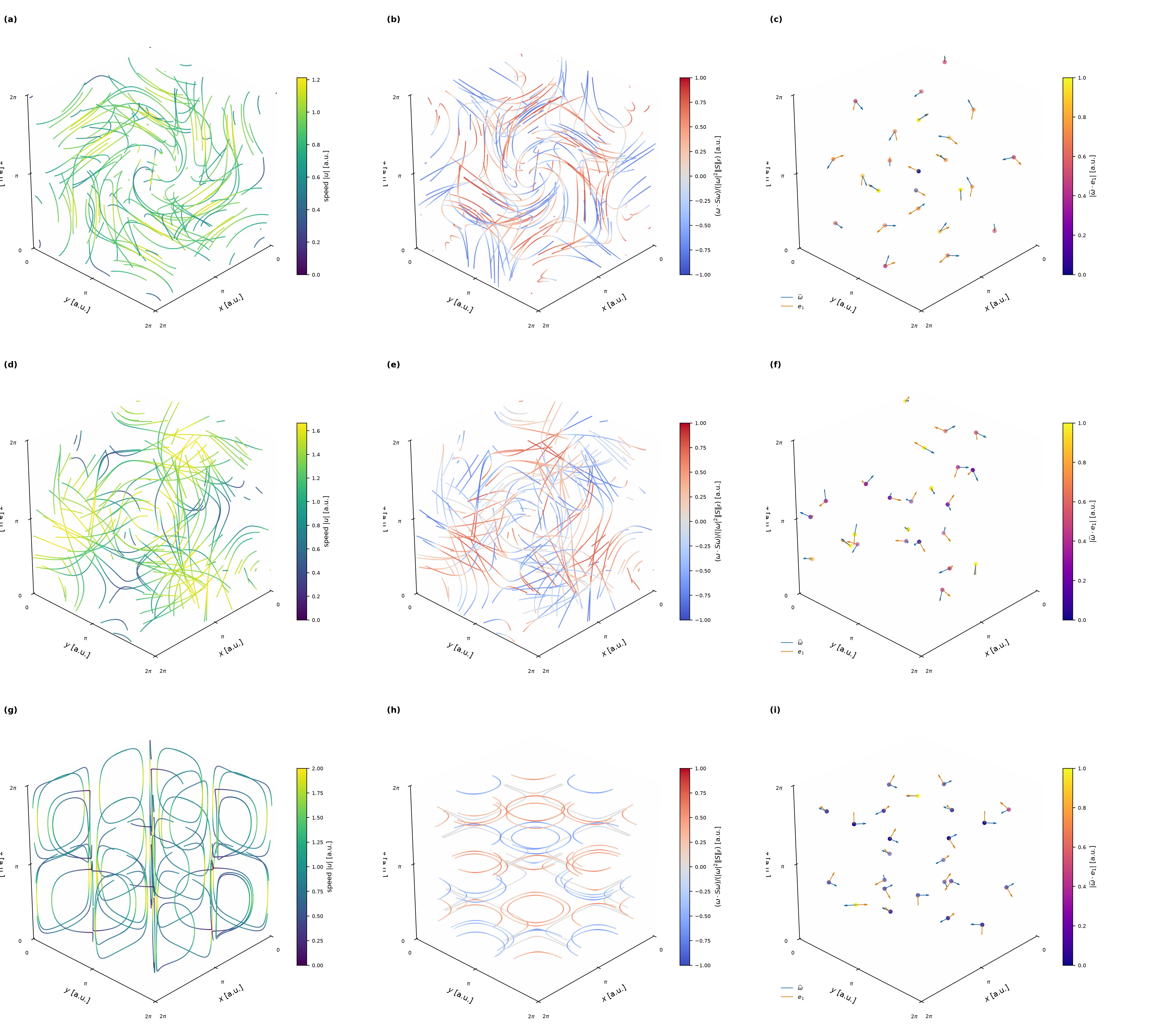}
\caption{Matched analytic geometry at $t=0$ for the three proved centres.
(a) Cyclic-shear velocity streamtubes coloured by speed. (b) Cyclic-shear
vorticity lines coloured by normalised signed stretching $\sigma_\omega$.
(c) Cyclic-shear vorticity and principal-extensional-strain directions coloured
by absolute alignment. (d) Equal-ABC velocity streamtubes coloured by speed.
(e) Equal-ABC vorticity lines coloured by $\sigma_\omega$. (f) Equal-ABC paired
directions coloured by absolute alignment. (g) TG3 velocity streamtubes coloured
by speed. (h) TG3 vorticity lines coloured by $\sigma_\omega$. (i) TG3 paired
directions coloured by absolute alignment. All panels use common cameras,
grids, integration rules, seed selection, and dimensionless colour ranges.}
\label{fig:multicentergeometry}
\end{figure}

Figure~\ref{fig:multicentergeometry}(a) shows the cyclic-shear centre's curved
velocity transport through all three coordinate directions. Panel (b) maps
normalised stretching along its vorticity lines. The alternating blue and red
segments identify compressive and extensional portions of the same coherent
curves. Panel (c) shows a broad alignment range between vorticity and the
principal extensional strain direction. Panel (d) displays the equal-ABC
velocity lines. Panel (e) reproduces the same line geometry for vorticity
because $\nabla\times U_{\rm ABC}=U_{\rm ABC}$, while its two-sign colour field
shows that zero projected nonlinearity does not require pointwise zero strain.
Panel (f) displays the corresponding alignment field. Panel (g) resolves the
more symmetric cellular velocity transport of TG3. Panel (h) shows two-sign
stretching on its vorticity lines. Panel (i) shows the TG3 alignment field.
The nonzero Fourier output proved in Appendix~\ref{app:algebra} distinguishes
this geometry from the Beltrami ABC case. All nine panels evaluate analytic
initial fields. The sampled curves represent the closed-form fields used directly in the
theorems.

\subsubsection{Paired-resolution configuration-space ensemble}
\label{sec:ensemble}

A preregistered finite-window ensemble contains 4,096 quasi-random
configurations and two discretizations per configuration, hence 8,192
Fourier--Galerkin trajectories. Independently of the exact comparison used in
Theorem~\ref{thm:main}, this ensemble quantifies finite-resolution variation
within the declared parameterisation.

The ensemble addresses variability rather than selecting a representative
trajectory. A scrambled Sobol low-discrepancy design supplies eight
coordinates. Two set the
centre amplitude $a$ and the exact perturbation radius
$\rho=\|h\|_{\A^2}$. The others control spectral tilt, phase coherence, phase
origin, two phase-dispersion coordinates, and polarization. For a fixed list
of 16 wave vectors $k_j$ with $1<|k_j|\leq\sqrt6$, choose an orthonormal pair
$e_{j,1},e_{j,2}\perp k_j$ and set
\begin{align}
 \widehat h(k_j)&=c_j\e^{i\phi_j}
 (\cos\psi_j\,e_{j,1}+\sin\psi_j\,e_{j,2}),\qquad
 \widehat h(-k_j)=\overline{\widehat h(k_j)},\label{eq:ensembleh}\\
 c_j&=C_\rho |k_j|^\beta,\qquad
 \sum_j2|k_j|^2|c_j|=\rho.\label{eq:ensemblenorm}
\end{align}
Appendix~\ref{app:numerics} verifies the reality, solenoidality, and exact
$\A^2$ normalisation of this perturbation parameterisation.
Thus every sample is real and divergence free, and normalisation
\eqref{eq:ensemblenorm} enforces its recorded $\A^2$ radius rather than an
expected radius. The phase-coherence coordinate continuously interpolates
between a common phase and the dispersed deterministic phases used by the
Sobol design. It is a parameter of this 16-mode family, not a universal
measure of Fourier coherence.

The first stratum sampled 2,048 points in
\begin{align}
 &0\leq a\leq0.7,\qquad0\leq\rho\leq0.01,
 \label{eq:detail088}
\end{align}
and the second sampled 2,048 points in the stress box
\begin{align}
 &0.5\leq a\leq1.3,\qquad0.005\leq\rho\leq0.05.
 \label{eq:detail089}
\end{align}
Their overlap produced 2,105 theorem-box configurations and 1,991 exterior
configurations. The solver evolved each datum at $(n,K,\Delta t)=(16,3,0.01)$ and
$(20,4,0.005)$ to $T=0.5$. The diagnostics use the moments $E,X,Z$ defined in
Section~\ref{sec:setting}, with the cutoff subscript suppressed.
Here $D_R=\|A^{1/2}(A-\mu)u\|_2^2$ is centred dissipation and $F$ is the
signed source in $R'/2+\nu D_R=F$, derived in
Appendix~\ref{app:observable}. The computation records
\begin{align}
 &J_{0.5}=\int_0^{0.5}R(t)^{2/3}\dd t,\qquad
 C_F=\frac{|\int_0^{0.5}F(t)\dd t|}
 {\int_0^{0.5}|F(t)|\dd t},
 \label{eq:ensemblemetrics}
\end{align}
as well as $E,X,R,D_R$, the outer-cutoff-layer $Z$ fraction, divergence, and
nonlinear energy orthogonality. In the following comparison, subscripts
$c$ and $f$ denote coarse and fine calculations at $m$ common output times.
The remaining quantities measure the spectral-tail fraction, normalised
residual-balance defect, divergence and nonlinear energy orthogonality,
respectively, as defined in Appendix~\ref{app:timestepping}.
A trajectory is numerically trusted only when all six fixed gates hold:
\begin{align}
 &\begin{aligned}
 \frac{|J_c-J_f|}{J_f}&\leq10^{-3},&
 \frac{\left(m^{-1}\sum_{j=1}^{m}|R_c(t_j)-R_f(t_j)|^2\right)^{1/2}}
 {\max_j|R_f(t_j)|}&\leq10^{-5},&
 \mathrm{Tail}_Z&\leq10^{-8},\\
 \mathrm{Bal}_R&\leq5\times10^{-4},&
 \mathrm{Div}&\leq10^{-10},&
 \mathrm{Orth}&\leq10^{-10}.
\end{aligned}
\label{eq:trustgates}
\end{align}
The balance tolerance reflects trapezoidal integration of the diagnostic
identity, whereas the PDE step uses RK4. The paired calculation supplies a discretization-convergence
diagnostic rather than an interval enclosure.
Appendix~\ref{app:statistics} defines the cohort predicates, empirical
quantiles, Pearson correlations, gate counts, and order-statistic selections
used in this subsection.

\begin{figure}[t]
\centering
\includegraphics[width=\linewidth,height=0.46\textheight,keepaspectratio]{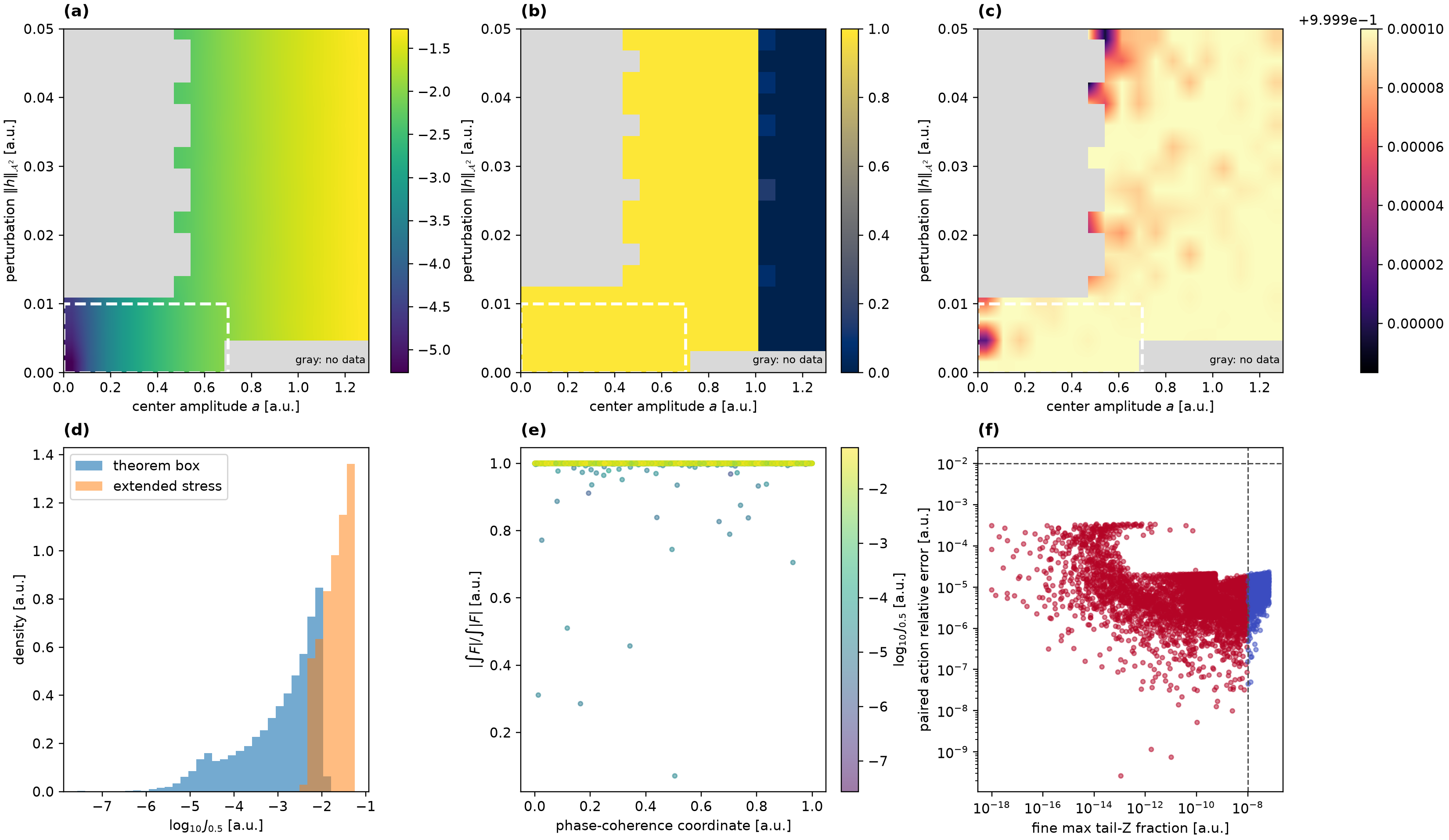}
\caption{Configuration-space ensemble: (a) binned mean $\log_{10}J_{0.5}$,
(b) trusted fraction under \eqref{eq:trustgates}, (c) median $C_F$,
(d) action densities, (e) phase-coordinate/cancellation scatter,
(f) tail and paired-action errors. Dashed rectangles mark the theorem box.
Gray heatmap cells contain no sampled datum.}
\label{fig:ensembleconfig}
\end{figure}

Figure~\ref{fig:ensembleconfig}(a) shows a monotone large-scale transition in
the action: across this design, $\log_{10}J_{0.5}$ has Pearson correlation
$0.916$ with $a$ and $0.538$ with $\rho$, but correlations with spectral tilt
and the phase coordinate are below $0.002$. These descriptive correlations
do not establish causation. They show that the centre amplitude dominates
the comparatively small 16-mode perturbation in this finite window. Panel
(b) separates proof membership from numerical resolution. All 2,105 theorem-box
samples pass every gate, whereas 1,257 of 1,991 exterior samples pass. The loss of
trust near $a\simeq1$ is entirely the outer-layer $Z$ gate, not a numerical
blowup or a violation of an analytic theorem. Gray regions were not sampled
and carry no inferred value.

Figure~\ref{fig:ensembleconfig}(c) shows that the median $C_F$ is nearly one
through most populated bins. Only 35 of 4,096 samples have $C_F<0.99$, and the
minimum is $0.0708$. At viscosity one and over $T=0.5$, the signed residual
source is therefore predominantly one-signed in this perturbation family.
This does not contradict the pointwise red/blue stretching in
Figure~\ref{fig:centerpanels}(d), because $F$ is a different, spatially
integrated residual observable. Panel (d) quantifies the separation in action:
the theorem-box $5\%/50\%/95\%$ quantiles are
$1.66\times10^{-5}$, $1.89\times10^{-3}$, and $9.76\times10^{-3}$. The
exterior quantiles are $5.70\times10^{-3}$, $2.22\times10^{-2}$, and
$5.24\times10^{-2}$. Panel (e) shows no systematic phase-coordinate trend.
The chosen phase coordinate therefore does not identify a cancellation
direction in this design. Panel (f) identifies the actual trust boundary: every
paired-action error is below $3.39\times10^{-4}$, while 734 exterior trajectories cross
the $10^{-8}$ spectral-tail threshold.

\begin{figure}[t]
\centering
\includegraphics[width=\linewidth,height=0.45\textheight,keepaspectratio]{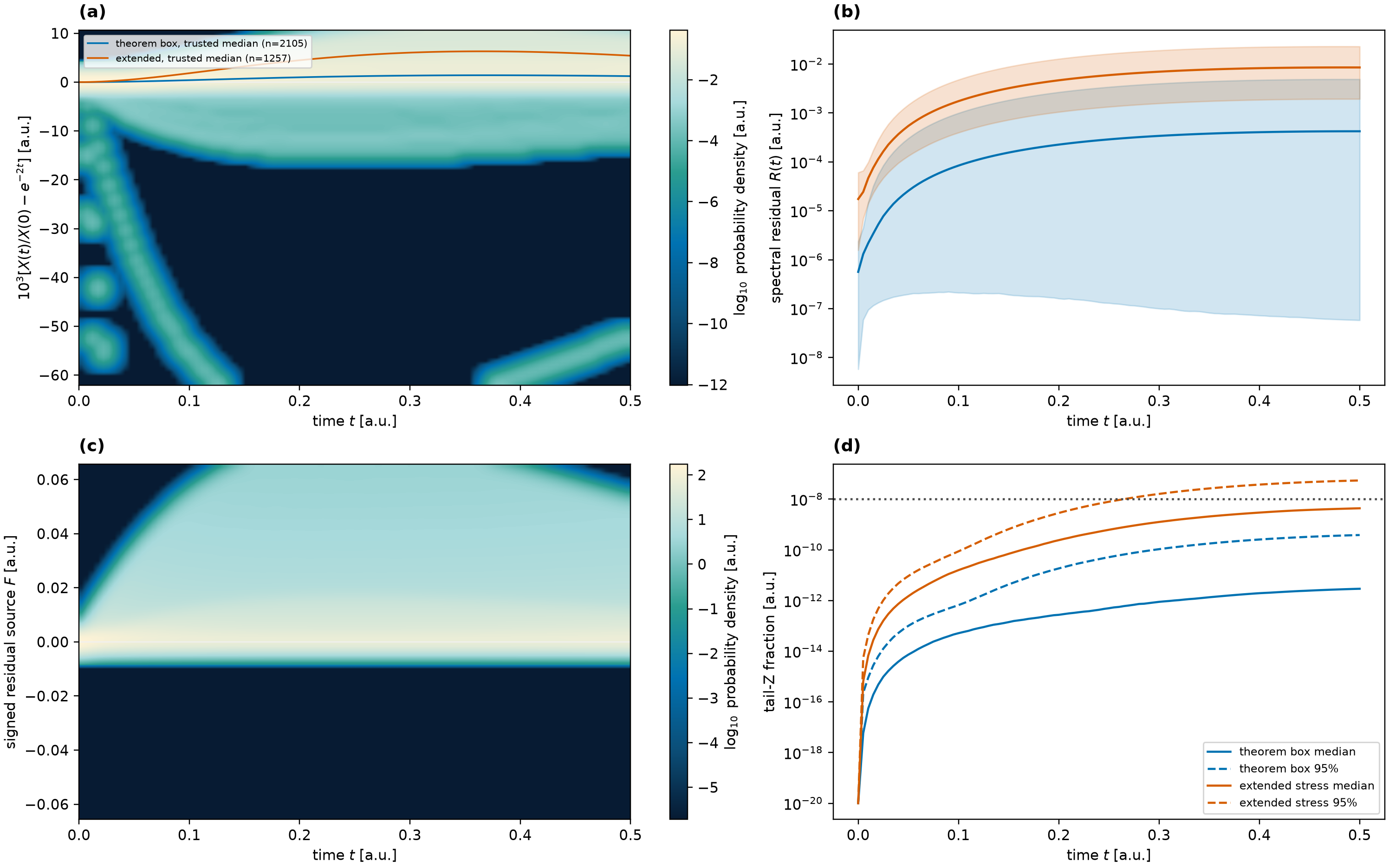}
\caption{Ensemble time transitions: (a) time-resolved density of
$10^3[X(t)/X(0)-e^{-2t}]$ over all trusted trajectories, with cohort medians,
(b) median and 5--95\% bands of $R$, (c) time-resolved density of the signed
source $F$, estimated with a mass-renormalized Gaussian kernel, (d) median and
95th-percentile outer-layer $Z$ fractions. Blue denotes theorem-box data and
vermillion denotes the exterior stress cohort throughout.}
\label{fig:ensembletransitions}
\end{figure}

Figure~\ref{fig:ensembletransitions}(a) resolves the narrow distribution
around the single-shell viscous reference $e^{-2t}$. The density uses every
trusted trajectory. The two thin curves are cohort medians, not the underlying
sample. Positive values measure nonlinear or multishell retardation relative
to that reference, while negative values indicate faster normalised decay.
Panel (b) restores the absolute distinction: exterior configurations produce a substantially larger
residual distribution, whereas the theorem-box band remains lower and wider
in relative terms because samples near $a=0$ are perturbation dominated.
Panel (c) exposes the distribution suppressed by a single mean curve. The
Gaussian kernel acts only on the displayed empirical density. Each time slice
is renormalized to unit probability and no trajectory value is altered. Its
high-density ridge lies just above $F=0$, with a broad positive tail and a
much thinner negative population. Panel (d) explains the trust transition:
the theorem-box 95th percentile stays below $10^{-9}$, while the exterior
95th percentile crosses the $10^{-8}$ gate near the end of the window. The
crossing diagnoses the active resolution boundary and motivates the
higher-cutoff exterior campaign.

\subsubsection{High-resolution spectral-boundary refinement}
\label{sec:refinement}

The resolution warning in Figure~\ref{fig:ensembletransitions}(d) motivates a
second study based on 512 configurations selected before the refined runs.
The selection contains 228 theorem-box configurations and 284 exterior
configurations. It is deliberately enriched near the original cutoff-four
tail boundary and includes fixed extreme-case controls. It therefore resolves
that boundary efficiently but does not represent a probability sample. The solver
embedded each datum without modification at cubical cutoffs $K=5,6,7$ and repeated
a 64-case subset at half the cutoff-seven time step. This design
comprises 1,600 trajectories. At cutoff $K$, the moving outer-shell fraction
is
\begin{align}
 &q_K=\max_{0\leq t\leq0.5}
 \frac{\sum_{\|k\|_\infty=K}|k|^4|\widehat u_k(t)|^2}
      {\sum_{\|k\|_\infty\leq K}|k|^4|\widehat u_k(t)|^2}.
 \label{eq:highrestail}
\end{align}
The frozen numerical classification requires cutoff-six/cutoff-seven action
agreement, residual-trajectory agreement, $q_7\leq10^{-8}$, and independent
balance, divergence, and nonlinear energy-orthogonality gates. These
finite-discretization diagnostics complement the exact bounds used by the
theorem.

\begin{figure}[t]
\centering
\includegraphics[width=\linewidth,height=0.47\textheight,keepaspectratio]{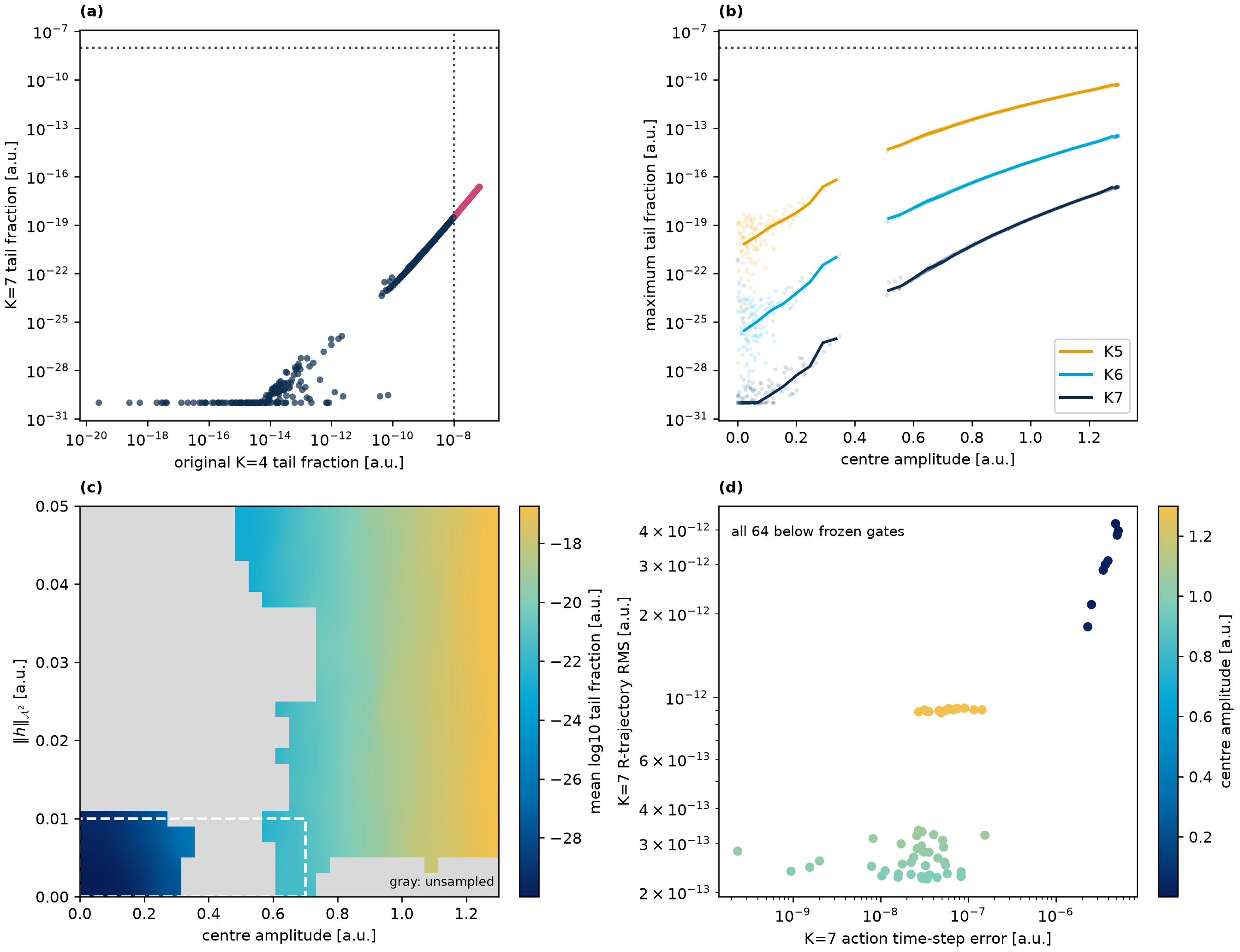}
\caption{High-resolution audit of 512 configurations selected before the
refined calculations. (a) Cutoff-four against cutoff-seven outer-shell
fractions, with navy and magenta denoting the original pass and fail classes.
(b) Individual outer-shell fractions and binned medians at cutoffs five, six,
and seven. (c) Display-smoothed mean cutoff-seven log-tail fraction in the
amplitude--perturbation plane, where the dashed rectangle marks the proved
neighbourhood and gray cells contain no selected configuration. (d) Half-step errors for the
64-case temporal-refinement subset. Dotted lines mark frozen gates.}
\label{fig:highresphase}
\end{figure}

Figure~\ref{fig:highresphase}(a) shows that all 124 configurations that failed
the original cutoff-four tail gate fall below the same threshold at cutoff
seven. Panel (b) demonstrates a systematic contraction rather than an
isolated two-grid change. Median outer-shell fractions decrease from
$9.29\times10^{-14}$ at cutoff five to $7.80\times10^{-18}$ at cutoff six
and $6.09\times10^{-22}$ at cutoff seven. Their respective maxima are
$5.12\times10^{-11}$, $3.37\times10^{-14}$, and $2.39\times10^{-17}$.
Because the measured shell moves with the cutoff, this trend establishes
finite-family spectral localization rather than a cutoff-independent bound.
The maximum relative action differences are $8.44\times10^{-5}$,
$2.11\times10^{-5}$, and $1.81\times10^{-14}$ for the three successive cutoff
pairs. Panel (c) locates the largest resolved tails at
high centre amplitude. Its irregular
gray boundary records unsampled parameter bins, not a physical transition.
The trust reclassification contains four counts: no configuration fails at
cutoff seven, all 124 original failures become passes, and all 388 original
passes remain passes. Panel (d) shows that all 64 half-step comparisons
satisfy both frozen temporal gates. Together these results separate the
spatial reclassification from its temporal check.

\begin{figure}[!t]
\centering
\includegraphics[width=\linewidth,height=0.36\textheight,keepaspectratio]{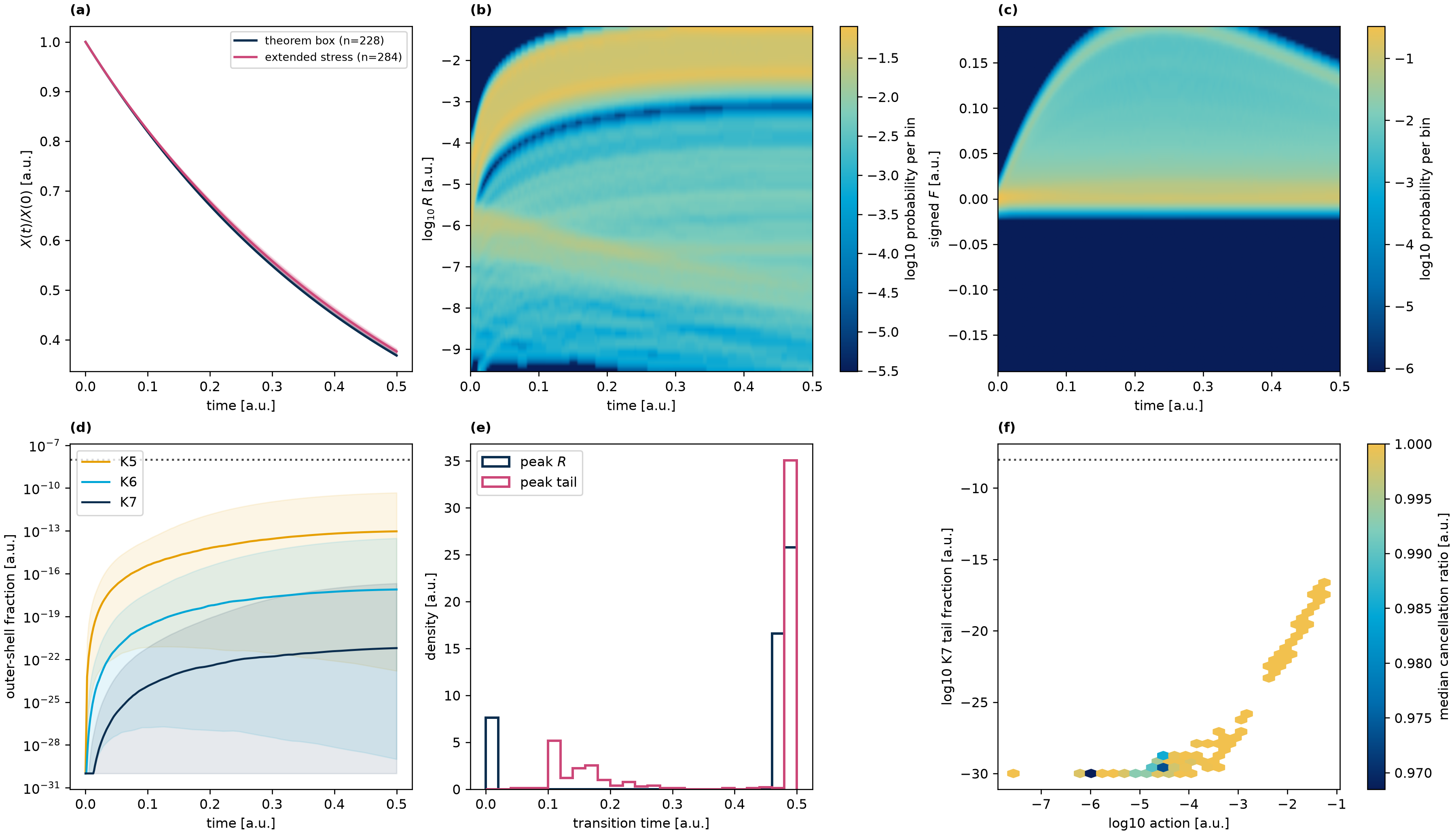}
\caption{Time-resolved distributions for the 512 cutoff-seven trajectories.
(a) Median, interquartile, and 5--95\% bands of normalised enstrophy for the
theorem-box and exterior cohorts. (b) Kernel-smoothed bin probabilities of $\log_{10}R$.
(c) Kernel-smoothed bin probabilities of the signed residual source $F$. (d) Median and
5--95\% bands of the moving outer-shell fraction at cutoffs five, six, and
seven. (e) Distributions of the times at which $R$ and the outer-shell
fraction attain their maxima. (f) Hexagonal aggregation of action and
cutoff-seven tail fraction, coloured by median signed cancellation. Density
normalisation conditions on the displayed range and assigns unit total probability to each smoothed time column.}
\label{fig:highrestransitions}
\end{figure}

Figure~\ref{fig:highrestransitions}(a) shows viscous enstrophy decay throughout
both cohorts over the observed window. Panel (b) resolves several persistent residual
bands that a representative curve would hide. Panel (c) separates a dense
near-zero source population from a positive branch that first grows and then
weakens. Panel (d) displays both the ensemble spread and its systematic
contraction with increasing cutoff. Panel (e) shows that the residual maximum
usually occurs at entry or near the terminal time, while the tail maximum is
commonly terminal. Panel (f) identifies centre amplitude as the main sampled
coordinate along which action and tail occupation grow together. Appreciable
signed cancellation is concentrated in the low-action portion of the design.
No cutoff-seven point approaches the $10^{-8}$ tail gate. The maximum
residual-balance, divergence, and nonlinear energy-orthogonality defects are
$1.64\times10^{-5}$, $6.51\times10^{-19}$, and $5.55\times10^{-17}$,
respectively.

\subsubsection{Physics-informed operator search across initial-data families}
\label{sec:pinoresults}

The operator experiment tests whether the physical constraints used by the
analysis can improve learned velocity evolution and focus expensive spectral
recomputations. Figure~\ref{fig:pinomain} summarises the independently
reconstructed OOD predictions and all three-level candidate rechecks.

\begin{figure}[t]
\centering
\includegraphics[width=\linewidth,height=0.47\textheight,keepaspectratio]{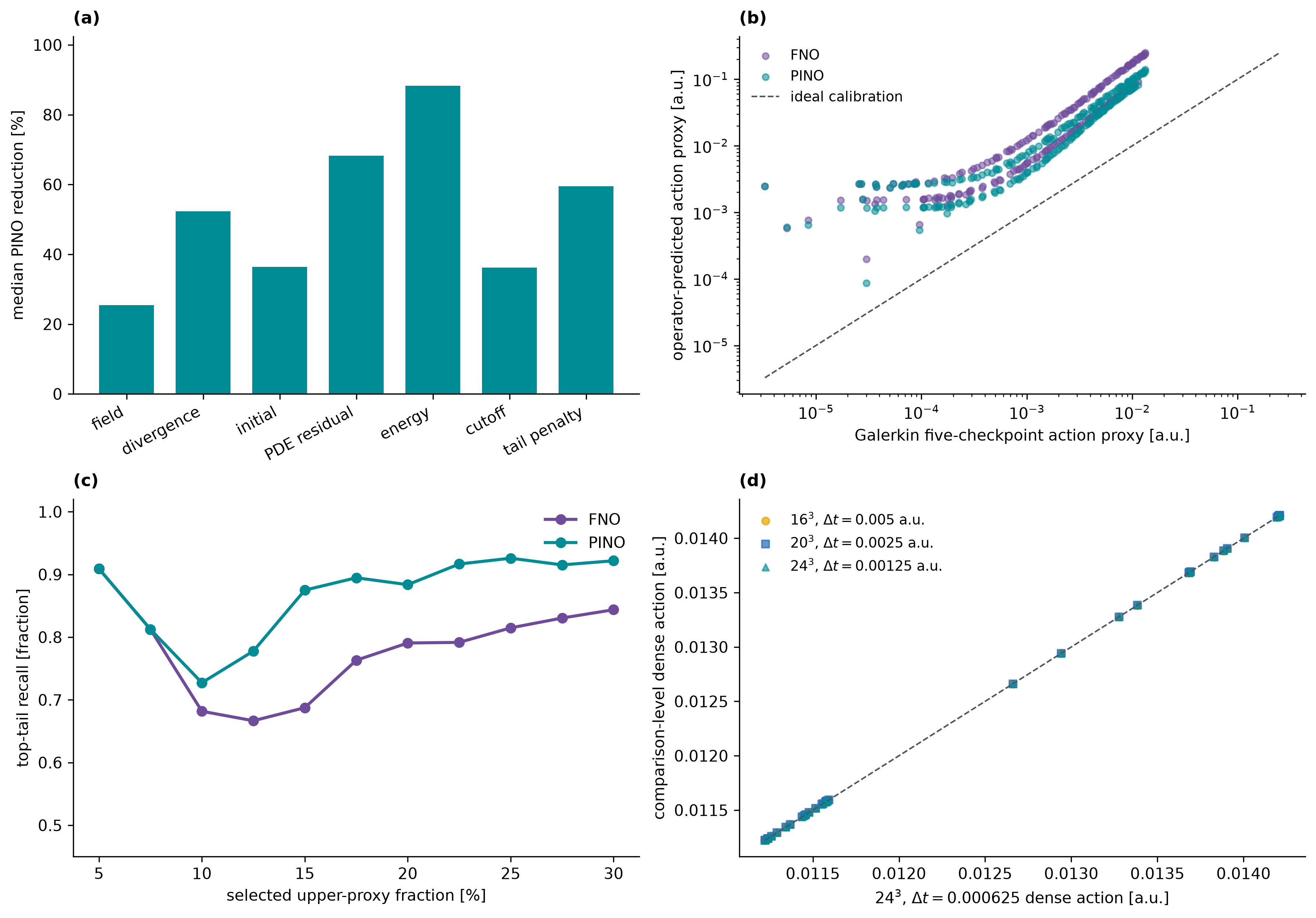}
\caption{Physics-informed neural-operator field prediction and targeted search.
(a) Percentage reduction in the OOD median recorded quantity for PINO relative to the
architecture-matched FNO for field prediction, divergence, initial condition,
projected Navier--Stokes residual, energy balance, unresolved nonlinear output,
and the predicted outer-band tail penalty. (b) Predicted five-checkpoint residual-action
proxy against the Galerkin proxy for 212 OOD fields. (c) Action-proxy upper-tail recall as the
selected fraction varies from 5\% to 30\%. (d) Actions for 48 selected fields
from dense, every-step quadrature at three grid--time-step pairs against the
$24^3$, $\Delta t=0.000625$ calculation. All actions, time steps, and nondimensional errors use
arbitrary units (a.u.).}
\label{fig:pinomain}
\end{figure}

Figure~\ref{fig:pinomain}(a) shows that physics-informed training lowers the
median of all seven recorded OOD quantities for this fitted seed. Median field error decreases by $25.5\%$, divergence
error by $52.3\%$, initial-condition error by $36.4\%$, projected
Navier--Stokes residual by $68.2\%$, energy-balance error by $88.3\%$,
native-grid cutoff-sensitivity defect by $36.2\%$, and predicted outer-band
tail penalty by $59.5\%$. Paired
resampling intervals for the FNO-minus-PINO metric differences remain
positive for every recorded quantity. These intervals resample the fixed OOD
cases and do not measure variability across training seeds. The reductions
therefore describe these fitted models rather than a universal effect of the
physics losses.

Panel (b) shows that both operators preserve the ordering of the action proxy more
accurately than its absolute scale. PINO reaches an OOD Spearman correlation of
$0.973$, compared with $0.964$ for FNO, and lowers the median predicted-to-
Galerkin proxy ratio from $12.5$ to $8.0$. Panel (c) adds a threshold-swept
action-proxy retrieval experiment. The models tie at the upper $5\%$ and
$7.5\%$ thresholds. PINO exceeds FNO at every displayed selection fraction
from $10\%$ to $30\%$, reaching recall $0.922$ rather than $0.844$ at $30\%$.
This exploratory sweep favours PINO over part of the displayed range. The
independently specified composite action--tail top-decile endpoint instead
decreases from $0.955$ for FNO to $0.864$ for PINO, with a paired 95\%
interval $[-0.227,0.045]$ for the gain. Appendix
\ref{app:pino} and Figure~\ref{fig:pinoappendix} give that complementary
endpoint and the complete diagnostic set.

Panel (d) closes the search loop with direct spectral recomputation rather than
neural prediction. The dense action records $R^{2/3}$ at every RK4 step for all
48 selected fields. Relative to the $24^3$, $\Delta t=0.000625$ calculation,
the maximum discrepancy is $1.50\times10^{-4}$ at $(16^3,0.005)$,
$2.09\times10^{-5}$ at $(20^3,0.0025)$, and $4.17\times10^{-6}$ for the
$24^3$ half-step comparison. By contrast, the original five-checkpoint proxy
differs from dense quadrature by $1.12\%$--$4.14\%$, with median $2.65\%$.
Under matched candidate counts, the best neural proposal exceeds the stronger
PINO-screened random-or-Sobol proposal by $0.68\%$ in the cyclic-shear family
and $1.46\%$ in the TG3 family when evaluated with the finest dense action.
None of the 48 fields enlarges an exactly proved neighbourhood. The exact
theorem uses only the independently validated coefficient cells established
above.

\subsubsection{Objectively selected dynamic vortex geometry}
\label{sec:dynamicgeometry}

Three ensemble members selected by fixed order statistics were rerun at
$48^3$, cutoff $K=8$, and $\Delta t=0.0025$. The selection rule excludes
manual choice based on appearance. The reruns provide the fields used in the
dynamic three-dimensional panels.

To prevent visual selection bias, the cases are the largest-action trusted
theorem member, the largest-action trusted exterior member, and the largest
tail among rejected exterior members. For each state, the displayed surface
satisfies
\begin{align}
 &\frac{|\omega(x,t)|}{\|\omega(t)\|_\infty}=0.78,
 \qquad
 c(x,t)=\frac{\omega\cdot S\omega}
 {\|\omega\cdot S\omega\|_\infty},
 \label{eq:vortexsurface}
\end{align}
where $c\in[-1,1]$ supplies the colour. The $Q$-criterion and helicity stored
in each Visualization Toolkit image-data (VTI) volume are
\begin{align}
 &Q=\frac14|\omega|^2-\frac12|S|^2,
 \qquad \mathcal H=u\cdot\omega.
 \label{eq:qhelicity}
\end{align}

\begin{figure}[t]
\centering
\includegraphics[width=\linewidth,height=0.48\textheight,keepaspectratio]{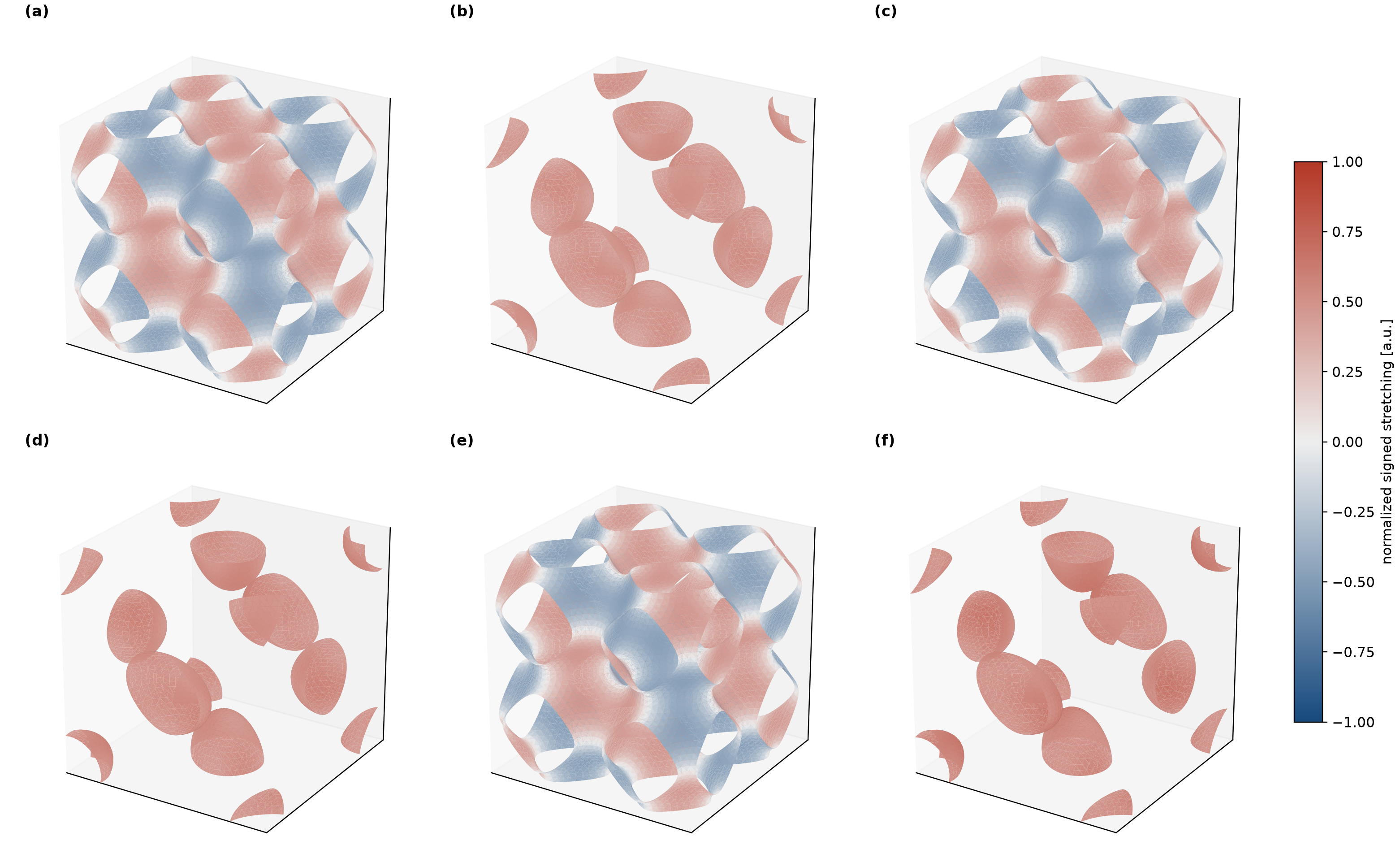}
\caption{Selected dynamic vortex surfaces. (a) Theorem-box maximum-action case
at $t=0$. (b) The same case at $t=0.5$. (c) Trusted exterior maximum-action
case at $t=0$. (d) The same case at $t=0.5$. (e) Original cutoff-four
maximum-tail case at $t=0$. (f) The same case at $t=0.5$. Geometry and colour
follow \eqref{eq:vortexsurface}. All panels use the same camera and colour
normalisation.}
\label{fig:dynamicvortices}
\end{figure}

\begin{center}\small
\begin{tabular}{@{}lrrrrrr@{}}
\toprule
Selection & $a$ & $\rho$ & $J_{0.5}$ & max tail $Z$ & $\|\omega(0)\|_\infty$ & $\|\omega(.5)\|_\infty$\\
\midrule
Theorem max action & 0.6998 & 0.00305 & 0.01107 & $5.57\times10^{-10}$ & 1.212 & 0.882\\
Exterior trusted & 1.0128 & 0.02601 & 0.02947 & $9.94\times10^{-9}$ & 1.756 & 1.360\\
Exterior max tail & 1.2990 & 0.02504 & 0.05670 & $6.90\times10^{-8}$ & 2.253 & 1.829\\
\bottomrule
\end{tabular}
\end{center}

Figure~\ref{fig:dynamicvortices}(a) shows the theorem-box maximum-action case
at $t=0$. Its intense-vorticity surface retains the interlocking single-shell
topology and contains both stretching signs. Panel (b) shows the same case at
$t=0.5$. The surface separates into lobes and becomes predominantly red on
the conditional level $|\omega|/\|\omega\|_\infty=0.78$. Panel (c) shows the
trusted exterior maximum-action case at entry, with the same interlocking
topology and two-sign stretching. Panel (d) shows its terminal geometry, where
separated positively stretched lobes again dominate the selected level set.
Panel (e) shows the original cutoff-four rejected maximum-tail case at entry.
Panel (f) shows its terminal lobes. This case has the largest action and tail,
yet its maximum vorticity still decreases. Its rejection therefore concerns
the original numerical resolution rather than observed singular growth.
The common normalisation permits topological comparison, while the table
retains the dimensional amplitudes: all three $\|\omega\|_\infty$ values
decrease over the window. None of the panels implies positive stretching
throughout the volume. Surfaces clipped at the displayed fundamental-cell
boundary reconnect periodically. Appendix~\ref{app:statistics} specifies the
fixed order-statistic selections.

The isosurfaces emphasize intense-vorticity geometry but suppress the vector
fields that generate it. Figure~\ref{fig:fieldgeometry} therefore compares
integral curves and local directions using a deterministic seeding rule. For
vorticity lines the colour is the dimensionless signed quantity
\begin{align}
 &\sigma_\omega(x,t)=
 \frac{\omega\cdot S\omega}{|\omega|^2\,\|S\|_F},
 \qquad -1\leq\sigma_\omega\leq1,
 \label{eq:normalizedstretching}
\end{align}
with values set to zero whenever $|\omega|\,\|S\|_F=0$. The alignment diagnostic is
$|\widehat\omega\cdot e_{\max}|$, where $e_{\max}$ is the unit eigenvector associated
with the largest eigenvalue of $S$.
Appendix~\ref{app:physicalidentities} proves the bound in
\eqref{eq:normalizedstretching} and evaluates its sharper extremum
$\sqrt{2/3}$ for the analytic centre.

\begin{figure}[t]
\centering
\includegraphics[width=\linewidth]{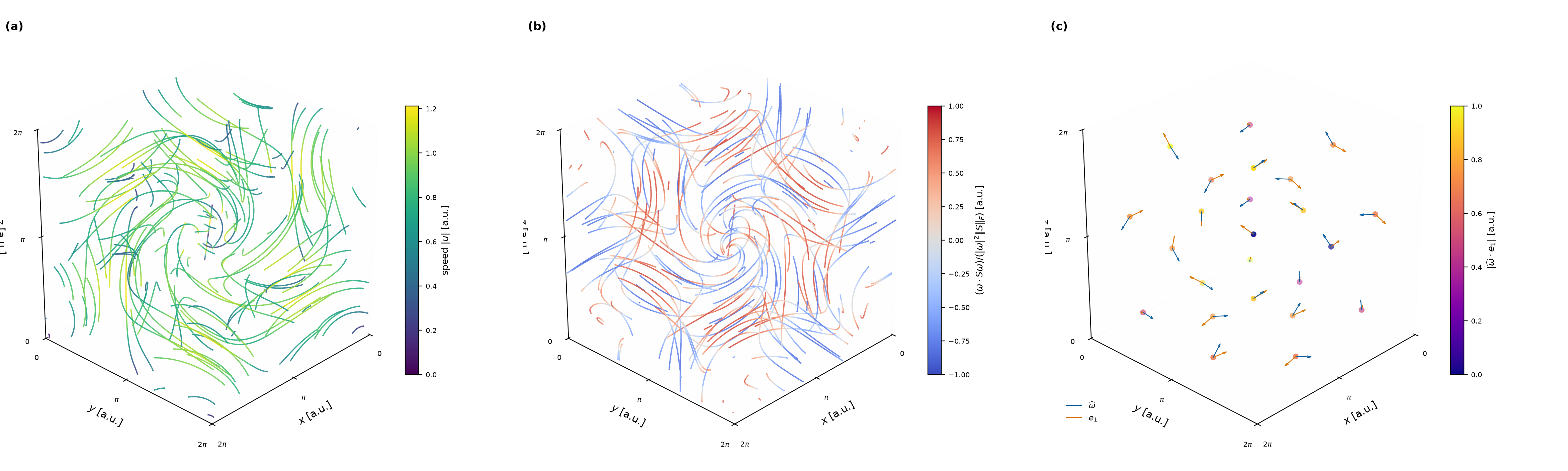}\\[-2pt]
\includegraphics[width=\linewidth]{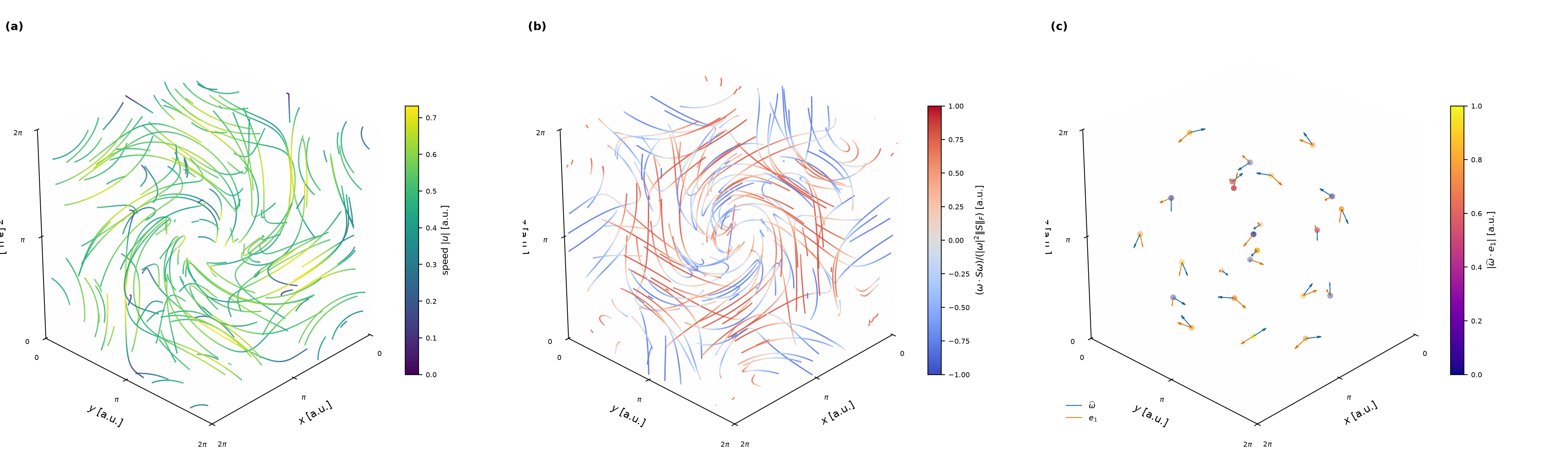}
\caption{Vector-field geometry under a common camera and deterministic seed
rule. (a) Centre velocity streamtubes coloured by speed, (b) centre vorticity
lines coloured by $\sigma_\omega$, (c) centre vorticity and principal-extensional
strain directions coloured by their absolute alignment, (d) case 1088 velocity
streamtubes at $t=0.5$, (e) case 1088 vorticity lines coloured by
$\sigma_\omega$ at $t=0.5$, (f) case 1088 paired directions coloured by absolute
alignment at $t=0.5$. Curves are split at periodic faces.}
\label{fig:fieldgeometry}
\end{figure}

Figure~\ref{fig:fieldgeometry}(a) makes the centre's three-dimensional
transport geometry explicit: none of the three velocity components is
passive, and the streamtubes turn through all coordinate directions. Panel
(b) separates the sign of instantaneous vortex stretching from vorticity
magnitude. Alternating red and blue curves show that the spatially integrated
dynamics arise from competing extensional and compressive regions rather
than a one-sign source. Panel (c) compares the normalised vorticity direction
with $e_{\max}$ at a deterministic sparse set of locations. Its median alignment
is $0.608$, with values spanning zero to one. Panel (d) applies the same camera
and integration rule to case 1088 at $t=0.5$. Its median speed decreases from
$0.857$ at the centre to $0.517$. Panel (e) shows that the median vorticity
magnitude decreases from $0.857$ to $0.503$, while normalised stretching
shifts from a zero median to $0.0517$. Its extrema remain close to the analytic
bounds $\pm\sqrt{2/3}$. Panel (f) shows the median principal-strain alignment
rising from $0.608$ to $0.658$. These changes describe the objectively
selected finite-resolution trajectory without asserting monotone alignment.
Apparent line endpoints at a bounding face
are periodic rendering cuts, not physical termination points.
Appendix~\ref{app:statistics} defines the median and range estimators.

\subsubsection{Genuine three-dimensional eigenfield transients}
\label{sec:eigentransients}

The finite-window Fourier--Galerkin calculations below accompany a
separately proved small-parameter asymptotic for exact solutions. This
asymptotic identifies the scaling tested across the varying family.

Consider the varying family $u_n(0,x)=U(nx)$ and set
\begin{align}
 &\tau=\nu n^2t,\qquad \varepsilon=(\nu n)^{-1},\qquad
 u_n(t,x)=v_\varepsilon(\tau,nx).
 \label{eq:detail096}
\end{align}
Then
\begin{align}
 &\partial_\tau v_\varepsilon+Av_\varepsilon
 =-\varepsilon B(v_\varepsilon,v_\varepsilon),
 \qquad v_\varepsilon(0)=U.
 \label{eq:detail097}
\end{align}
The first resonant Duhamel term and a uniform weighted Sobolev remainder give
\begin{align}
 &v_\varepsilon(\tau)
 =e^{-\tau}U-\varepsilon\tau e^{-2\tau}V
 +O_{H^s}(\varepsilon^2e^{-\tau}),\qquad s\geq4.
 \label{eq:detail098}
\end{align}
Since $AU=U$, $AV=2V$, $\langle U,V\rangle=0$, and
$\|V\|_2^2=3/4$, Taylor expansion of \eqref{eq:Rmoments} yields
\begin{align}
 &R(v_\varepsilon(\tau))
 =\frac34\varepsilon^2\tau^2e^{-4\tau}
 +O(\varepsilon^3e^{-2\tau}).
 \label{eq:Rasymptotic}
\end{align}
Consequently
\begin{align}
 \int_0^\infty R(u_n(t))^{2/3}\dd t
 &=C_R\nu^{-7/3}n^{-2/3}
 \left[1+O((\nu n)^{-2/3})\right],\label{eq:actionasymptotic}\\
 C_R&=\left(\frac34\right)^{2/3}
 \Gamma\!\left(\frac73\right)\left(\frac38\right)^{7/3}.
\end{align}
For the signed source $\mathcal F$ in
$\frac12R'+\nu D_R=\mathcal F$, the same expansion gives
\begin{align}
 &\int_0^\infty\mathcal F(u_n(t))\dd t
 =\frac{3n^2}{64\nu^2}\left[1+O((\nu n)^{-1})\right].
 \label{eq:fluxasymptotic}
\end{align}
Appendix~\ref{app:eigenasymptotic} derives
Equations~\eqref{eq:Rasymptotic}--\eqref{eq:fluxasymptotic}, including the
space--time rescaling, the uniform weighted-space remainder, Gamma integral,
and centred-dissipation coefficient. Appendix~\ref{app:heatexpansion}
expands the contraction radius, the second-order remainder and the
fractional-power integration step. These formulae are asymptotic as
$\nu n\to\infty$. The contraction argument supplies a positive but
non-optimised smallness threshold for $|\varepsilon|=(\nu n)^{-1}$.

For $n=4,8,16,32$, the computed trajectories illustrate the predicted
asymptotic by collapsing toward
$3\tau^2/4$ after division of $R$ by $\varepsilon^2e^{-4\tau}$, as predicted
by \eqref{eq:Rasymptotic}. The integrated signed source grows as $n^2$ while
the critical action decreases as $n^{-2/3}$, in agreement with
\eqref{eq:actionasymptotic} and \eqref{eq:fluxasymptotic}. Thus a large signed
residual-source budget can coexist with a small regularity-relevant action.
Two spatial cutoffs and half-step checks at $n=4,32$ pass the declared gates.
These four finite values illustrate the limit and do not determine the
analytic contraction threshold. Because the datum varies with $n$, the calculation tests scaling across a
family rather than high-frequency evolution of one fixed datum. It rejects
the heuristic that source magnitude alone measures proximity to loss of
regularity. Appendix~\ref{app:numerics}
defines the deterministic refinement comparisons.

\subsubsection{Three-resolution stress test and VTI time slices}
\label{sec:fixedstress}

A finite-window, double-precision projected-RK4 experiment tests a fixed
three-component triad. Ten of 11 preregistered gates pass, while the coarse
tail gate prevents a trusted classification.

The stress datum is a $z$-invariant, two-dimensional three-component (2D3C)
triad at viscosity $0.01$. The solver projected the same Fourier coefficients onto grids
$32^3,48^3,64^3$ with cubical cutoffs $7,11,15$ and integrated each system to
$t=0.2$ with maximum steps $0.01$ and $0.005$. The solver uses strict
projection, Leray projection, and classical RK4. It records
\begin{align}
 &E=\|u\|_2^2,\quad X=\|\nabla u\|_2^2,\quad Z=\|Au\|_2^2,\quad
 R=Z-X^2/E,
 \label{eq:detail102}
\end{align}
together with the centred dissipation, signed source, spectra, transfer,
divergence, nonlinear energy defect, and integrated balance residuals.

Enstrophy grows by only $0.73\%$ over the short window, while the two
spectral-spread observables grow more strongly. The $48/64$ common-mode
discrepancy at $t=0.2$ is $9.26\times10^{-7}$, compared with
$6.81\times10^{-5}$ for $32/48$. Temporal errors are approximately
$4.3\times10^{-9}$.

\begin{figure}[t]
\centering
\includegraphics[width=\linewidth,height=0.38\textheight,keepaspectratio]{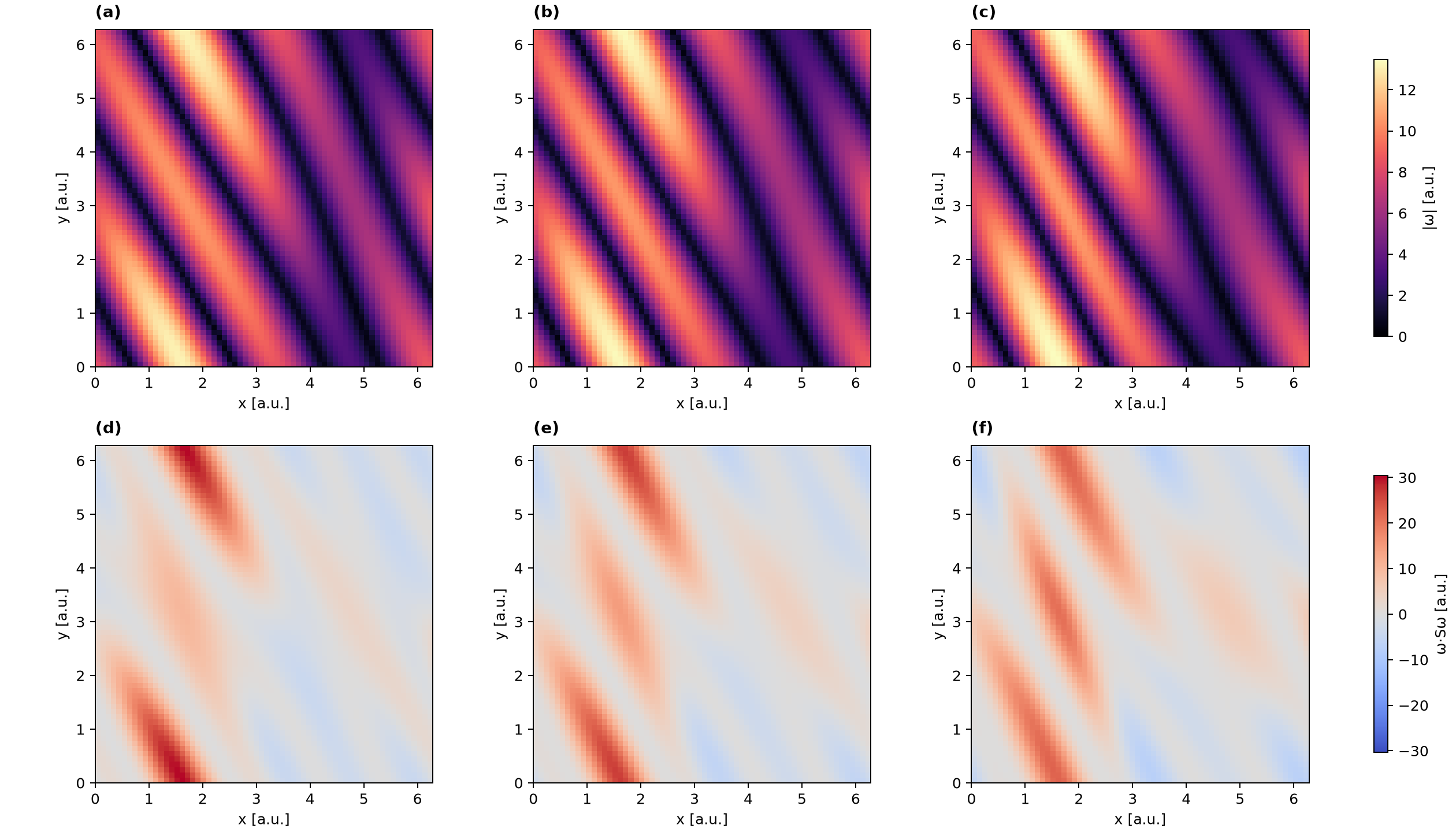}
\caption{$64^3$, cutoff-15 reference slices. (a) $|\omega|$ at $t=0$.
(b) $|\omega|$ at $t=0.1$. (c) $|\omega|$ at $t=0.2$.
(d) $\omega\cdot S\omega$ at $t=0$. (e) $\omega\cdot S\omega$ at $t=0.1$.
(f) $\omega\cdot S\omega$ at $t=0.2$. Each observable uses one common colour
range.}
\label{fig:fixedslices}
\end{figure}

Figure~\ref{fig:fixedslices}(a) shows the initial vorticity-magnitude slice.
Panel (b) shows its moderate redistribution at $t=0.1$, and panel (c) shows
the corresponding terminal state at $t=0.2$. Their common nonnegative colour
scale permits direct amplitude comparison. Panel (d) shows the initial signed
stretching field. Panel (e) resolves its intermediate rearrangement, and panel
(f) shows the terminal stretching pattern. Red and blue denote positive and
negative stretching on one common diverging scale. The coherent diagonal
pattern in every panel is imposed by the declared $z$-translation-invariant
2D3C symmetry and must not be interpreted as generic three-dimensional
turbulence. The moderate change is consistent with the recorded $0.73\%$
enstrophy increase. The recorded shell spectra
show later-time population of higher shells, while signed nonlinear transfer
contains both production and depletion. Taking absolute values would erase
the transfer direction. Values near $10^{-30}$ are the plotting floor, not
resolved physics. The maximum last-two-layer enstrophy
fractions are $1.13\times10^{-5}$, $2.21\times10^{-9}$, and
$3.69\times10^{-13}$ for cutoffs $7,11,15$. The first exceeds the
preregistered $10^{-6}$ gate. The coarse run therefore retains an unresolved
tail classification despite strong refinement, in contrast with the
theorem-box ensemble in Figure~\ref{fig:ensembleconfig}(b).

The largest integrated balance residual was $3.54\times10^{-6}$, the largest
discrete divergence norm $4.40\times10^{-17}$, and the largest nonlinear
energy defect $1.45\times10^{-15}$. These values support implementation
consistency on this finite window. The theorem instead obtains its
cutoff-uniform conclusion from the exact comparison bounds.
Appendix~\ref{app:numerics} defines
each discrepancy, tail fraction, and balance defect used in these numerical
assertions.

\subsubsection{Public-DNS tensor geometry across flow classes}
\label{sec:publicdns}

The external comparison uses seven official JHTDB acquisitions
\cite{Perlman,JHTDB} and $73{,}728$ provider-interpolated velocity-gradient
evaluations. Eight stored times and a second five-time sample that also
includes velocity come from the extended forced-isotropic $1024^3$ dataset,
whose service identifier is \texttt{isotropic1024coarse}
\cite{JHTDBIso1024}. Four $16^3$ structured queries sample the same dataset,
the forced-isotropic $4096^3$ snapshot \texttt{isotropic4096}
\cite{JHTDBIso4096}, the forced magnetohydrodynamic (MHD) $1024^3$ dataset
\texttt{mhd1024} \cite{JHTDBMHD}, and the rotating--stratified $4096^3$
dataset \texttt{rotstrat4096} \cite{JHTDBRotStrat}. A final Sobol sample uses
the turbulent-channel dataset \texttt{channel} \cite{JHTDBChannel}. These
flows do not satisfy a common
set of equations, forcing, or boundary conditions. They are therefore used
only to test the physical specificity of the geometric observables, never as
premises of the regularity theorem.

In the public-DNS diagnostics, $A$ denotes the local velocity-gradient
matrix rather than the Laplacian operator, and $R$ denotes its cubic trace
diagnostic rather than the global spectral residual. For each returned
gradient $A_{ij}=\partial_j u_i$, define
\begin{align}
 &S=\tfrac12(A+A^\top),\qquad
 \omega=(A_{32}-A_{23},A_{13}-A_{31},A_{21}-A_{12}),
 \label{eq:jhtdbstrain}
\end{align}
and evaluate
\begin{align}
 &P_\omega=\omega\cdot S\omega,\qquad
 \sigma=\frac{P_\omega}{|\omega|^2|S|_F},\qquad
 Q=-\tfrac12\operatorname{tr}(A^2),\qquad
 R=-\tfrac13\operatorname{tr}(A^3).
 \label{eq:jhtdbgeometry}
\end{align}
The convention is $\sigma=0$ when $|\omega|\,|S|_F=0$. Because the
provider-interpolated gradients have a nonzero trace defect, the displayed
$Q$ is the incompressible-form diagnostic $-\operatorname{tr}(A^2)/2$, not
the general second principal invariant
$[(\operatorname{tr}A)^2-\operatorname{tr}(A^2)]/2$.
Equations \eqref{eq:jhtdbstrain}--\eqref{eq:jhtdbgeometry} define
coordinate-invariant local observables rather than assumptions in the proof.
Their tensor contractions and physical sign interpretation are derived in
Appendix~\ref{app:physicalidentities}.
The alignment variables are $c_i=|\widehat\omega\cdot e_i|$. The calculation orders the
strain eigenvectors $e_i$ from most compressive to most extensional.
When velocity is present, the helicity density is $h=u\cdot\omega$. The queries requested
gradients directly with JHTDB's fourth-order
finite-difference/fourth-order Lagrange operator rather than reconstructing
derivatives from the sampled lattice.

\begin{figure}[t]
\centering
\includegraphics[width=\linewidth,height=0.46\textheight,keepaspectratio]{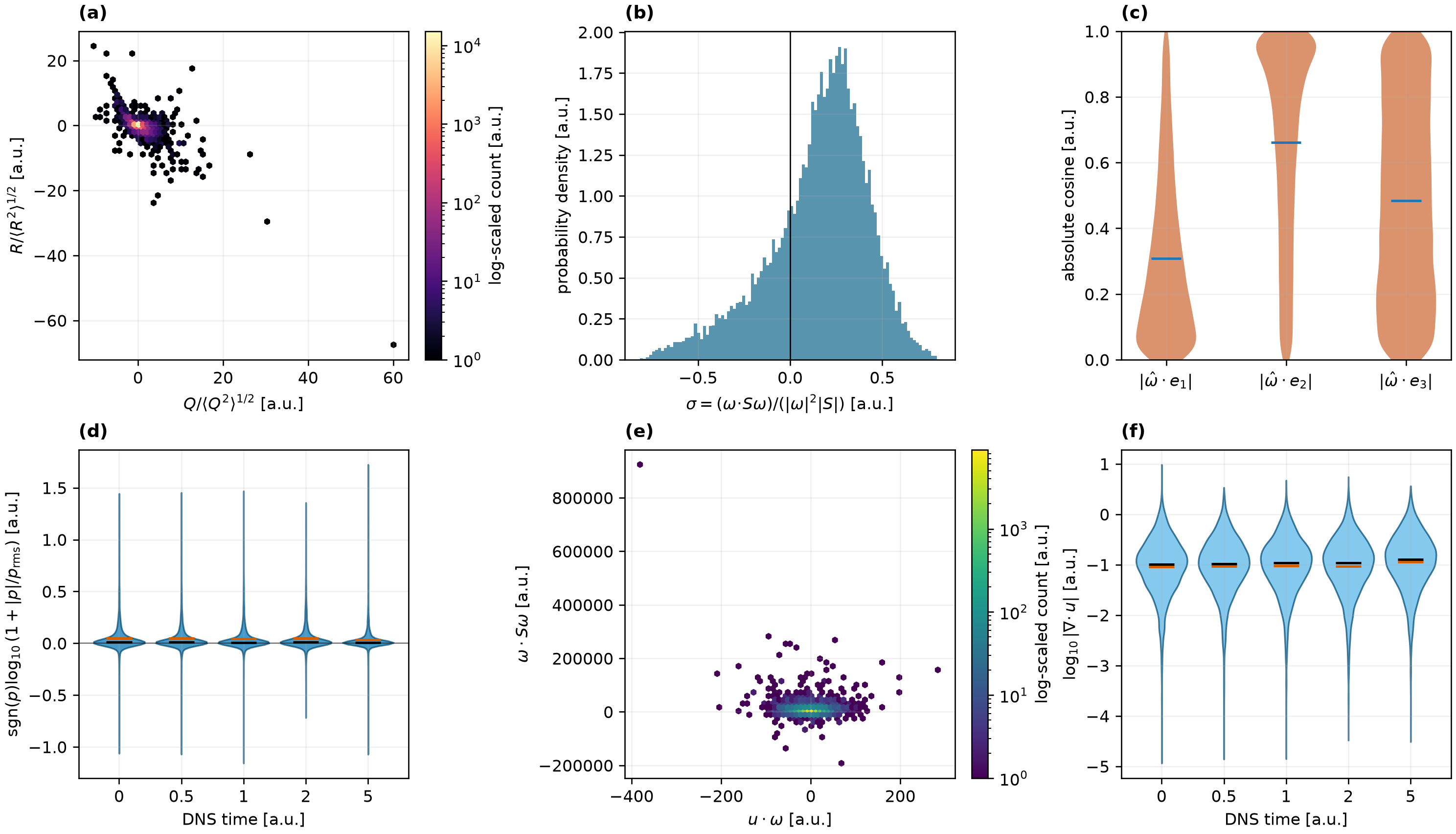}
\caption{Forced-isotropic public-DNS geometry. (a) Joint $Q$--$R$ density
from $20{,}480$ samples, normalised by pooled root-mean-square (RMS) values,
(b) probability density of normalised signed stretching $\sigma$,
(c) distributions of absolute vorticity alignment with the three ordered
strain eigenvectors, with horizontal bars denoting means, (d) full production
distributions at five stored DNS times under the monotone signed-log map
$p\mapsto\operatorname{sgn}(p)\log_{10}(1+|p|/p_{\rm rms})$,
with orange means and black medians,
(e) joint density of helicity and production, (f) provider-interpolation
divergence distributions, shown as $\log_{10}|\nabla\cdot u|$. Density colours
encode logarithmic sample counts.}
\label{fig:jhtdbgeometry}
\end{figure}

Figure~\ref{fig:jhtdbgeometry}(a) exhibits the inclined, non-Gaussian core and
rare excursions characteristic of the velocity-gradient invariant plane.
Panel (b) shows why a positive mean stretching rate is not a pointwise sign
law: the median of $\sigma$ is $0.1965$, its first and 99th percentiles are
$-0.6256$ and $0.6559$, and $75.79\%$ of the samples have $P_\omega>0$.
Panel (c) recovers preferential vorticity alignment with the intermediate
strain eigenvector, with pooled mean absolute cosine $0.6613$. Panel (d) retains all $4{,}096$ production values at each
stored time. The signed-log coordinate compresses the long tails without
changing sign or rank after normalisation by the frame's root-mean-square
production.
The separated mean and median marks expose the positive tail responsible for
the sign of the mean. The positive fraction remains between $0.7498$ and
$0.7656$ over the five frames. Panel (e) gives pooled helicity--production
correlation $-0.1178$. Local velocity--vorticity alignment is consequently a
poor proxy for signed stretching in this sample. Panel (f) retains the full
distribution of the pointwise trace defect instead of displaying five
RMS values alone. The corresponding RMS
$\operatorname{tr}A$ ranges from $0.31$
to $0.37$. The returned interpolation is therefore unsuitable for exact
solenoidal integral identities, although it remains informative for
distributional geometry.
Appendix~\ref{app:statistics} specifies the pooling, quantile, fraction,
correlation, and root-mean-square estimators used for these values.

\begin{figure}[t]
\centering
\includegraphics[width=\linewidth,height=0.43\textheight,keepaspectratio]{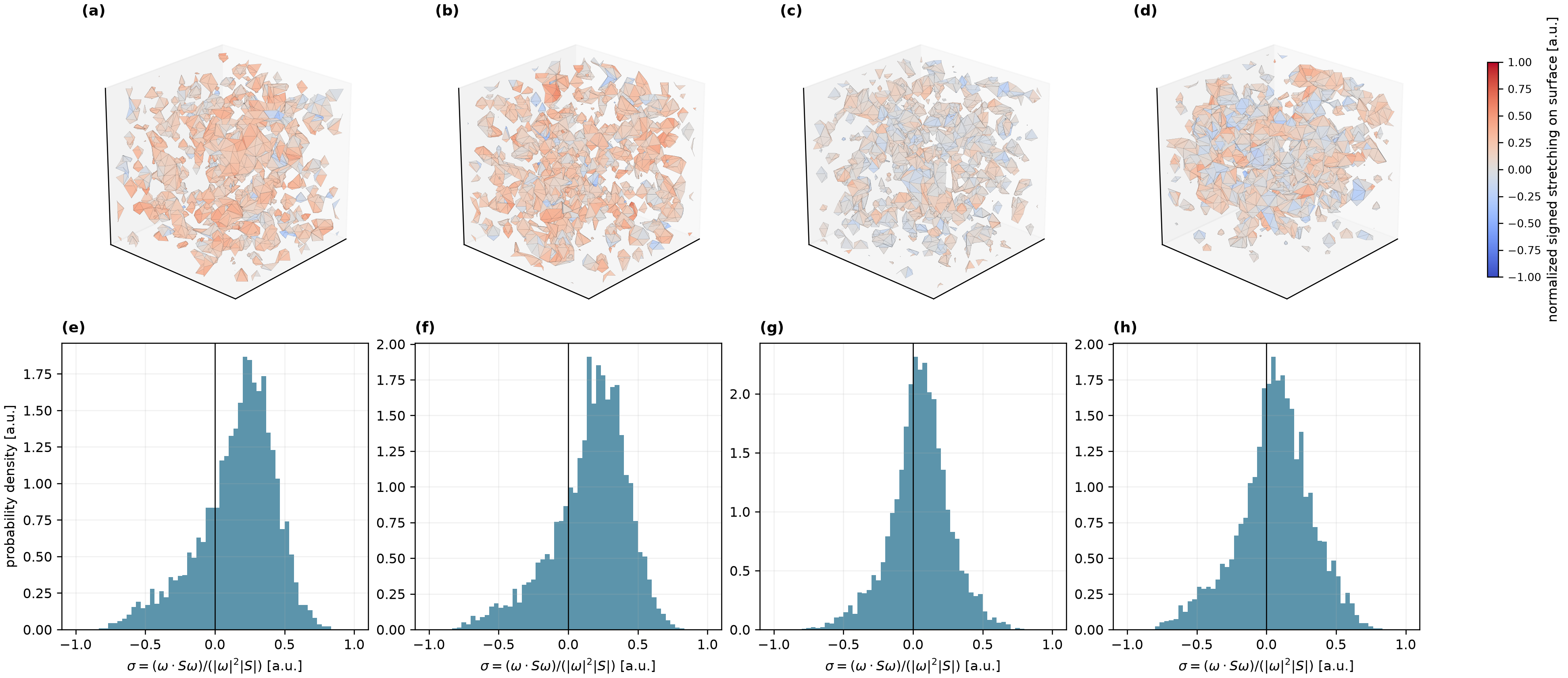}
\caption{Cross-flow structured public-DNS comparison. (a) Isotropic-$1024^3$
vorticity isosurface. (b) Isotropic-$4096^3$ vorticity isosurface. (c)
MHD-$1024^3$ vorticity isosurface. (d) Rotating--stratified-$4096^3$ vorticity
isosurface. Each surface uses the 82nd percentile of sampled $|\omega|$ and is
coloured by $\sigma\in[-1,1]$. Red denotes stretching, blue compression, and
near-white weak signed stretching. (e) Isotropic-$1024^3$ probability density
of $\sigma$. (f) Isotropic-$4096^3$ probability density of $\sigma$. (g)
MHD-$1024^3$ probability density of $\sigma$. (h) Rotating--stratified-$4096^3$
probability density of $\sigma$. Each panel derives from a $16^3$ structured
provider query.}
\label{fig:jhtdbstructured}
\end{figure}

Figure~\ref{fig:jhtdbstructured}(a) shows the sampled isotropic-$1024^3$
vorticity surface with mixed stretching signs. Panel (b) shows the
isotropic-$4096^3$ sample under the same dimensionless colour scale. Their
surface organization is visually similar despite different dimensional
vorticity amplitudes. Panel (c) shows the MHD-$1024^3$ sample, where near-zero
colours occupy more of the displayed surface. Panel (d) shows the
rotating--stratified-$4096^3$ sample with the same shift toward weak signed
stretching. Panel (e) quantifies the isotropic-$1024^3$ distribution with
$\langle\sigma\rangle=0.15665$ and positive fraction $0.7571$. Panel (f)
gives the corresponding isotropic-$4096^3$ values, $0.15661$ and $0.7627$.
Their agreement, despite an approximately $18.5$-fold difference in sampled
mean enstrophy, is descriptive rather than a resolution-convergence result
because both source fields use the same coarse query lattice. Panel (g) gives
the MHD mean $0.04809$ and positive fraction $0.6221$. Panel (h) gives the
rotating--stratified values $0.05811$ and $0.6287$. Mean intermediate-
eigenvector alignment is $0.6550$ and $0.6628$ in the two isotropic samples,
$0.8069$ in the MHD sample, and $0.7194$ in the rotating--stratified sample.
Vorticity--strain alignment
and signed production are therefore related but non-equivalent observables.
The coarse surfaces do not resolve the vortex-tube geometry of the underlying
$1024^3$ or $4096^3$ simulations. Appendix~\ref{app:statistics} defines the
cross-flow means, positive fractions, and alignment summaries.

\section{Scope and outlook}

The central advance is a family-level regularity theorem with explicit radii,
terminal times, and bounds uniform in the Galerkin cutoff. The acceptance
theorem separates the centre-independent analysis from the finite Fourier
algebra of each reference family. This separation allows new centre-specific estimates to reuse the continuum
argument whenever they meet the acceptance hypotheses.

The present applications comprise three structurally distinct centres and a
connected non-Beltrami coefficient interval. The next mathematical objective
is an algorithmic coverage theorem that converts a prescribed compact class of
finite-Fourier centres into finitely many validated parameter cells. Such a
result would replace centre-by-centre construction with an adaptive family
decomposition and would quantify the verified radius as a function of
amplitude, viscosity, spectral support, and path degree.

The numerical programme provides two complementary maps of the proved
families. The 4,096-point ensemble resolves the interior configuration space,
while the 512-case, 1,600-trajectory refinement resolves the apparent
spectral-tail boundary at higher cutoffs. Three-dimensional centre fields,
dynamic vortex surfaces, and $73{,}728$ public-DNS gradient samples connect the
analytic observables to strain, stretching, alignment, and spectral transfer.
The neural-operator study adds a targeted search layer. For one fitted seed,
physics-informed training lowers the medians of seven OOD penalties and
diagnostics and returns candidates from two centre families for resolved
solver evaluation. These layers use their natural
evidentiary roles: exact inequalities establish the theorem, Galerkin
calculations resolve finite-dimensional dynamics, neural operators prioritise
new calculations, and public DNS tests the physical interpretation across
canonical turbulence regimes.

Two extensions would broaden the impact most directly. An independent
interval-arithmetic or proof-assistant implementation would provide a
clean-room validation of the acceptance pipeline. A wider parameter campaign
would test the scaling of verified radii across multiple viscosities and
changing Fourier supports. Both extensions build on the same centre-independent
theorem and preserve the explicit separation between mathematical validation
and physical diagnostics.

\section{Conclusion}

Finite-reference-path comparison establishes global regularity on three
explicit closed balls, each of which contains an open neighbourhood, and on a
connected positive-dimensional family of initial data around genuinely
three-dimensional periodic centres. The theorem covers
continuously many amplitudes, centre coefficients, and smooth
infinite-dimensional perturbations. Parameter dependence is enclosed
analytically, perturbations are controlled in the weighted Fourier--Wiener
norm, and the Galerkin cutoff is removed through uniform estimates. Separate
arguments handle the exceptional cutoff and terminal continuation, which
makes the continuum conclusion independent of numerical extrapolation.

The computational contribution realises the centre-independent argument as a
reusable validation method. Exact fraction arithmetic, independently
implemented calculations, and targeted perturbation tests resolve the full
nonlinear Fourier output and its continuum comparison bounds. The
paired-resolution ensemble, high-resolution boundary refinement, and
three-dimensional field diagnostics locate the theorem within a broader
configuration space and show how residual production, spectral tails, and
vortex stretching evolve. For the fitted models and finite OOD set,
physics-informed operator learning lowers field, divergence, residual,
energy-balance, cutoff-sensitivity, and spectral-tail penalties while ranking
the five-checkpoint action proxy for direct solver evaluation.

The coefficient-family theorem demonstrates that one validated comparison can
cover a connected continuum of non-Beltrami centre shapes with a uniform
perturbation radius. Its specific contribution is an explicitly verified
coefficient interval, radius, and cutoff treatment within the established
a posteriori regularity paradigm. An
adaptive coverage theorem for changing supports, together with independent
formal validation, now offers the clearest route from the present result to a
general computational theory of nonlinear regularity neighbourhoods.

\section*{Data availability}

The numerical source data underlying every figure, the finite comparison
records used by the computer-assisted arguments, and a per-figure source-data
index accompany the submission package. The package includes the 384
neural-operator rollouts, held-out diagnostic arrays, 48 selected candidate
coordinates, their original 144 three-level proxy recomputations, and 192
dense-quadrature recomputations at four grid--time-step pairs. The JHTDB
observations derive from five public records: extended forced isotropic
turbulence at $1024^3$ (\texttt{isotropic1024coarse}), forced isotropic
turbulence at $4096^3$ (\texttt{isotropic4096}), forced MHD turbulence at
$1024^3$ (\texttt{mhd1024}), rotating--stratified turbulence at $4096^3$
(\texttt{rotstrat4096}), and turbulent channel flow (\texttt{channel})
\cite{JHTDBIso1024,JHTDBIso4096,JHTDBMHD,JHTDBRotStrat,JHTDBChannel}. The local
query extents, sampling rules, and dataset identifiers are included in the
accompanying data registry. No numerical simulation or public-DNS sample is
used as a premise of the regularity theorems.

\section*{Code availability}

The accompanying reproducibility package contains the exact-arithmetic
implementations, comparison records, numerical solvers, independent numerical
verification routines, trained neural-operator definitions, and machine-readable indices that link each figure to
its underlying data. A permanent repository with a digital object identifier
(DOI) will archive the software and source data before journal submission.
The archive will preserve the source version used to generate the published
figures and numerical tables.

\section*{Author contributions}

Jose Luis Lima de Jesus Silva conceived the study, developed the mathematical
analysis and computer-assisted methodology, implemented and validated the
software, designed and analysed the numerical experiments, prepared the
visualizations, and wrote the manuscript.

\section*{Competing interests}

The author declares no competing interests.

\section*{Acknowledgements}

The author acknowledges the Johns Hopkins Turbulence Databases for providing
public access to the turbulence data analysed in the physical comparison.
The broader research program from which this study emerged was initiated
within the Swedish National Infrastructure for Computing (SNIC) Small Compute
project Artificial Intelligence for Physics and Engineering, Modeling and
Simulation, Project No. SNIC 2022/22-843, conducted at Link\"oping University
under the author's principal investigatorship. This work was supported by the
Brazilian National Council for Scientific and Technological Development
(CNPq) under grant No. 445344/2024-5. The author also acknowledges financial
support provided through the Program Talentos Brasil, as Project Investigator.

\section*{Funding}

CNPq grant No. 445344/2024-5 and the Program Talentos Brasil supported this
work, as acknowledged above.
\bibliographystyle{unsrtnat}
\bibliography{references}

\clearpage
\appendix

\section{Fourier, Leray, Wiener, and energy identities}
\label{app:foundations}

This appendix supplies the first-principles calculations invoked immediately
after Equations~\eqref{eq:galerkin} and \eqref{eq:fourierB} in the main text.
For the Fourier convention stated there, differentiation of
$e^{ik\cdot x}$ gives
\begin{align}
 &\partial_j e^{ik\cdot x}=ik_j e^{ik\cdot x},
 \qquad -\Delta e^{ik\cdot x}=|k|^2e^{ik\cdot x}.
 \label{eq:detail105}
\end{align}
Consequently, the Fourier coefficient of a derivative is
$\widehat{\partial_j f}(k)=ik_j\widehat f(k)$, and the coefficient of a
product follows by multiplying the two series and collecting terms with
$p+q=k$:
\begin{align}
 &\widehat{f_j\partial_j g}(k)
 =\sum_{p+q=k}\widehat f_j(p)\,iq_j\widehat g(q).
 \label{eq:detail106}
\end{align}
Summing over $j=1,2,3$ yields
\begin{align}
 &\widehat{(f\cdot\nabla)g}(k)
 =i\sum_{p+q=k}(\widehat f(p)\cdot q)\widehat g(q).
 \label{eq:detail107}
\end{align}
For $k\ne0$, any vector $z$ decomposes uniquely into a component parallel to
$k$ and a component perpendicular to $k$:
\begin{align}
 &z=\frac{k(k\cdot z)}{|k|^2}
   +\left(I-\frac{k\otimes k}{|k|^2}\right)z.
 \label{eq:detail108}
\end{align}
The second term is divergence free because its scalar product with $k$
vanishes. This proves the Leray multiplier and, after applying it to the
preceding convolution, proves Equation~\eqref{eq:fourierB}. Since an
orthogonal projection cannot increase Euclidean length,
$|\Pp_kz|\leq|z|$.

The Wiener bilinear estimate used to obtain Equation~\eqref{eq:dini0} now
follows without an omitted convolution step. Starting from
Equation~\eqref{eq:fourierB}, applying the triangle inequality, and then
reindexing $k=p+q$ gives
\begin{align}
 \|B(f,g)\|_{\A^0}
 &\leq\sum_k\sum_{p+q=k}
       |\widehat f(p)|\,|q|\,|\widehat g(q)|\\
 &=\sum_p|\widehat f(p)|
   \sum_q|q|\,|\widehat g(q)|
 =\|f\|_{\A^0}\|g\|_{\A^1}.
 \label{eq:detail109}
\end{align}

The energy identity extends each finite Galerkin trajectory globally.
To obtain that identity, take the real $L^2$ inner product of Equation~\eqref{eq:galerkin} with $u_N$.
The projection may be removed inside this pairing because $P_N$ and $\Pp$
are self-adjoint and $u_N$ lies in their ranges. Periodicity and
$\nabla\cdot u_N=0$ imply
\begin{align}
 &\langle (u_N\cdot\nabla)u_N,u_N\rangle
 =\frac1{2(2\pi)^3}\int_{\T^3}u_N\cdot\nabla|u_N|^2\dd x
 =-\frac1{2(2\pi)^3}\int_{\T^3}(\nabla\cdot u_N)|u_N|^2\dd x=0.
 \label{eq:detail110}
\end{align}
Integration by parts also gives
$\langle Au_N,u_N\rangle=\|\nabla u_N\|_2^2$. Therefore
\begin{align}
 &\frac12\frac{d}{dt}\|u_N\|_2^2+\|\nabla u_N\|_2^2=0.
 \label{eq:detail111}
\end{align}
Only the nonlinear part of the viscous ODE preserves energy, so $\|u_N(t)\|_2\leq\|P_Nu_0\|_2$ for every time on which the
finite-dimensional solution exists. A polynomial finite-dimensional ODE can
cease to exist at finite time only if its coefficient vector becomes
unbounded. The displayed energy bound excludes that alternative and extends
$u_N$ globally.

\section{Minimization of the Laplacian residual and spectral variance}
\label{app:observable}

This section proves Equations~\eqref{eq:Rmoments} and
\eqref{eq:spectralvariance} from the main text and derives the variance bound
used in Equation~\eqref{eq:lowcutoff}. For a nonzero field $u_N$, expand the
squared residual as a quadratic polynomial in $\lambda$:
\begin{align}
 \|(A-\lambda)u_N\|_2^2
 &=\|Au_N\|_2^2-2\lambda\langle Au_N,u_N\rangle
   +\lambda^2\|u_N\|_2^2\\
 &=Z_N-2\lambda X_N+\lambda^2E_N.
 \label{eq:detail112}
\end{align}
Its derivative is $-2X_N+2\lambda E_N$, so the unique minimizer is
$\mu_N=X_N/E_N$. Substitution gives
$R_N=Z_N-X_N^2/E_N$. If $u_N=0$, then $R_N=0$ directly from its infimum
definition. The quotient formula extends continuously with value zero, but
$\mu_N$ and the probability weights below remain undefined at the zero field.

Parseval's identity gives
\begin{align}
 &E_N=\sum_k|\widehat u_N(k)|^2,
 \quad X_N=\sum_k|k|^2|\widehat u_N(k)|^2,
 \quad Z_N=\sum_k|k|^4|\widehat u_N(k)|^2.
 \label{eq:detail113}
\end{align}
With $\rho_k=|\widehat u_N(k)|^2/E_N$, the weights satisfy
$\rho_k\geq0$ and $\sum_k\rho_k=1$. Dividing the preceding moments by
$E_N$ and substituting into the residual gives
\begin{align}
 R_N
 &=E_N\left\{\sum_k\rho_k|k|^4
      -\left(\sum_k\rho_k|k|^2\right)^2\right\}\\
 &=E_N\operatorname{Var}_{\rho}(|k|^2),
 \label{eq:detail114}
\end{align}
which is Equation~\eqref{eq:spectralvariance}. A variance is zero precisely
when its random variable is constant on the support of its probability
measure. Hence $R_N=0$ exactly when all energetic modes lie on one Laplacian
sphere.

Let a scalar random variable $S$ satisfy $m\leq S\leq M$.
The pointwise inequality $(S-m)(M-S)\geq0$ gives
$\mathbb E[S^2]\leq(m+M)\mathbb E[S]-mM$. Writing
$\bar S=\mathbb E[S]$ therefore yields
\begin{align}
 \operatorname{Var}(S)
 &\leq(m+M)\bar S-mM-\bar S^2\\
 &=\frac{(M-m)^2}{4}
   -\left(\bar S-\frac{m+M}{2}\right)^2
 \leq\frac{(M-m)^2}{4}.
 \label{eq:detail115}
\end{align}
Taking $S=|k|^2$, $m=1$, and $M=3N^2$ proves the finite-spectrum estimate
used in the exceptional-cutoff argument.

The signed source and centred dissipation used in the numerical sections can
also be derived from the same three moments. To avoid confusing nonlinear
production with the Galerkin projector $P_N$, write
\begin{align}
 &W_N=\|A^{3/2}u_N\|_2^2,
 \quad \Pi_N=-\operatorname{Re}\langle B(u_N,u_N),Au_N\rangle,
 \quad \Pi_{2,N}=-\operatorname{Re}\langle B(u_N,u_N),A^2u_N\rangle.
 \label{eq:detail116}
\end{align}
For the viscosity-$\nu$ analogue of \eqref{eq:galerkin}, pairing the equation
successively with $u_N$, $Au_N$, and $A^2u_N$
produces
\begin{align}
 &E_N'=-2\nu X_N,\qquad
 X_N'=2\Pi_N-2\nu Z_N,\qquad
 Z_N'=2\Pi_{2,N}-2\nu W_N.
 \label{eq:detail117}
\end{align}
Differentiate $R_N=Z_N-X_N^2/E_N$ and substitute these three identities:
\begin{align}
 \frac12R_N'
 &=\Pi_{2,N}-\nu W_N
   -\frac{X_N}{E_N}(2\Pi_N-2\nu Z_N)
   -\nu\frac{X_N^3}{E_N^2}\\
 &=\left(\Pi_{2,N}-2\mu_N\Pi_N\right)
   -\nu\left(W_N-2\mu_NZ_N+\mu_N^2X_N\right).
 \label{eq:detail118}
\end{align}
Thus the quantities used throughout the figures are
\begin{align}
 &F_N=\Pi_{2,N}-2\mu_N\Pi_N,\qquad
 D_{R,N}=W_N-2\mu_NZ_N+\mu_N^2X_N,
 \label{eq:detail119}
\end{align}
and they satisfy
\begin{align}
 &\frac12R_N'+\nu D_{R,N}=F_N.
 \label{eq:detail120}
\end{align}
The dissipation is nonnegative because Parseval rewrites it as
\begin{align}
 &D_{R,N}=\sum_k |k|^2(|k|^2-\mu_N)^2
                  |\widehat u_N(k)|^2\geq0.
 \label{eq:detail121}
\end{align}

\section{Enstrophy, bridge action, and terminal continuation}
\label{app:continuation}

This section expands the derivation of Equation~\eqref{eq:enstrophy} and the
two action estimates used to prove Equation~\eqref{eq:action}. Pair the
Galerkin equation with $Au_N$. Since
$X_N=\langle Au_N,u_N\rangle$, differentiation and self-adjointness give
\begin{align}
 &\frac12X_N'+Z_N=-\langle B(u_N,u_N),Au_N\rangle.
 \label{eq:detail122}
\end{align}
Using the Fourier convolution before summing the final $Au_N$ factor gives
\begin{align}
 &|\langle B(u_N,u_N),Au_N\rangle|
 \leq\|u_N\|_{\A^0}\sqrt{X_NZ_N}.
 \label{eq:detail123}
\end{align}
For completeness, the Fourier series gives
$\|u_N\|_\infty\leq\sum_k|\widehat u_N(k)|=\|u_N\|_{\A^0}$.
Pointwise Cauchy--Schwarz first yields
$|(u_N\cdot\nabla)u_N|\leq|u_N|\,|\nabla u_N|_F$.
After integration and an $L^2$ Cauchy--Schwarz estimate,
\begin{align}
 |\langle B(u_N,u_N),Au_N\rangle|
 &\leq\|(u_N\cdot\nabla)u_N\|_2\|Au_N\|_2\\
 &\leq\|u_N\|_\infty\|\nabla u_N\|_2\|Au_N\|_2
 \leq\|u_N\|_{\A^0}\sqrt{X_NZ_N}.
 \label{eq:enstrophyproductdetail}
\end{align}
Multiplying this inequality by two and
using $2ab\leq a^2+b^2$ with
$a=Z_N^{1/2}$ and
$b=\|u_N\|_{\A^0}X_N^{1/2}$ gives
\begin{align}
 &X_N'+2Z_N
 \leq2\|u_N\|_{\A^0}\sqrt{X_NZ_N}
 \leq Z_N+\|u_N\|_{\A^0}^2X_N.
 \label{eq:detail125}
\end{align}
After subtracting $Z_N$, this is the first inequality in
Equation~\eqref{eq:enstrophy}. The second follows mode by mode because every
nonzero integer mode has $|k|^4\geq|k|^2$.

For the viscosity-$\nu$ setting of Theorem~\ref{thm:acceptance}, the same
calculation with $\|u_N(t)\|_{\A^0}\leq M(t)$ gives
\begin{align}
 &X_N'+\nu Z_N\leq\frac{M(t)^2}{\nu}X_N.
 \label{eq:detail126}
\end{align}
If $X_N(0)\leq X_0$ and $S=\int_0^TM(t)^2\dd t$, Gronwall's inequality yields
\begin{align}
 &\sup_{t\leq T}X_N(t)\leq X_0e^{S/\nu}.
 \label{eq:detail127}
\end{align}
Integrating the differential inequality and using this supremum gives
\begin{align}
 &\int_0^TZ_N(t)\dd t
 \leq\frac{X_0}{\nu}
 +\frac{SX_0e^{S/\nu}}{\nu^2}=:K_{\rm bridge}.
 \label{eq:detail128}
\end{align}
Consequently $R_N\leq Z_N$ and H\"older's inequality give the optional
dimensionless bridge-action bound
\begin{align}
 &\nu^{-1/3}\int_0^TR_N(t)^{2/3}\dd t
 \leq \nu^{-1/3}T^{1/3}K_{\rm bridge}^{2/3}.
 \label{eq:detail129}
\end{align}
The tail and exceptional-cutoff contributions remain instance-specific. This
is why Theorem~\ref{thm:acceptance} does not assert an action bound by itself.

On $[0,T_*]$, write $m(t)=\|u_N(t)\|_{\A^0}^2$. Dropping the nonnegative
$Z_N$ term and multiplying by
$\exp(-\int_0^tm(s)\dd s)$ gives
\begin{align}
 &X_N(t)\leq X_N(0)\exp\left(\int_0^tm(s)\dd s\right).
 \label{eq:detail130}
\end{align}
Because $X_N(0)<1$ and the main-text bounds give
$\int_0^{T_*}m<8$, one obtains $X_N(T_*)<e^8$. For the integrated
dissipation, set $g(t)=\exp(-\int_0^tm(s)\dd s)$. The full inequality gives
$(gX_N)'+gZ_N\leq0$, and therefore
\begin{align}
 &\int_0^{T_*}g(t)Z_N(t)\dd t
 \leq X_N(0)-g(T_*)X_N(T_*)<1.
 \label{eq:detail131}
\end{align}
Since $\int_0^{T_*}m<8$, one has $g(t)>e^{-8}$ throughout the bridge, and
hence $\int_0^{T_*}Z_N<e^8$. Since $R_N\leq Z_N$, Hölder's inequality with
exponents $3/2$ and $3$ gives
\begin{align}
 \int_0^{T_*}R_N^{2/3}\dd t
 &\leq\int_0^{T_*}Z_N^{2/3}\dd t\\
 &\leq\left(\int_0^{T_*}Z_N\dd t\right)^{2/3}
       \left(\int_0^{T_*}1\dd t\right)^{1/3}<418.
 \label{eq:detail132}
\end{align}

For $t\geq T_*$, Appendix~\ref{app:comparison} proves that
$\|u_N(t)\|_{\A^0}\leq c=99/100$. Hence
$X_N'+dZ_N\leq0$ with $d=1-c^2$. Put $s=t-T_*$ and
$w(s)=e^{ds/2}$. Since $X_N\leq Z_N$,
\begin{align}
 (wX_N)'&=wX_N'+\frac d2wX_N\\
 &\leq-dwZ_N+\frac d2wX_N
 \leq-\frac d2wZ_N.
 \label{eq:detail133}
\end{align}
Integration from $s=0$ to infinity and nonnegativity of $wX_N$ yield
\begin{align}
 &\int_{T_*}^{\infty}w(t-T_*)Z_N(t)\dd t
 \leq\frac{2X_N(T_*)}{d}.
 \label{eq:detail134}
\end{align}
Finally write $Z_N^{2/3}=(wZ_N)^{2/3}w^{-2/3}$ and apply Hölder with
exponents $3/2$ and $3$:
\begin{align}
 \int_{T_*}^{\infty}Z_N^{2/3}\dd t
 &\leq\left(\int_{T_*}^{\infty}wZ_N\dd t\right)^{2/3}
       \left(\int_0^\infty w^{-2}\dd s\right)^{1/3}\\
 &\leq\left(\frac{2X_N(T_*)}{d}\right)^{2/3}
       \left(\frac1d\right)^{1/3}
 =\frac{2^{2/3}X_N(T_*)^{2/3}}d.
 \label{eq:detail135}
\end{align}
This is the weighted tail estimate stated in the main text.
The numerical constants require only rational comparisons and one exponential
enclosure. The triangle inequality and $\|U\|_{\dot H^1}=\sqrt{3/2}<5/4$ give
\begin{align}
 X_N(0)^{1/2}&\leq\frac7{10}\sqrt{\frac32}+\frac1{100}
 <\frac78+\frac1{100}<1,\\
 \int_0^{T_*}m(t)\dd t&<\left(\frac{17}{8}\right)^2\frac{27}{16}<8.
 \label{eq:bridgeconstantsdetail}
\end{align}
The exponential enclosure in Appendix~\ref{app:slabalgorithm} gives
$e^8<3000$, and $3000^2<209^3$. Therefore $e^{16/3}<209$,
$T_*^{1/3}<2$, and $2^{2/3}<2$. Both action coefficients are consequently
below $2(209)=418$. Finally,
\begin{align}
 418+\frac{418}{199/10000}
 &=\frac{4263182}{199}<22000,
 \label{eq:finalactionconstant}
\end{align}
because $4263182<22000(199)=4378000$. These steps recover the constant
in Theorem~\ref{thm:main} without rounded decimal estimates.

\section{Cutoff removal, smoothness, uniqueness, and scaling}
\label{app:pdepassage}

This appendix expands the continuum passage and the two rescalings stated at
the end of the mathematical methodology. The uniform estimates provide
$u_N$ bounded in $L^\infty(0,T,H^1)$ and in $L^2(0,T,H^2)$. In three
dimensions, $H^2\hookrightarrow L^\infty$, so
\begin{align}
 &\|(u_N\cdot\nabla)u_N\|_2
 \leq\|u_N\|_\infty\|\nabla u_N\|_2
 \leq C\|u_N\|_{H^2}\|u_N\|_{H^1}.
 \label{eq:detail138}
\end{align}
The right-hand side belongs to $L^2(0,T)$ because the first Sobolev norm is
square integrable in time and the second is uniformly bounded. Equation
\eqref{eq:galerkin} therefore bounds $\partial_tu_N$ in
$L^2(0,T,L^2)$. The compact embedding $H^2\Subset H^1$ and the continuous
embedding $H^1\hookrightarrow L^2$ meet the hypotheses of the Aubin--Lions
lemma. A subsequence consequently converges strongly in $L^2(0,T,H^1)$ and
weakly in the two bounded spaces.

To pass the quadratic term, write
\begin{align}
 &B(u_N,u_N)-B(u,u)=B(u_N-u,u_N)+B(u,u_N-u).
 \label{eq:detail139}
\end{align}
Strong $L^2_tH^1_x$ convergence, the uniform $L^\infty_tH^1_x$ bound, and
the Sobolev embedding $H^1\hookrightarrow L^6$ make both terms converge to
zero in distributions (indeed in $L^1_tH^{-1}_x$). The linear terms pass by
weak convergence, and $P_Nu_0\to u_0$ in $H^1$. Thus the limit solves the
periodic equation with the stated initial datum. Explicitly, for a test field
$\phi\in H^1$, integration by parts and H\"older give
\begin{align}
 |\langle B(f,g),\phi\rangle|
 &\leq\|f\|_3\|g\|_6\|\nabla\phi\|_2
 \leq C\|f\|_{H^1}\|g\|_{H^1}\|\phi\|_{H^1}.
 \label{eq:nonlinearlimitdetail}
\end{align}
Apply this bound to the two differences above and use time Cauchy--Schwarz.
The resulting $L^1_tH^{-1}_x$ bound tends to zero. For a fixed smooth test
field, $P_N\phi\to\phi$ in every Sobolev norm, which also removes the outer
Galerkin projection in the weak formulation.

For every fixed Fourier mode, the equation and the preceding bounds make
$\widehat u_N(k,t)$ equicontinuous in time. In fact,
\begin{align}
 |\widehat u_N(k,t)-\widehat u_N(k,s)|
 &\leq |t-s|^{1/2}\|\partial_tu_N\|_{L^2(0,T;L^2)},
 \label{eq:fouriertimetrace}
\end{align}
uniformly in $N$. A diagonal Arzel\`a--Ascoli
argument gives coefficientwise convergence at $T_*$. Fatou's lemma then
implies
\begin{align}
 &\sum_k|\widehat u(k,T_*)|
 \leq\liminf_{N\to\infty}\sum_k|\widehat u_N(k,T_*)|
 \leq\frac{99}{100}.
 \label{eq:detail142}
\end{align}
Appendix~\ref{app:comparison} propagates this Wiener smallness. The uniform
$H^1$ and integrated $H^2$ bounds prevent the strong-solution continuation
criterion from failing on any finite interval. Differentiating the equation
and applying the heat semigroup successively then yields all higher Sobolev
norms for positive time. Smooth initial data supply smoothness at time zero.
If two strong solutions have the same datum, their difference $w$ obeys the
energy identity obtained by subtracting their equations. The divergence-free
transport term vanishes after integration by parts, while H\"older and Sobolev
interpolation give
\begin{align}
 &|\langle B(w,u),w\rangle|
 \leq \|\nabla u\|_2\|w\|_3\|w\|_6
 \leq C\|\nabla u\|_2\|w\|_2^{1/2}\|\nabla w\|_2^{3/2}.
 \label{eq:detail143}
\end{align}
Young's inequality spends one half of the unit viscosity and yields
\begin{align}
 &\frac12\frac{d}{dt}\|w\|_2^2+\frac12\|\nabla w\|_2^2
 \leq C\|\nabla u\|_2^4\|w\|_2^2,
 \label{eq:detail144}
\end{align}
on finite intervals. The Young exponents are $4/3$ and $4$, applied to
$\|\nabla w\|_2^{3/2}$ and
$C\|\nabla u\|_2\|w\|_2^{1/2}$ after rescaling to allocate one half of the
viscosity. Since $\|\nabla u\|_2$ is bounded in time, its fourth power is
integrable on every finite interval. Gronwall's lemma and $w(0)=0$ give $w=0$,
establishing uniqueness among strong solutions and removing subsequence
dependence. The same estimate gives weak--strong uniqueness for energy
solutions through their relative-energy inequality.

For viscosity $\nu$, set $\tau=\nu t$ and
$u(t,x)=\nu v(\tau,x)$. Then
\begin{align}
 &\partial_tu=\nu^2\partial_\tau v,
 \qquad \nu Au=\nu^2Av,
 \qquad B(u,u)=\nu^2B(v,v),
 \label{eq:detail145}
\end{align}
so the unit-viscosity equation for $v$ becomes the viscosity-$\nu$ equation
for $u$. Since $R_u(t)=\nu^2R_v(\tau)$ and $dt=d\tau/\nu$,
\begin{align}
 &\int_0^\infty R_u(t)^{2/3}\dd t
 =\nu^{4/3}\int_0^\infty R_v(\tau)^{2/3}\frac{d\tau}{\nu}
 =\nu^{1/3}\int_0^\infty R_v(\tau)^{2/3}\dd\tau.
 \label{eq:detail146}
\end{align}
On the torus of side $2\pi L$, let $y=x/L$, $\tau=\nu t/L^2$, and
$u(t,x)=(\nu/L)v(\tau,y)$. Then
$(A_x-\lambda_x)u=(\nu/L^3)(A_y-\lambda_y)v$ when
$\lambda_x=\lambda_y/L^2$. With normalised spatial measures this gives
$R_u=(\nu^2/L^6)R_v$ and $dt=(L^2/\nu)d\tau$, hence the action factor
$\nu^{1/3}/L^2$ stated in the main text.

\section{Analytic field geometry and public-DNS invariants}
\label{app:physicalidentities}

This section derives Equations~\eqref{eq:omegaU}--\eqref{eq:nonlinearU} and
explains the contractions in Equations~\eqref{eq:jhtdbstrain}--
\eqref{eq:jhtdbgeometry}. For
$u=(a\sin y,a\sin z,a\sin x)$, the curl is
\begin{align}
 \omega_1&=\partial_yu_3-\partial_zu_2=-a\cos z,\\
 \omega_2&=\partial_zu_1-\partial_xu_3=-a\cos x,\\
 \omega_3&=\partial_xu_2-\partial_yu_1=-a\cos y.
 \label{eq:detail147}
\end{align}
All diagonal velocity derivatives vanish. Symmetrizing the remaining ones
gives $S_{12}=a\cos y/2$, $S_{13}=a\cos x/2$, and
$S_{23}=a\cos z/2$. Because both off-diagonal entries contribute to the
Frobenius norm,
\begin{align}
 &|S|_F^2=2(S_{12}^2+S_{13}^2+S_{23}^2)
 =\frac{a^2}{2}(\cos^2x+\cos^2y+\cos^2z)
 =\frac12|\omega|^2.
 \label{eq:detail148}
\end{align}
The diagonal entries of $S$ vanish, so its quadratic form contains the three
off-diagonal pairs:
\begin{align}
 \omega\cdot S\omega
 &=2(S_{12}\omega_1\omega_2+S_{13}\omega_1\omega_3
       +S_{23}\omega_2\omega_3)\\
 &=3a^3\cos x\cos y\cos z.
 \label{eq:detail149}
\end{align}
Finally, direct substitution into
$(u\cdot\nabla)u_j=\sum_i u_i\partial_i u_j$ produces the three components
in Equation~\eqref{eq:nonlinearU}.

The norm comparisons in Equation~\eqref{eq:normposition} follow from the same
six Fourier modes. Each sine contributes two coefficients of magnitude
$1/2$, all on $|k|=1$, so
\begin{align}
 &\|aU\|_{\A^0}=\|aU\|_{\A^2}=6\frac a2=3a,
 \qquad
 \|aU\|_{\dot H^{1/2}}^2=6\frac{a^2}{4}=\frac32a^2.
 \label{eq:detail150}
\end{align}
The three sine factors can attain magnitude one simultaneously, giving
$\|aU\|_\infty=a\sqrt3$. For a mean-zero perturbation, $|k|\geq1$ on every
active mode. Hence
\begin{align}
 &\|h\|_{\A^0}\leq\|h\|_{\A^2},
 \quad \|h\|_\infty\leq\sum_k|\widehat h(k)|=\|h\|_{\A^0},
 \label{eq:detail151}
\end{align}
and the $\ell^2\leq\ell^1$ inequality gives
\begin{align}
 &\|h\|_{\dot H^{1/2}}
 =\left(\sum_k|k|\,|\widehat h(k)|^2\right)^{1/2}
 \leq\sum_k|k|^{1/2}|\widehat h(k)|
 \leq\|h\|_{\A^2}.
 \label{eq:detail152}
\end{align}
Combining these estimates with the triangle inequality yields every bound in
Equation~\eqref{eq:normposition}.

For a general gradient tensor $A=\nabla u$, its symmetric part
$S=(A+A^\top)/2$ is the infinitesimal strain, while its antisymmetric
part contains the curl components displayed in Equation~\eqref{eq:jhtdbstrain}.
The scalar $P_\omega=\omega\cdot S\omega$ is positive when the strain
increases vorticity magnitude and negative when it compresses it. Dividing by
$|\omega|^2|S|_F$ removes the local magnitude and leaves the signed geometric
factor $\sigma$. Indeed, Cauchy--Schwarz and the operator--Frobenius norm
inequality give
\begin{align}
 &|\omega\cdot S\omega|
 \leq |\omega|\,|S\omega|
 \leq |\omega|^2\|S\|_{\mathrm{op}}
 \leq |\omega|^2|S|_F,
 \label{eq:detail153}
\end{align}
which proves $|\sigma|\leq1$ whenever the denominator is nonzero. For the
analytic centre, writing $r_i=|\cos x_i|$ gives
\begin{align}
 &|\sigma_\omega|
 =\frac{3\sqrt2\,r_1r_2r_3}
 {(r_1^2+r_2^2+r_3^2)^{3/2}}
 \leq\sqrt{\frac23},
 \label{eq:detail154}
\end{align}
where the final inequality follows from
$r_1r_2r_3\leq[(r_1^2+r_2^2+r_3^2)/3]^{3/2}$. Equality occurs when
$r_1=r_2=r_3\ne0$, establishing the analytic extrema quoted in the Results.
Under an orthogonal change of coordinates
$A\mapsto OAO^\top$, cyclicity of the trace gives
$\operatorname{tr}[(OAO^\top)^m]=\operatorname{tr}(A^m)$. Thus the
$Q$ and $R$ in Equation~\eqref{eq:jhtdbgeometry} are coordinate invariants.
If $S=E\operatorname{diag}(\lambda_1,\lambda_2,\lambda_3)E^\top$, the columns
$E_{:i}=e_i$ are the ordered strain eigenvectors and the alignment is
\begin{align}
 &c_i=|\widehat\omega\cdot e_i|
 =\left|\sum_{j=1}^{3}\widehat\omega_jE_{ji}\right|.
 \label{eq:detail155}
\end{align}
Vorticity is an axial vector. An orthogonal coordinate change sends
$(\widehat\omega,E)$ to $(\det(O)O\widehat\omega,OE)$, so the absolute
contraction remains unchanged. At a repeated strain eigenvalue, individual
eigenvectors within its eigenspace are not unique. Only the projection onto
that entire eigenspace is basis independent. The numerical
implementation therefore contracts vorticity with eigenvector columns, not
matrix rows. Appendix~\ref{app:tensorexpansion} expands these contractions
and distinguishes tensor invariants from the global spectral residual.
They are diagnostic definitions and are not used to prove
Theorem~\ref{thm:main}.

\section{Dilated-eigenfield asymptotics}
\label{app:eigenasymptotic}

This appendix provides the calculation behind
Equations~\eqref{eq:Rasymptotic}--\eqref{eq:fluxasymptotic}. Under
$y=nx$ and $\tau=\nu n^2t$, one spatial derivative contributes a factor
$n$ and one time derivative contributes $\nu n^2$. Substitution into the
viscosity-$\nu$ equation, followed by division by $\nu n^2$, gives
\begin{align}
 &\partial_\tau v_\varepsilon+Av_\varepsilon
 =-\varepsilon B(v_\varepsilon,v_\varepsilon),
 \qquad \varepsilon=(\nu n)^{-1}.
 \label{eq:detail156}
\end{align}
The first nonlinear Duhamel term is
\begin{align}
 &-\varepsilon\int_0^\tau
 e^{-(\tau-s)A}B(e^{-sA}U,e^{-sA}U)\dd s.
 \label{eq:detail157}
\end{align}
Since $AU=U$ and $B(U,U)=V$ with $AV=2V$, the integrand equals
$e^{-2(\tau-s)}e^{-2s}V=e^{-2\tau}V$. Its integral is therefore
$-\varepsilon\tau e^{-2\tau}V$. The following estimate makes the remainder uniform
for all $\tau\geq0$, which is required before integrating the asymptotic to
infinite time. Fix an integer $s\geq4$. The spectral gap on mean-zero fields,
the $H^s$ product estimate, and one derivative of heat smoothing give
\begin{align}
 &\|e^{-rA}B(f,g)\|_{H^s}
 \leq C_s(1+r^{-1/2})e^{-r}\|f\|_{H^s}\|g\|_{H^s},
 \qquad r>0.
 \label{eq:weightedheatbilinear}
\end{align}
Introduce the weighted Banach space of continuous mean-zero solenoidal
$H^s$-valued paths
\begin{align}
 &X_s=\left\{v:\|v\|_{X_s}:=\sup_{\tau\geq0}
 e^\tau\|v(\tau)\|_{H^s}<\infty\right\}.
 \label{eq:detail159}
\end{align}
If $f,g\in X_s$, multiplication of the Duhamel estimate by $e^\tau$ yields
\begin{align}
 &e^\tau\left\|\int_0^\tau e^{-(\tau-r)A}B(f(r),g(r))\dd r\right\|_{H^s}
 \notag\\
 &\quad\leq C_s\|f\|_{X_s}\|g\|_{X_s}
 \int_0^\tau[1+(\tau-r)^{-1/2}]e^{-r}\dd r.
 \label{eq:weightedduhamel}
\end{align}
The final integral is uniformly bounded. Its nonsingular part is at most one.
For $0<\tau\leq1$, the singular part is at most $2\sqrt\tau$. For $\tau>1$,
split at $r=\tau/2$. On the first half, $(\tau-r)^{-1/2}\leq
(\tau/2)^{-1/2}$ and $\int_0^{\tau/2}e^{-r}\dd r\leq1$. On the second half,
set $q=\tau-r$ and use $e^{-r}\leq e^{-\tau/2}$ to obtain the bound
$2e^{-\tau/2}\sqrt{\tau/2}$. Consequently the mild map
\begin{align}
 &\mathcal T_\varepsilon(v)(\tau)=e^{-\tau A}U
 -\varepsilon\int_0^\tau e^{-(\tau-r)A}B(v(r),v(r))\dd r
 \label{eq:detail161}
\end{align}
maps a fixed ball of $X_s$ into itself and has Lipschitz constant
$C_s|\varepsilon|$. For $|\varepsilon|\leq\varepsilon_s$ it is therefore a
contraction, and its fixed point satisfies
\begin{align}
 &\|v_\varepsilon-e^{-\tau}U\|_{X_s}\leq C_s|\varepsilon|.
 \label{eq:firstweightedremainder}
\end{align}
Let $b(\tau)=\tau e^{-2\tau}V$. Subtracting the linear term and the exact first
Duhamel coefficient gives
\begin{align}
 v_\varepsilon-e^{-\tau}U+\varepsilon b
 =-\varepsilon\int_0^\tau e^{-(\tau-r)A}
 \{B(v_\varepsilon,v_\varepsilon)-B(e^{-r}U,e^{-r}U)\}\dd r.
 \label{eq:detail163}
\end{align}
Polarizing the difference, then applying
Equations~\eqref{eq:weightedduhamel} and
\eqref{eq:firstweightedremainder}, proves
\begin{align}
 &\sup_{\tau\geq0}e^\tau
 \|v_\varepsilon(\tau)-e^{-\tau}U
 +\varepsilon\tau e^{-2\tau}V\|_{H^s}
 \leq C_s\varepsilon^2.
 \label{eq:secondweightedremainder}
\end{align}
Thus the remaining Duhamel terms are
$O_{H^s}(\varepsilon^2e^{-\tau})$ uniformly on the half-line. Because the
first derivative of the spectral variance functional vanishes at the
monochromatic state $U$, substituting
\eqref{eq:secondweightedremainder} into the exact moment formula produces a
cubic, rather than linear, residual error.

The leading two eigenspaces are orthogonal, and
$\|U\|_2^2=3/2$, $\|V\|_2^2=3/4$. Relative to the eigenvalue-one shell, the
eigenvalue-two coefficient has displacement $2-1=1$. The variance identity
of Appendix~\ref{app:observable} consequently gives
\begin{align}
 &R(v_\varepsilon(\tau))
 =(2-1)^2\|V\|_2^2
   \varepsilon^2\tau^2e^{-4\tau}
   +O(\varepsilon^3e^{-2\tau}),
 \label{eq:detail165}
\end{align}
which is Equation~\eqref{eq:Rasymptotic}. Spatial dilation multiplies every
Laplacian eigenvalue by $n^2$, hence multiplies $R$ by $n^4$, while
$dt=d\tau/(\nu n^2)$. Therefore
\begin{align}
 \int_0^\infty R(u_n(t))^{2/3}\dd t
 &=\frac{n^{8/3}}{\nu n^2}
   \left(\frac34\right)^{2/3}\varepsilon^{4/3}
   \int_0^\infty\tau^{4/3}e^{-8\tau/3}\dd\tau
   +\text{remainder}.
 \label{eq:detail166}
\end{align}
The Gamma identity
$\int_0^\infty\tau^{p-1}e^{-q\tau}\dd\tau=\Gamma(p)q^{-p}$,
with $p=7/3$ and $q=8/3$, gives the constant $C_R$ displayed in
Equation~\eqref{eq:actionasymptotic}. The elementary inequality
$|x^{2/3}-y^{2/3}|\leq|x-y|^{2/3}$ for $x,y\geq0$ turns the cubic residual
error into a relative $O(\varepsilon^{2/3})$ error after integration. Finally,
$n^{2/3}\nu^{-1}\varepsilon^{4/3}=\nu^{-7/3}n^{-2/3}$, which establishes
the stated scaling.

For the signed source, integrate the balance derived in
Appendix~\ref{app:observable}. Both endpoint residuals vanish to the required
order, so $\int\mathcal F\dd t=\nu\int D_R\dd t$. On the two leading shells,
the centred dissipation has the expansion
\begin{align}
 &D_R(v_\varepsilon(\tau))
 =2(2-1)^2\|V\|_2^2\varepsilon^2\tau^2e^{-4\tau}
  +O(\varepsilon^3e^{-2\tau})
 =\frac32\varepsilon^2\tau^2e^{-4\tau}+O(\varepsilon^3e^{-2\tau}).
 \label{eq:detail167}
\end{align}
Spatial dilation multiplies $D_R$ by $n^6$. Using again
$dt=d\tau/(\nu n^2)$ and
$\int_0^\infty\tau^2e^{-4\tau}\dd\tau=2/4^3=1/32$ gives
\begin{align}
 \nu\int_0^\infty D_R(u_n(t))\dd t
 &=n^4\frac32\varepsilon^2\frac1{32}
   [1+O(\varepsilon)]\\
 &=\frac{3n^2}{64\nu^2}[1+O((\nu n)^{-1})],
 \label{eq:detail168}
\end{align}
which proves Equation~\eqref{eq:fluxasymptotic}.

\section{Exact Fourier algebra and radical bounds}
\label{app:algebra}

Exact heat integrals and Fourier coefficients determine the path and
majorants in Equations~\eqref{eq:path}--\eqref{eq:r}. After the cyclic-shear
calculation, Sections~\ref{app:abcalgebra}, \ref{app:tg3algebra}, and
\ref{app:boxalgebra} derive the remaining centre identities, and
Section~\ref{app:bernsteindetail} converts coefficient dependence into
whole-interval bounds. The two
heat convolutions can be evaluated without dividing by a vanishing spectral
gap. For the eigenvalue-one output,
\begin{align}
 K_1(t)
 &=e^{-t}\int_0^t s e^{-2s}\dd s\\
 &=e^{-t}\left[-\frac{(2s+1)e^{-2s}}4\right]_{s=0}^{s=t}
 =\frac{e^{-t}-(1+2t)e^{-3t}}4.
 \label{eq:detail169}
\end{align}
For the eigenvalue-five output,
\begin{align}
 K_5(t)
 &=e^{-5t}\int_0^t s e^{2s}\dd s\\
 &=e^{-5t}\left[\frac{(2s-1)e^{2s}}4\right]_{s=0}^{s=t}
 =\frac{(2t-1)e^{-3t}+e^{-5t}}4.
 \label{eq:detail170}
\end{align}
Both kernels are nonnegative because their defining integrands are
nonnegative. Differentiating the integral forms gives
$K_\ell'+\ell K_\ell=te^{-3t}$ and $K_\ell(0)=0$.

The spatial calculation starts from the six coefficients
\begin{align}
 &\widehat U(0,\pm1,0)=\left(\frac{\pm1}{2i},0,0\right),\quad
 \widehat U(0,0,\pm1)=\left(0,\frac{\pm1}{2i},0\right),
 \label{eq:detail171}
\end{align}
\begin{align}
 &\widehat U(\pm1,0,0)=\left(0,0,\frac{\pm1}{2i}\right).
 \label{eq:detail172}
\end{align}
Starting from these coefficients, the finite convolution repeatedly applies
\eqref{eq:fourierB}. For each output $k$, it
forms every ordered pair $p+q=k$, evaluates
$i(\widehat f(p)\cdot q)\widehat g(q)$, and applies
$I-k\otimes k/|k|^2$. Equal modes are summed before any norm is taken. Exact conjugate symmetry and $k\cdot\widehat f(k)=0$ verify reality and
solenoidality of the resulting fields.

The resulting exact calculation table is
\begin{center}\small
\begin{tabular}{@{}llll@{}}
\toprule
Field & Definition & spectral support used & exact $\A^0$ norm\\
\midrule
$U$ & seed & $|k|^2=1$ & $3$\\
$V$ & $B(U,U)$ & $|k|^2=2$ & $3$\\
$C$ & $B(U,V)+B(V,U)$ & $|k|^2=1,5$ & $3$\\
$H$ & $C+U/2$ & $|k|^2=5$ & $3/2$\\
$D$ & $B(V,V)$ & finite exact output & $\sqrt3$\\
$E$ & $B(U,H)+B(H,U)$ & finite exact output & $(3+\sqrt{21})/2$\\
$F$ & $B(V,H)+B(H,V)$ & finite exact output & $9\sqrt{22}/44+9\sqrt5/20+\sqrt{29}/4$\\
$G$ & $B(H,H)$ & finite exact output & $2\sqrt{21}/7$\\
\bottomrule
\end{tabular}
\end{center}

Rational upper bounds for the radicals preserve the exact inequalities in
\eqref{eq:r}:
\begin{align}
 &\sqrt3<\frac74,\qquad \sqrt{21}<\frac{37}{8},\qquad
 \sqrt{22}<\frac{19}{4},\qquad \sqrt5<\frac94,\qquad
 \sqrt{29}<\frac{11}{2}.
 \label{eq:detail173}
\end{align}
Squaring proves each inequality. In particular,
\begin{align}
 &\frac9{44}\frac{19}{4}+\frac9{20}\frac94+\frac14\frac{11}{2}
 <\frac72,\qquad
 \frac27\frac{37}{8}<\frac43.
 \label{eq:detail174}
\end{align}
Because $H$ lies on $|k|^2=5$,
$\|H\|_{\A^1}=\sqrt5\|H\|_{\A^0}<27/8$. Because $V$ lies on
$|k|^2=2$, $\|V\|_{\A^1}=3\sqrt2<9/2$. These are precisely the
coefficients of \eqref{eq:A1}. This algorithmic description and the displayed
seed coefficients determine the full finite calculation uniquely. The
independent implementation reconstructs it rather than reading stored field
values.
Appendix~\ref{app:pathexpansion} lists the coefficient multiplicities behind
every radical in this table and expands all nine bilinear terms before
collecting the residual. Thus the norm bounds can be checked without
inferring cancellations from the final expression.

\subsection{Exact equal-ABC identities}
\label{app:abcalgebra}

This subsection supplies the algebra used in Theorem~\ref{thm:abc}. Direct
differentiation gives
\begin{align}
 &\nabla\times U_{\rm ABC}
 =(\partial_yU_3-\partial_zU_2,
   \partial_zU_1-\partial_xU_3,
   \partial_xU_2-\partial_yU_1)
 =U_{\rm ABC}.
 \label{eq:detail175}
\end{align}
Each sine and cosine has Laplacian eigenvalue one, so
$AU_{\rm ABC}=U_{\rm ABC}$. The vector identity
\begin{align}
 &(u\cdot\nabla)u
 =\nabla\frac{|u|^2}{2}-u\times(\nabla\times u)
 \label{eq:detail176}
\end{align}
then gives
$(U_{\rm ABC}\cdot\nabla)U_{\rm ABC}=\nabla(|U_{\rm ABC}|^2/2)$.
The Leray projection annihilates this gradient, proving
$B(U_{\rm ABC},U_{\rm ABC})=0$.

The Fourier support consists of the six modes
$\{\pm e_1,\pm e_2,\pm e_3\}$. At each mode, two orthogonal component
coefficients have magnitude $1/2$, so the Euclidean magnitude of the vector
coefficient is $1/\sqrt2$. Since every active mode has $|k|=1$,
\begin{align}
 &\|U_{\rm ABC}\|_{\A^0}
 =\|U_{\rm ABC}\|_{\A^1}
 =6\frac1{\sqrt2}=3\sqrt2.
 \label{eq:detail177}
\end{align}
Thus $v(t)=ae^{-t}U_{\rm ABC}$ satisfies
$v_t+Av+B(v,v)=0$ exactly, and every cubical cutoff $N\geq1$ contains its
full support. These identities establish the path and norm claims used in
Theorem~\ref{thm:abc}.

\subsection{Exact TG3 identities}
\label{app:tg3algebra}

For signs $s_x,s_y,s_z\in\{-1,1\}$, the eight coefficients of the TG3
centre are
\begin{align}
 &\widehat U_{\rm TG3}(s_x,s_y,s_z)
 =\frac{(s_x,s_y,-2s_z)}{8i}.
 \label{eq:detail178}
\end{align}
Their scalar product with $(s_x,s_y,s_z)$ is zero, proving solenoidality.
Reversal of all signs gives complex conjugation and proves reality. Every mode
has $|k|^2=3$, so $AU_{\rm TG3}=3U_{\rm TG3}$. Each coefficient has
Euclidean magnitude $\sqrt6/8$, and therefore
\begin{align}
 &\|U_{\rm TG3}\|_{\A^0}=8\frac{\sqrt6}{8}=\sqrt6,
 \qquad
 \|U_{\rm TG3}\|_{\A^1}=\sqrt3\sqrt6=3\sqrt2.
 \label{eq:detail179}
\end{align}

Applying \eqref{eq:fourierB} to all ordered pairs gives eight nonzero output
modes $(\pm2,0,\pm2)$ and $(0,\pm2,\pm2)$. At each output, the two nonzero
components have magnitude $3/32$ with signs fixed by the wave vector. For
example,
\begin{align}
 &\widehat V_{\rm TG3}(-2,0,-2)
 =\left(-\frac{3i}{32},0,\frac{3i}{32}\right).
 \label{eq:detail180}
\end{align}
Every output coefficient therefore has magnitude $3\sqrt2/32$, whence
\begin{align}
 &V_{\rm TG3}\ne0,
 \qquad \|V_{\rm TG3}\|_{\A^0}
 =8\frac{3\sqrt2}{32}=\frac{3\sqrt2}{4}.
 \label{eq:detail181}
\end{align}
The ordered-pair calculation below establishes this coefficient as the
$\theta=0$ case of the full family and supplies the non-Beltrami path
algebra used in Theorem~\ref{thm:tg3}.

\subsection{Exact TG3 coefficient-family identities}
\label{app:boxalgebra}

Write $U_\theta=U+\theta W$, with $U=U_{\rm TG3}$ and
\begin{align}
 &W=(\sin x\cos y\cos z,-\sin y\cos z\cos x,0).
 \label{eq:detail182}
\end{align}
At a corner mode $k=(s_x,s_y,s_z)$, $s_j\in\{-1,1\}$, the Fourier
coefficients are
\begin{align}
 &\widehat U(k)=\frac{(s_x,s_y,-2s_z)}{8i},\qquad
 \widehat W(k)=\frac{(s_x,-s_y,0)}{8i}.
 \label{eq:detail183}
\end{align}
Their sum has squared Euclidean magnitude
\begin{align}
 &|\widehat U_\theta(k)|^2
 =\frac{(1+\theta)^2+(1-\theta)^2+4}{64}
 =\frac{6+2\theta^2}{64}.
 \label{eq:detail184}
\end{align}
All eight modes have $|k|^2=3$. Summing their magnitudes and then multiplying
by the common $\A^1$ weight $\sqrt3$ gives
\begin{align}
 &\|U_\theta\|_{\A^0}^2=6+2\theta^2,
 \qquad \|U_\theta\|_{\A^1}^2=18+6\theta^2.
 \label{eq:detail185}
\end{align}
On $[-\delta,\delta]$, $\delta=1/10$, the degree-two Bernstein
coefficients of $6+2\theta^2$ are
\begin{align}
 &6+2\delta^2,\quad 6-2\delta^2,\quad 6+2\delta^2
 =\frac{301}{50},\frac{299}{50},\frac{301}{50}.
 \label{eq:detail186}
\end{align}
Multiplication by three gives
$903/50,897/50,903/50$ for the $\A^1$ norm squared. The Bernstein
convex-hull property, followed by the strict rational square comparisons,
therefore yields
\begin{align}
 &\|U_\theta\|_{\A^0}<\frac{123}{50},\qquad
 \|U_\theta\|_{\A^1}<\frac{17}{4}.
 \label{eq:detail187}
\end{align}

Bilinearity gives the complete, rather than projected or sampled, expansion
\begin{align}
 &B(U_\theta,U_\theta)=V_0+\theta V_1+\theta^2V_2,
 \label{eq:detail188}
\end{align}
where $V_0=B(U,U)$, $V_1=B(U,W)+B(W,U)$, and $V_2=B(W,W)$.
For completeness, write $(c_1,c_2,c_3)=(1+\theta,1-\theta,-2)$, so that
$\widehat U_\theta(p)_j=c_jp_j/(8i)$ at a corner $p$ and $\sum_jc_j=0$.
For any ordered corner pair,
\begin{align}
 i(\widehat U_\theta(p)\cdot q)\widehat U_\theta(q)_j
 &=-\frac i{64}\left(\sum_\ell c_\ell p_\ell q_\ell\right)c_jq_j.
 \label{eq:expandtgpair}
\end{align}
At $k=(2\sigma,0,2\tau)$ there are two such pairs, with
$p_1=q_1=\sigma$, $p_3=q_3=\tau$, $p_2=-q_2=\pm1$.
The scalar sum equals $c_1-c_2+c_3=-2c_2$. Adding both raw vectors and
then projecting gives
\begin{align}
 b_k&=\frac{ic_2}{16}(c_1\sigma,0,c_3\tau),
 &k\cdot b_k&=-\frac{ic_2^2}{8},\\
 \mathbb P_kb_k
 &=b_k-k\frac{k\cdot b_k}{8}
 =\frac{ic_2(c_1-c_3)}{32}(\sigma,0,-\tau).
 \label{eq:expandtgprojection}
\end{align}
Permuting the coordinate calculation yields all nonzero family outputs:
\begin{align}
 \widehat V_\theta(2\sigma,0,2\tau)
 &=\frac{i(3-2\theta-\theta^2)}{32}(\sigma,0,-\tau),\\
 \widehat V_\theta(0,2\sigma,2\tau)
 &=\frac{i(3+2\theta-\theta^2)}{32}(0,\sigma,-\tau),\\
 \widehat V_\theta(2\sigma,2\tau,0)
 &=\frac{-4i\theta}{32}(\sigma,-\tau,0),
 \qquad \sigma,\tau\in\{-1,1\}.
 \label{eq:expandtgoutputs}
\end{align}
An output with three nonzero coordinates has $p=q$ and scalar factor
$\sum_jc_j=0$. With only one nonzero coordinate, summing pairs leaves a
vector parallel to $k$, which the projection removes. The zero output
vanishes by solenoidality. These cases exhaust the sums of two corners.
Collecting powers of $\theta$ gives the coefficients of $V_0,V_1,V_2$.
Their norms are respectively $8(3\sqrt2/32)$,
$[4\cdot2+4\cdot2+4\cdot4]\sqrt2/32$, and $8\sqrt2/32$.
Thus the complete mode calculation gives
\begin{center}\small
\begin{tabular}{@{}lccc@{}}
\toprule
Field & Nonzero modes & Squared shell & $\A^0$ norm\\
\midrule
$U$ & 8 & 3 & $\sqrt6$\\
$W$ & 8 & 3 & $\sqrt2$\\
$V_0$ & 8 & 8 & $3\sqrt2/4$\\
$V_1$ & 12 & 8 & $\sqrt2$\\
$V_2$ & 8 & 8 & $\sqrt2/4$\\
\bottomrule
\end{tabular}
\end{center}
The union of the three residual supports contains 12 modes. For
$q=|\theta|\leq1/10$, the triangle inequality after exact grouping gives
\begin{align}
 &\|B(U_\theta,U_\theta)\|_{\A^0}
 \leq\sqrt2\left(\frac34+q+\frac{q^2}{4}\right).
 \label{eq:detail189}
\end{align}
On $q\in[0,1/10]$, the parenthesis has Bernstein coefficients
\begin{align}
 &\frac34,\qquad\frac45,\qquad\frac{341}{400}.
 \label{eq:detail190}
\end{align}
Its maximum is therefore at most $341/400$. Using
$\sqrt2<99/70$ proves the residual bound used in
Theorem~\ref{thm:tg3box}. These bounds establish the parameter-uniform algebraic estimates used in
that theorem.

\subsection{Bernstein conversion and whole-interval bounds}
\label{app:bernsteindetail}

The interval bounds in Theorem~\ref{thm:tg3box} follow from a polynomial
identity rather than a sampling assumption. Map an interval $[l,r]$ to
$x=(\theta-l)/(r-l)\in[0,1]$ and expand a degree-$m$ polynomial as
$p(l+(r-l)x)=\sum_{j=0}^{m}c_jx^j$. The Bernstein basis consists of the
nonnegative polynomials
\begin{align}
 b_{i,m}(x)&=\binom mi x^i(1-x)^{m-i},
 &\sum_{i=0}^{m}b_{i,m}(x)&=(x+1-x)^m=1.
 \label{eq:bernsteinpartition}
\end{align}
The binomial identity
$\binom mi\binom ij=\binom mj\binom{m-j}{i-j}$ gives
\begin{align}
 \sum_{i=j}^{m}\frac{\binom ij}{\binom mj}b_{i,m}(x)
 &=x^j\sum_{i=j}^{m}\binom{m-j}{i-j}x^{i-j}(1-x)^{m-i}
 =x^j.
 \label{eq:bernsteinmonomial}
\end{align}
Substitute this identity into the monomial expansion and interchange the
finite sums. The Bernstein coefficients and their enclosure are
\begin{align}
 p(l+(r-l)x)&=\sum_{i=0}^{m}\beta_i b_{i,m}(x),
 &\beta_i&=\sum_{j=0}^{i}c_j\frac{\binom ij}{\binom mj},\\
 \min_i\beta_i&\leq p(\theta)\leq\max_i\beta_i.
 \label{eq:bernsteinconversion}
\end{align}
For $p(\theta)=6+2\theta^2$ and $\theta=-\delta+2\delta x$,
the monomial coefficients are $6+2\delta^2,-8\delta^2,8\delta^2$.
At degree two the conversion gives $\beta_0=c_0$,
$\beta_1=c_0+c_1/2$, and $\beta_2=c_0+c_1+c_2$, recovering
$6+2\delta^2,6-2\delta^2,6+2\delta^2$.
For $3/4+q+q^2/4$ with $q=x/10$, it gives
$3/4,4/5,341/400$. The rational square comparisons
$(123/50)^2>301/50$, $(17/4)^2>903/50$, and $(99/70)^2>2$
complete the bounds used in the Methods.

\section{Exact comparison induction and terminal barrier}
\label{app:comparison}

The scalar comparison preserves the error bound until terminal Wiener
smallness supplies all-time decay. The derivation below proves comparison
and the terminal barrier, Section~\ref{app:slabalgorithm} specifies the slab
construction, and Sections~\ref{app:abccomparison}, \ref{app:tg3comparison},
and \ref{app:boxcomparison} apply it to the three heat-path families.
For the viscosity-$\nu$ error equation in Theorem~\ref{thm:acceptance}, set
$Y_j=\|w_N\|_{\A^j}$ and $Y=Y_0$. The convolution estimate gives
\begin{align}
 &D^+Y+\nu Y_2\leq d+A_0Y_1+A_1Y+YY_1.
 \label{eq:detail194}
\end{align}
For $Y>0$, write $z=Y_2/Y\geq1$, $w=\sqrt z$, and
$b=(A_0+Y)/\nu$. Since $Y_1\leq\sqrt{YY_2}=Yw$, the terms containing
$Y_1$ and $Y_2$ satisfy
\begin{align}
 &(A_0+Y)Y_1-\nu Y_2\leq\nu Y(bw-w^2)
 \leq\frac{\nu Y}{2}(b^2-1).
 \label{eq:detail195}
\end{align}
Expansion of the right-hand side gives
\begin{align}
 &D^+Y\leq d+
 \left(A_1-\frac\nu2+\frac{A_0^2}{2\nu}\right)Y
 +\frac{A_0}{\nu}Y^2+\frac1{2\nu}Y^3,
 \label{eq:detail196}
\end{align}
which is the scalar inequality in Theorem~\ref{thm:acceptance}. The case
$Y=0$ follows by the upper-Dini limit.

Let
\begin{align}
 &F_i(q)=r_i+\left(A_{1,i}-\frac12+\frac{A_{0,i}^2}{2}\right)q
 +A_{0,i}q^2+\frac12q^3.
 \label{eq:detail197}
\end{align}
For $q\geq0$, $F_i''(q)=2A_{0,i}+3q\geq0$. Therefore the maximum of
$F_i$ along an affine segment is attained at an endpoint. If the exact slack
\eqref{eq:slack} is nonnegative, then $y'\geq F_i(y)$ throughout the slab.
To justify comparison even at tangential contact, fix a bound $M$ containing
the ranges of $Y$ and $y$ on the slab and let $L$ be a Lipschitz constant for
$F_i$ on $[0,M]$. A finite Galerkin trajectory has continuously differentiable
Fourier coefficients, so its Wiener norm $Y$ is locally Lipschitz. Thus
$z=(Y-y)_+$ is absolutely continuous. Where $Y>y$, the differential
inequality gives $z'\leq F_i(Y)-F_i(y)\leq Lz$ almost everywhere.
Where $Y<y$, $z'=0$, and the derivative of an absolutely continuous function
vanishes almost everywhere on a level set, giving the same inequality where
$Y=y$. Therefore
\begin{align}
 (e^{-L(t-t_i)}z(t))'&\leq0\quad\text{almost everywhere},
 &z(t_i)&=0,
 \label{eq:comparisonpositivepart}
\end{align}
so $z=0$ and $Y\leq y$. Induction applies because the stored endpoints agree
exactly. This argument supplies the scalar comparison used in
Section~\ref{sec:mathmethods} without assuming that non-strict contact is a
contradiction.

For the terminal barrier, let $W(t)=\|u_N(t)\|_{\A^0}$. Coefficientwise
evaluation of the Galerkin equation, followed by
$\|u_N\|_{\A^1}\leq\|u_N\|_{\A^2}$ on nonzero integer modes, gives
\begin{align}
 &D^+W+\|u_N\|_{\A^2}
 \leq W\|u_N\|_{\A^1}
 \leq W\|u_N\|_{\A^2}.
 \label{eq:detail199}
\end{align}
Thus
\begin{align}
 &D^+W\leq-(1-W)\|u_N\|_{\A^2}.
 \label{eq:detail200}
\end{align}
At any first attempted crossing of $W=c<1$, the right side is negative.
Therefore $W\leq c$. Since $\|u_N\|_{\A^2}\geq W$ on mean-zero modes,
\begin{align}
 &W(t)\leq c\exp(-(1-c)(t-T_*)),\qquad t\geq T_*.
 \label{eq:detail201}
\end{align}
This gives global decay independently of the residual-action criterion.
For viscosity $\nu>0$, the same calculation retains $\nu$ in the dissipative
term. If $W(T_*)\leq\eta<\nu$, it gives
\begin{align}
 D^+W&\leq-(\nu-W)\|u_N\|_{\A^2},
 &W(t)&\leq\eta e^{-(\nu-\eta)(t-T_*)}.
 \label{eq:viscousbarrierdetail}
\end{align}
The bound $W\leq\eta$ follows first by the same barrier argument, after which
$\|u_N\|_{\A^2}\geq W$ yields the exponential estimate by integration.
This establishes the terminal continuation step for the general-viscosity
comparison principle in Section~\ref{sec:mathmethods}.

\subsection{Reconstruction of the rational comparison trajectory}
\label{app:slabalgorithm}

The following construction recovers the finite records used in
Section~\ref{sec:exactmethods} from the formulas in
Section~\ref{sec:mathmethods}. All endpoints and coefficients are fractions.
For a chosen integer precision $b$, outward rounding means
\begin{align}
 \operatorname{up}_b(x)&=2^{-b}\lceil2^bx\rceil,
 &\operatorname{down}_b(x)&=-\operatorname{up}_b(-x).
 \label{eq:outwardrounding}
\end{align}
These operations preserve upper and lower bounds even when repeated.
For $x\geq0$, let $S_m(x)=\sum_{j=0}^{m}x^j/j!$.
The first omitted term is $x^{m+1}/(m+1)!$, and every subsequent term
has ratio at most $x/(m+2)$. If this ratio is less than one, a geometric sum
therefore bounds the positive Taylor tail:
\begin{align}
 S_m(x)&\leq e^x\leq S_m(x)+T_m(x),
 &T_m(x)&=\frac{x^{m+1}}{(m+1)!}\frac1{1-x/(m+2)},\\
 \operatorname{down}_b\!\left(\frac1{S_m+T_m}\right)
 &\leq e^{-x}\leq\operatorname{up}_b\!\left(\frac1{S_m}\right).
 \label{eq:exponentialenclosure}
\end{align}
The independent calculation uses the degree-129 and degree-128 Taylor
polynomials of $e^{-x}$ as lower and upper bounds. Taylor's integral
remainder has sign $(-1)^{m+1}$ because every derivative of $e^{-x}$ has
that sign, which proves the enclosure without requiring the early terms of
the alternating series to decrease.

Whole-slab coefficients follow from these exponential intervals. For
$t^pe^{-qt}$, differentiation gives
$t^{p-1}e^{-qt}(p-qt)$, so its maximum on $[l,r]$ occurs at an endpoint or
at $t=p/q$ if that point lies in the slab. Its minimum occurs at an endpoint.
For $K_1$ and $K_5$, the positive summands in their signed formulas use upper
interval endpoints and the negative summands use lower endpoints when forming
an upper bound. Substitution into Equations~\eqref{eq:A0}--\eqref{eq:r},
with outward rounding after each declared coefficient construction, gives the
cyclic-shear slab coefficients. The heat-path bounds in the other three
applications decrease with time, so their slab maxima occur at $l$.

Given $t_i,y_i$, the algorithm first chooses a trial length $h$ and computes
the enclosing scalar polynomial $F_i$ defined above. Write $v_i=F_i(y_i)$.
The trial endpoint and its tested slack are
\begin{align}
 y_{i+1}&=\operatorname{up}_b\!\left(\max\{0,y_i+h\kappa_i v_i\}\right),\\
 \sigma_i&=\min_{q\in\{y_i,y_{i+1}\}}
 \left\{\frac{y_{i+1}-y_i}{h}-F_i(q)\right\}.
 \label{eq:slabreconstruction}
\end{align}
If $\sigma_i<0$, the algorithm halves $h$, recomputes the coefficients on
the smaller slab, and retries. If $\sigma_i\geq0$, it accepts the slab and
sets $t_{i+1}=t_i+h$. At each accepted endpoint it tests the pointwise bound
$A_0(t_{i+1})+y_{i+1}<99/100$ and stops at its first success.
The proposal factors create slack, but only the exact endpoint inequality
determines acceptance. Failure to find a slab supplies no regularity claim.

The four applications use the following deterministic settings. The Taylor
degree refers to the positive Taylor sum in
Equation~\eqref{eq:exponentialenclosure}. Every application allows 16 trial
lengths per step and uses its theorem's perturbation radius as $y_0$.
\begin{center}\small
\begin{tabular}{@{}lrrrrr@{}}
\toprule
Family & $b$ & Taylor degree & Initial $h$ & $\kappa_i$ if $v_i\geq0$ & $\kappa_i$ if $v_i<0$\\
\midrule
Cyclic shear & 72 & 96 & $1/512$ & $65/64$ & $63/64$\\
Equal ABC & 80 & 96 & $1/1024$ & $129/128$ & $127/128$\\
TG3 & 80 & 128 & $1/1024$ & $129/128$ & $127/128$\\
TG3 coefficient interval & 80 & 128 & $1/1024$ & $129/128$ & $127/128$\\
\bottomrule
\end{tabular}
\end{center}
The maximum search times are $2$, $4$, $4/3$, and $4/3$, respectively.
All four constructions stop earlier, at the terminal times in their theorems.
The stored fractions are independently checked against fresh whole-slab
coefficient bounds, endpoint continuity, nonnegativity, and terminal
smallness. Thus a reader can generate the trajectory or verify the supplied
trajectory using the same analytic inequalities.

\subsection{Equal-ABC comparison}
\label{app:abccomparison}

For the exact ABC heat path, both path norms entering
\eqref{eq:dini0} are bounded by
\begin{align}
 &A(t)=\frac{297}{70}ae^{-t},
 \label{eq:detail205}
\end{align}
because $\sqrt2<99/70$, and the residual majorant is zero. Substitution of
$A_0=A_1=A$ and $r=0$ into \eqref{eq:dini} gives the scalar inequality in
the proof of Theorem~\ref{thm:abc}. On a slab with a constant exact upper
bound $A_i$ and affine comparison function $y$, define
\begin{align}
 &F_i^{\rm ABC}(q)=
 \left(A_i-\frac12+\frac{A_i^2}{2}\right)q
 +A_iq^2+\frac12q^3.
 \label{eq:detail206}
\end{align}
For $q\geq0$,
$({F_i^{\rm ABC}})''(q)=2A_i+3q\geq0$. Therefore checking the affine slope
against $F_i^{\rm ABC}$ at both endpoint values of $y$ controls the entire
slab. Applying the first-contact induction above to 2390 adjacent exact slabs
from $0$ to $1195/512$, starting at $y(0)=1/100$, yields
\begin{align}
 &\|u_N(1195/512)\|_{\A^0}
 \leq
 \frac{598396355886649288745723}{604462909807314587353088}
 <\frac{99}{100}
 \label{eq:detail207}
\end{align}
and
\begin{align}
 &\sup_{0\leq t\leq1195/512}\|u_N(t)\|_{\A^0}
 \leq
 \frac{437511399315879405468719}{151115727451828646838272}.
 \label{eq:detail208}
\end{align}
The inequalities hold for all $0\leq a\leq17/25$, because $A(t)$ is
nondecreasing in $a$, and for every $N\geq1$, because the path has no omitted
cutoff modes. The first bound enters the Wiener barrier. The second supplies
the finite-time bound used in the Galerkin compactness argument. Together, these estimates supply the comparison bounds used in
Theorem~\ref{thm:abc}.

\subsection{TG3 comparison}
\label{app:tg3comparison}

For $0\leq a\leq41/40$, the TG3 heat path and its full residual satisfy the
rational bounds
\begin{align}
 &A_0(t)\leq\frac{49}{20}\frac{41}{40}e^{-3t},\qquad
 A_1(t)\leq3\frac{99}{70}\frac{41}{40}e^{-3t},
 \label{eq:detail209}
\end{align}
\begin{align}
 &d(t)\leq\frac34\frac{99}{70}
 \left(\frac{41}{40}\right)^2e^{-6t},
 \label{eq:detail210}
\end{align}
using $\sqrt6<49/20$ and $\sqrt2<99/70$. Substitution into the
unit-viscosity scalar inequality of Theorem~\ref{thm:acceptance} defines the
convex slab polynomial. The endpoint induction above, applied to 873 adjacent
exact slabs from $0$ to $873/1024$ with $y(0)=1/40$, yields
\begin{align}
 &\left\|u_N\!\left(\frac{873}{1024}\right)\right\|_{\A^0}
 \leq
 \frac{1196766436682595646613963}{1208925819614629174706176}
 <\frac{99}{100}
 \label{eq:detail211}
\end{align}
and the bridge bound
\begin{align}
 &\sup_{0\leq t\leq873/1024}\|u_N(t)\|_{\A^0}
 \leq
 \frac{1533837106283467124791861}{604462909807314587353088}.
 \label{eq:detail212}
\end{align}
All path modes lie in cube one, while the full nonlinear output is retained
in $d(t)$ before projection. Hence the comparison applies to every $N\geq1$
without an exceptional cutoff. These estimates supply the comparison bounds used in
Theorem~\ref{thm:tg3}.

\subsection{Uniform coefficient-family comparison}
\label{app:boxcomparison}

For $|\theta|\leq1/10$ and $0\leq a\leq1$, the bounds proved in the preceding
algebra appendix give
\begin{align}
 &A_0(t)\leq\frac{123}{50}e^{-3t},\qquad
 A_1(t)\leq\frac{17}{4}e^{-3t},\qquad
 d(t)\leq\frac{99}{70}\frac{341}{400}e^{-6t}.
 \label{eq:detail213}
\end{align}
The first two expressions bound $\|ae^{-3t}U_\theta\|_{\A^0}$ and
$\|ae^{-3t}U_\theta\|_{\A^1}$ because $a\leq1$. The third bounds the full
residual $a^2e^{-6t}B(U_\theta,U_\theta)$ because $a^2\leq1$.
Substitution into the scalar polynomial derived at the start of this appendix
produces one comparison problem valid for the entire coefficient interval.
The convex endpoint induction, applied to 992 adjacent exact slabs from zero
to $31/32$ with $y(0)=1/40$, yields
\begin{align}
 &\sup_{\substack{|\theta|\leq1/10,\ 0\leq a\leq1\\
                   \|h\|_{\A^2}\leq1/40}}
 \sup_{N\geq1}\|u_N(31/32)\|_{\A^0}
 \leq
 \frac{1196626932805349354526769}{1208925819614629174706176}
 <\frac{99}{100}
 \label{eq:detail214}
\end{align}
and
\begin{align}
 &\sup_{\substack{|\theta|\leq1/10,\ 0\leq a\leq1\\
                   \|h\|_{\A^2}\leq1/40}}
 \sup_{N\geq1}\sup_{0\leq t\leq31/32}\|u_N(t)\|_{\A^0}
 \leq
 \frac{3005818583529760358215883}{1208925819614629174706176}.
 \label{eq:detail215}
\end{align}
Every path mode lies in cube one. Although the full residual occupies cube
two, for $N=1$ its contribution is exactly projected by $P_N$, and Wiener
projection is contractive. Hence the same residual majorant and comparison
cover every $N\geq1$. The first inequality activates the terminal barrier.
The second supplies the finite-time bridge used in the compactness argument.
Together, these bounds establish the comparison step in
Theorem~\ref{thm:tg3box}.

\section{Floating-point Galerkin update and recorded metrics}
\label{app:numerics}

The numerical algorithms define the finite-resolution diagnostics in
Equations~\eqref{eq:ensembleh}--\eqref{eq:trustgates}, separately from the
proof of Theorem~\ref{thm:main}. After the perturbation normalisation below,
Section~\ref{app:parameters} gives the parameter maps,
Section~\ref{app:timestepping} specifies time integration, and
Section~\ref{app:statistics} defines the statistical summaries. The perturbation construction is real
because the negative-mode coefficient is the complex conjugate of the
positive-mode coefficient. It is solenoidal because each polarization vector
is perpendicular to its wave vector. Since the two polarizations are
orthonormal, $|\widehat h(k_j)|=c_j$. The positive and negative modes make
equal contributions, and therefore
\begin{align}
 &\|h\|_{\A^2}
 =\sum_j\left(|k_j|^2|\widehat h(k_j)|
              +|-k_j|^2|\widehat h(-k_j)|\right)
 =\sum_j2|k_j|^2c_j.
 \label{eq:detail216}
\end{align}
Choosing
\begin{align}
 &C_\rho=\frac{\rho}{2\sum_j|k_j|^{\beta+2}}
 \label{eq:detail217}
\end{align}
in $c_j=C_\rho|k_j|^\beta$ proves
Equation~\eqref{eq:ensemblenorm} exactly.

\subsection{Complete parameter maps and initial fields}
\label{app:parameters}

This construction supplies the input fields used in
Section~\ref{sec:computationalmethods} and the ensemble of
Section~\ref{sec:ensemble}. The ordered positive-mode list is
\begin{align}
 (k_1,k_2,\ldots,k_{16})={}&((1,1,0),(1,-1,0),(1,0,1),(1,0,-1),\notag\\
 &(0,1,1),(0,1,-1),(1,1,1),(1,1,-1),\notag\\
 &(2,1,0),(2,0,1),(1,2,0),(0,2,1),\notag\\
 &(1,0,2),(1,2,1),(2,1,1),(1,1,2)).
 \label{eq:perturbationmodelist}
\end{align}
For each $k_j$, choose $r_j=(1,0,0)$ unless $|k_j\cdot r_j|/|k_j|>0.85$,
in which case choose $r_j=(0,1,0)$. The two polarisations are
\begin{align}
 e_{j,1}&=\frac{k_j\times r_j}{|k_j\times r_j|},
 &e_{j,2}&=\frac{k_j\times e_{j,1}}{|k_j\times e_{j,1}|}.
 \label{eq:polarisationbasis}
\end{align}
The branch prevents a zero denominator, and cross-product orthogonality gives
$e_{j,1}\cdot e_{j,2}=k_j\cdot e_{j,1}=k_j\cdot e_{j,2}=0$.

Let $(z_1,z_2,\ldots,z_8)\in[0,1]^8$ be a Sobol point, let
$\{x\}=x-\lfloor x\rfloor$, and put $g=(\sqrt5-1)/2$.
The theorem stratum uses $(a,\rho)=(0.7z_1,0.01z_2)$, and the stress
stratum uses $(a,\rho)=(0.5+0.8z_1,0.005+0.045z_2)$.
The remaining coordinates give
\begin{align}
 \beta&=-3+6z_3, &c&=z_4,\\
 d_j&=\{z_5+jz_6+j^2z_7\},
 &\phi_j&=2\pi[cz_5+(1-c)d_j],\\
 \psi_j&=2\pi\{z_8+jg\}, &j&=1,2,\ldots,16.
 \label{eq:ensemblephasemap}
\end{align}
Substitution into Equation~\eqref{eq:ensembleh}, followed by its exact
normalisation, completes the initial field. The two strata use 2,048
scrambled Sobol points each, with seeds 20260909 and 20260910.

The neural experiment uses the same mode list with ten coordinates.
The first 384 points of a scrambled ten-dimensional Sobol sequence with
seed 20260910 are stored in single precision. The sample with zero-based
index $q$ uses family indicator $f=q\bmod2$, where $f=0$ selects cyclic shear
and $f=1$ selects $U_\theta$. Its parameter map is
\begin{align}
 a&=1.3z_1, &\theta&=-0.2+0.4z_2, &\rho&=0.05z_3,\\
 \beta&=-3+6z_4, &c&=z_5,\\
 d_j&=\{z_6+jz_7+j^2z_8\}, &\phi_j&=2\pi[cz_6+(1-c)d_j],\\
 \psi_j&=2\pi\{z_9+jg+z_{10}\}, &j&=1,2,\ldots,16.
 \label{eq:pinoparametermap}
\end{align}
The $\theta$ coordinate affects only the TG3 centre but participates in the
fixed train/OOD split for both families. Targets use a $12^3$ grid, cutoff
two, unit viscosity, and RK4 step $0.005$. The resulting physical velocity
arrays are stored in single precision. These details reproduce the parameter
split and target generation of Section~\ref{sec:pinomethod}.

\subsection{Time integration and diagnostics}
\label{app:timestepping}

This subsection supplies the update and diagnostics used in
Section~\ref{sec:computationalmethods} and the convergence comparisons in
Section~\ref{sec:ensemble}.

The Fourier update below defines the trajectories used by the numerical
Methods. The conservative condition $4K<n$ prevents wrap-around aliasing in
spatial means of products of up to four cutoff-$K$ factors, since a nonzero
sum of at most four supported wave vectors cannot equal a grid frequency
multiple. The neural loss uses a different native-grid construction, specified
in Equation~\eqref{eq:pinoimplementedrhs}.

The non-proof simulations use the projected Fourier ODE
\begin{align}
 &\dot u=\mathcal L(u):=-\nu Au-P_KB(u,u).
 \label{eq:detail222}
\end{align}
With step $\Delta t$, classical RK4 is
\begin{align}
 k_1&=\mathcal L(u^n),&
 k_2&=\mathcal L(u^n+\tfrac12\Delta tk_1),\\
 k_3&=\mathcal L(u^n+\tfrac12\Delta tk_2),&
 k_4&=\mathcal L(u^n+\Delta tk_3),\\
 u^{n+1}&=P_K\mathbb P\left[u^n+\frac{\Delta t}{6}
 (k_1+2k_2+2k_3+k_4)\right].
 \label{eq:detail223}
\end{align}
Here $k_1$ is the slope at the beginning of the step, $k_2$ and $k_3$ are
two midpoint slopes, and $k_4$ is the endpoint slope predicted from $k_3$.
The weights $1,2,2,1$ form the fourth-order Runge--Kutta quadrature. The final
projection $P_K\mathbb P$ removes roundoff components outside the divergence-free truncated
space. It is not an analytic regularization. Appendix~\ref{app:discreteexpansion}
derives the discrete transform normalisation, aliasing condition, and
fourth-order expansion of this update.
The fixed-datum solver additionally restricts the time step through a declared
Courant--Friedrichs--Lewy (CFL) safety factor. The ensemble, boundary-refinement,
and neural candidate calculations instead use the fixed RK4 steps stated in
their corresponding methods.
The nonlinear energy defect is
$|\operatorname{Re}\langle P_KB(u,u),u\rangle|$. The divergence metric is
$\|k\cdot\widehat u(k)\|_{\ell^2_k}$. Shell enstrophy and transfer are
\begin{align}
 &X_m=\sum_{m\leq|k|<m+1}|k|^2|\widehat u(k)|^2,
 \qquad
 T_m=-2\operatorname{Re}\sum_{m\leq|k|<m+1}
 |k|^2\overline{\widehat u(k)}\cdot\widehat{B(u,u)}(k).
 \label{eq:detail224}
\end{align}
For the fixed-datum campaign, the tail fraction is the enstrophy in
the final two cubical cutoff layers divided by total enstrophy. Equal-time spatial errors restrict both
solutions to the common lower cutoff before taking the relative $H^1$ norm.
Temporal errors compare base and half-step solutions at the same cutoff.
Integrated balance defects compare trapezoidal time integrals with the exact
endpoint change for each recorded balance. These definitions determine every
fixed-datum metric evaluated in the Results.

For the ensemble, $Z=R+X^2/E$ and the tail statistic is
\begin{align}
 &\operatorname{Tail}_Z(t)=
 \frac{\sum_{\max_i|k_i|=K}|k|^4|\widehat u(k,t)|^2}{Z(t)}.
 \label{eq:detail225}
\end{align}
The paired action error and residual-trajectory error are the first two
quantities in \eqref{eq:trustgates}. The normalised residual-balance defect is
\begin{align}
 &\operatorname{Bal}_R=
 \frac{|\frac12(R(T)-R(0))+\int_0^T D_R\dd t-\int_0^T F\dd t|}
 {\max\{\frac12|R(T)-R(0)|,|\int_0^T D_R\dd t|,
 |\int_0^T F\dd t|\}}.
 \label{eq:detail226}
\end{align}
Here $\nu=1$. The denominator is positive for every retained sample. Heatmap
cells are two-dimensional bins in $(a,\rho)$ after marginalizing the other six
Sobol coordinates. Bilinear interpolation is used only to display continuous
bin statistics. The trust fraction uses nearest-bin rendering, and empty bins
are gray. In Figure~\ref{fig:ensembletransitions}(c), the raw signed-$F$
histogram is convolved with a Gaussian kernel of standard deviations $1.35$
$F$-bins and $0.9$ time bins. Every time column is then renormalized so that
its discrete integral in $F$ equals one. This operation controls raster
aliasing only. All quantiles, gates, correlations, and theorem statements use
the unsmoothed trajectories.

Figure~\ref{fig:highrestransitions} uses a different normalization that must
not be confused with a probability density. At each time, the refinement
calculation bins $\log_{10}\max(R,10^{-30})$ and $F$ with 96 bin edges.
The displayed residual interval spans the pooled 0.2--99.8 percentiles, and
the signed-source interval is symmetric about zero with half-width equal
to the 99.5 percentile of $|F|$. A Gaussian filter with standard deviations
$1.2$ value bins and $0.8$ time bins smooths the counts. Division by the
sum in each time column gives probability per bin conditional on the
displayed interval, without division by bin width. The colour scale shows
its base-ten logarithm. A display floor at the first percentile of positive
smoothed bin probabilities affects only
colour contrast, not the recorded statistics. These definitions explain
the colour bars and normalization used in the Results.

The vector-field curves use periodic trilinear interpolation of saved
uniform-grid arrays, at $48^3$ in Figure~\ref{fig:fieldgeometry} and $64^3$
in Figure~\ref{fig:multicentergeometry}. Classical
RK4 is applied to the normalised local vector direction with arclength step
$0.045$ for 105 steps in each direction. Candidate seeds lie on a stride-three
subgrid at or above its own 70th percentile of vector magnitude. A
deterministic periodic farthest-point rule weighted by magnitude retains 42
velocity seeds and 48 vorticity seeds. Tubes have radii $0.018$ and $0.015$,
respectively. Panels (c) and (f) use paired unit glyphs for $\widehat\omega$
and $e_{\max}$. The field arrays, seeds, camera, streamline parameters and summary
statistics are retained with the processed data.
For Figure~\ref{fig:multicentergeometry}, the same procedure is applied to
analytic $64^3$ evaluations of the three centre formulas at $t=0$. These
arrays are descriptive render inputs and are not read by any theorem checker.

For deterministic refinement comparisons, let $q^{(c)}_m$ and $q^{(f)}_m$
denote coarse and fine values sampled at the same times. The displayed curve
error and scalar error are, respectively,
\begin{align}
 &\mathcal E_q=
 \frac{\left(M^{-1}\sum_{m=1}^{M}|q^{(c)}_m-q^{(f)}_m|^2\right)^{1/2}}
      {\max\{\max_m|q^{(f)}_m|,10^{-14}\}},
 \qquad
 \mathcal E_s=\frac{|s^{(c)}-s^{(f)}|}{\max\{|s^{(f)}|,10^{-14}\}}.
 \label{eq:detail227}
\end{align}
The fine trajectory is subsampled at the $M$ coarse output times before the
root-mean-square comparison. These formulae reproduce the normalisation and
denominator floors used by the ensemble and refinement implementations.
Spatial comparisons first restrict both Fourier arrays to the common lower
cutoff. Temporal comparisons use the same cutoff and base versus half-step
trajectories. These scalar definitions apply to the ensemble and
cutoff-refinement diagnostics in Figures~\ref{fig:highresphase} and
\ref{fig:highrestransitions}. The fixed-datum field comparisons instead use
the relative common-cutoff $H^1$ norm defined above, rather than a
scalar-trajectory error.

\subsection{Statistical estimators and reconstruction of empirical claims}
\label{app:statistics}

The estimators below define the ensemble, field-geometry and public-DNS
summaries in
Figures~\ref{fig:ensembleconfig}, \ref{fig:ensembletransitions},
\ref{fig:highresphase}, \ref{fig:highrestransitions},
\ref{fig:dynamicvortices}, \ref{fig:fieldgeometry},
\ref{fig:jhtdbgeometry}, and \ref{fig:jhtdbstructured}. All summaries are
descriptive functions of deterministic arrays. No $p$-value or
independent-replicate interpretation is assigned to Sobol points or spatially
correlated DNS samples.

For finite values $x_1,x_2,\ldots,x_m$, define the mean, positive fraction, and
root-mean-square value by
\begin{align}
 &\bar x=\frac1m\sum_{i=1}^m x_i,
 \qquad f_+(x)=\frac1m\sum_{i=1}^m\mathbb I_{\{x_i>0\}},
 \qquad x_{\rm rms}=\left(\frac1m\sum_{i=1}^m x_i^2\right)^{1/2}.
 \label{eq:detail228}
\end{align}
For paired arrays $x_i,y_i$, the Pearson correlation is
\begin{align}
 &r_{xy}=\frac{\sum_i(x_i-\bar x)(y_i-\bar y)}
 {\left[\sum_i(x_i-\bar x)^2\sum_i(y_i-\bar y)^2\right]^{1/2}}.
 \label{eq:detail229}
\end{align}
It is used only as a descriptive measure of linear association. For a
probability level $p\in[0,1]$, sort the data as
$x_{(0)}\leq x_{(1)}\leq\ldots\leq x_{(m-1)}$, set
$h=(m-1)p$, $j=\lfloor h\rfloor$, and $\theta=h-j$. The empirical quantile is
\begin{align}
 &Q_p(x)=(1-\theta)x_{(j)}+\theta x_{(\min\{j+1,m-1\})}.
 \label{eq:detail230}
\end{align}
This linear order-statistic interpolation defines every median, percentile,
and $5\%/50\%/95\%$ band quoted in the Results.

For the configuration ensemble, theorem membership is the Boolean predicate
$0\leq a\leq0.7$ and $0\leq\rho\leq0.01$. Exterior membership is its
complement within the sampled union of strata. Numerical trust is the
conjunction of the six inequalities in \eqref{eq:trustgates}. Thus each
cohort or gate count equals a sum of indicator functions, for example
\begin{align}
 &n_{\rm trusted,theorem}
 =\sum_{i=1}^{4096}
  \mathbb I_{\{i\ \mathrm{satisfies\ the\ theorem\ predicate}\}}
  \mathbb I_{\{i\ \mathrm{satisfies\ all\ six\ gates}\}}.
 \label{eq:detail231}
\end{align}
The action quantiles apply $Q_p$ to the cohort-specific $J_{0.5}$ values, and
the correlations apply $r_{xy}$ to $x=\log_{10}J_{0.5}$ and the stated Sobol
coordinate. Two-dimensional heatmap bins condition on $(a,\rho)$ and
marginalize the other six coordinates. Empty bins remain undefined.

The vortex representatives are order statistics with deterministic tie
breaking by sample index: maximise $J_{0.5}$ within trusted theorem members,
maximise $J_{0.5}$ within trusted exterior members, and maximise
$\operatorname{Tail}_Z$ within rejected exterior members. The table beneath
Figure~\ref{fig:dynamicvortices} evaluates the listed scalar functions on
these three selected rows. The medians and extrema associated with
Figure~\ref{fig:fieldgeometry} apply $Q_{1/2}$, minimum, and maximum to the
saved grid or glyph arrays specified in Appendix~\ref{app:numerics}.

The $73{,}728$ public-DNS gradient evaluations decompose as
\begin{align}
 &8(4096)+5(4096)+4(16^3)+4096=73{,}728,
 \label{eq:detail232}
\end{align}
corresponding to the eight-time forced-isotropic acquisition, the five-time
velocity-inclusive acquisition, four structured volumes, and the channel-flow
sample. Pooled statistics concatenate the indicated finite arrays before
applying the definitions above. Framewise statistics apply them separately.
The helicity--production correlation uses only rows containing both velocity
and gradient data. The structured cross-flow means and positive fractions
are computed separately for each $16^3$ query and therefore do not constitute
resolution-convergence estimates.

\section{Neural-operator construction and diagnostics}
\label{app:pino}

The discrete construction below specifies the neural predictions and
candidate comparisons in Section~\ref{sec:pinomethod} and
Figure~\ref{fig:pinomain}. The split and layer map below define the prediction
problem, Section~\ref{app:networkdetail} specifies its operations, and
Section~\ref{app:neuralobjective} derives the losses and evaluation metrics.
The experiment uses 384 deterministic Sobol
configurations split into 127 training, 45 validation, and 212 OOD fields. If
$\theta\in[-0.2,0.2]$ denotes the TG3 coefficient, $\beta\in[-3,3]$ the
perturbation spectral tilt, and $c\in[0,1]$ its phase-coherence coordinate,
the OOD predicate is
\begin{align}
 &|\theta|\geq0.14\quad\text{or}\quad |\beta|\geq2.3
 \quad\text{or}\quad c\geq0.85.
 \label{eq:pinoood}
\end{align}
The remaining configurations whose indices are multiples of five form the
validation set, and all other configurations form the training set. The input
normalisation uses only the training indices.

For a latent field $v_\ell$ with eight channels, let $\operatorname{GN}$
denote group normalisation and $\operatorname{GELU}$ the Gaussian error linear
unit. One Fourier layer has the
form
\begin{align}
 &v_{\ell+1}=\operatorname{GELU}\!\left[
  \operatorname{GN}\!\left(W_\ell v_\ell+
  \mathcal F^{-1}(\mathcal R_\ell\mathcal Fv_\ell)\right)\right].
 \label{eq:fnolayer}
\end{align}
Here $W_\ell$ is a pointwise channel map, $\mathcal F$ is the three-dimensional
discrete Fourier transform, $\mathcal R_\ell$ is the learned complex multiplier
on the retained $3\times3\times3$ Fourier block, $\operatorname{GN}$ denotes
group normalisation, and $\operatorname{GELU}$ denotes the Gaussian error
linear unit. A pointwise lift maps the three normalised velocity components and
two family indicators to $v_0$. Three layers of the form
\eqref{eq:fnolayer} precede a pointwise projection to the fifteen output
channels formed by three velocity components at five checkpoints.

\subsection{Network operations and optimisation}
\label{app:networkdetail}

This subsection makes the architecture in Section~\ref{sec:pinomethod}
implementable without inferring Fourier conventions from a schematic layer.
The spectral branch uses a real-to-complex transform with unnormalised
forward transform and reciprocal-size inverse transform. Its learned block
occupies indices $k_x,k_y,k_z\in\{0,1,2\}$ in the real-transform storage
array. All other entries of that branch are zero. This is the implemented
one-block parameterisation, rather than a full signed low-frequency cube.
The pointwise branch acts on every grid point and can carry components absent
from the spectral branch. The output layer imposes neither a Leray projection
nor a hard zero-mean constraint.

Group normalisation partitions the eight channels into two groups of four.
For one sample and group, write $\mu_g$ and $v_g$ for the mean and biased
variance over its $4n^3$ values. Its affine output and the activation are
\begin{align}
 \operatorname{GN}(z)_c&=\gamma_c
 \frac{z_c-\mu_g}{\sqrt{v_g+10^{-5}}}+\beta_c,\\
 \operatorname{GELU}(z)&=\frac z2
 \left[1+\operatorname{erf}\!\left(\frac z{\sqrt2}\right)\right].
 \label{eq:networkoperations}
\end{align}
The trainable coefficients $\gamma_c,\beta_c$ act per channel. The input
normalisation instead uses the training-set mean and sample standard
deviation of each physical velocity component, with standard deviations
floored at $10^{-6}$. The family channels are $(f,1-f)$. The network outputs
physical velocity directly, so the losses below compare unnormalised fields.

Both models start from copies of one initial parameter state under seed
20260910. Adam uses learning rate $0.0015$, moment coefficients $(0.9,0.999)$,
denominator offset $10^{-8}$, and no weight decay. Before each update the
gradient vector is rescaled if its Euclidean norm exceeds one. Batches contain
four samples. Training runs for at most 18 epochs, with a patience of five
epochs, and restores the parameters at the smallest validation objective.
Each model uses its own training objective for model selection, so the two
validation-loss ordinates are not directly comparable physical errors.

\subsection{Discrete objective and measured quantities}
\label{app:neuralobjective}

The following formulas expand Equation~\eqref{eq:pinoloss} term by term.
The field and initial losses compare velocity values, the midpoint penalties
measure discrete physical consistency, and the Fourier penalties act on the
specified grid. Their averaging factors are part of the objective.

Write $\widetilde u_j=\mathcal G_\vartheta(u_0,f)_j$ and let $u_j$ denote the
Galerkin target. The implementation uses a periodic $n^3$ grid with $n=12$
and the forward-normalised discrete Fourier transform. For a three-component
grid field $v$, define the component-averaged mean square
\begin{align}
 &\mathfrak M_3(v)=\frac1{3n^3}\sum_{c=1}^{3}\sum_{x\in G_n}|v_c(x)|^2,
 \qquad
 \mathfrak M_1(q)=\frac1{n^3}\sum_{x\in G_n}|q(x)|^2
 \label{eq:pinodiscretemeans}
\end{align}
for vector and scalar fields, respectively. The field and initial losses are
\begin{align}
 &\mathcal L_{\rm field}=\frac15
 \sum_{j=0}^{4}\mathfrak M_3(\widetilde u_j-u_j),
 \qquad
 \mathcal L_{\rm initial}=\mathfrak M_3(\widetilde u_0-u_0).
 \label{eq:pinofieldloss}
\end{align}
For $\delta t_j=t_{j+1}-t_j$, set
$\widetilde u_{j+1/2}=(\widetilde u_{j+1}+\widetilde u_j)/2$.
Let $P_K$ multiply the Fourier coefficient at $k$ by
$\mathbb I_{\{\|k\|_\infty\leq K\}}$, with $K=2$, and let $\mathbb P_k$ be
the discrete Leray multiplier with zero mean retained. The implemented
native-grid generator is
\begin{align}
 &\mathcal A_K(v)=-\mathcal F^{-1}
 \left[\mathbb P_k\mathbb I_{\{\|k\|_\infty\leq K\}}
 \widehat{(v\cdot\nabla)v}(k)\right]+\nu\Delta v.
 \label{eq:pinoimplementedrhs}
\end{align}
The pointwise product is formed on the $12^3$ grid without zero padding. The
midpoint residual is therefore
\begin{align}
 &r_j=\frac{\widetilde u_{j+1}-\widetilde u_j}{\delta t_j}
 -\mathcal A_K(\widetilde u_{j+1/2}).
 \label{eq:pinomidresidual}
\end{align}
The Navier--Stokes and divergence losses are the component- and
scalar-averaged mean squares
\begin{align}
 &\mathcal L_{\rm NS}=\frac14\sum_{j=0}^{3}\mathfrak M_3(r_j),
 \qquad
 \mathcal L_{\rm div}=\frac15\sum_{j=0}^{4}
 \mathfrak M_1(\nabla\cdot\widetilde u_j).
 \label{eq:pinopdeloss}
\end{align}
With
$E_{\rm kin}(v)=(2n^3)^{-1}\sum_{c,x}|v_c(x)|^2$ and
$X(v)=\sum_{c,k}|k|^2|\widehat v_c(k)|^2$, the energy penalty is
\begin{align}
 &\mathcal L_{\rm energy}=\frac14\sum_{j=0}^{3}
 \left[\frac{E_{\rm kin}(\widetilde u_{j+1})-E_{\rm kin}(\widetilde u_j)}{\delta t_j}
 +\nu X(\widetilde u_{j+1/2})\right]^2.
 \label{eq:pinoenergyloss}
\end{align}
Here $E_{\rm kin}$ is one half of the squared $L^2$ norm. The moments
$E_N$ and the residual-action calculation instead use the full squared norm,
which explains the different factor in their energy identities.
To define the remaining losses, let $\mathcal A_\infty$ denote the same
native-grid generator with the nonlinear cutoff mask omitted. The difference
$d_j=\mathcal A_\infty(\widetilde u_j)-\mathcal A_K(\widetilde u_j)$ is an
on-grid cutoff-sensitivity diagnostic. Because the product is evaluated on
the native grid, it does not enclose the complete continuous quadratic output.
The de-aliasing and spectral losses used in training are
\begin{align}
 \mathcal L_{\rm dealias}
 &=\frac15\sum_{j=0}^{4}\mathfrak M_3(d_j),
 \label{eq:pinodealiasloss}\\
 \mathcal L_{\rm spectral}
 &=\frac1{15n^3}\sum_{j=0}^{4}\sum_{c=1}^{3}\sum_k |k|^2
 |\widehat{\widetilde u_j}(k)-\widehat u_j(k)|^2.
 \label{eq:pinospectralloss}
\end{align}
Finally, define the predicted outer-band fraction
\begin{align}
 &\mathcal T_K(v)=
 \frac{\displaystyle\sum_{c=1}^{3}
 \sum_{\lceil2K/3\rceil\leq\|k\|_\infty\leq K}
 |k|^4|\widehat v_c(k)|^2}
 {\displaystyle\max\!\left\{\sum_{c=1}^{3}\sum_k
 |k|^4|\widehat v_c(k)|^2,10^{-20}\right\}},
 \qquad
 \mathcal L_{\rm tail}=\frac15\sum_{j=0}^{4}\mathcal T_K(\widetilde u_j).
 \label{eq:pinotailloss}
\end{align}
The denominator in $\mathcal T_K$ includes every native-grid Fourier mode,
whereas the numerator includes only the outer portion of the retained cube.
Consequently $\mathcal T_K$ is a training penalty, not an error against a
known true tail. The factors in Equations~\eqref{eq:pinodiscretemeans}--
\eqref{eq:pinotailloss} reproduce the component and grid averages in the
trained implementation. Substitution of
Equations~\eqref{eq:pinofieldloss}, \eqref{eq:pinopdeloss},
\eqref{eq:pinoenergyloss}, \eqref{eq:pinodealiasloss},
\eqref{eq:pinospectralloss}, and \eqref{eq:pinotailloss} into
Equation~\eqref{eq:pinoloss} yields the complete training objective used in
the main text.

At the five network output times $t_j$, the action observable used for model
evaluation and ranking is the trapezoidal proxy
\begin{align}
 &J_{0.2}^{(5)}(v)=\sum_{j=0}^{3}
 \frac{t_{j+1}-t_j}{2}\{R(v_j)^{2/3}+R(v_{j+1})^{2/3}\}.
 \label{eq:pinofivepointproxy}
\end{align}
This quantity is not asserted to approximate the continuous action with the
same accuracy as the RK4 state trajectory. For ranking, the predicted
composite score is
\begin{align}
 &s_{\rm pred}=\log\!\left(J_{0.2}^{\rm pred}+10^{-12}\right)
 +0.15\log\!\left(q_{\max}^{\rm pred}+10^{-12}\right),
 \label{eq:pinosearchscore}
\end{align}
where $J_{0.2}^{\rm pred}=J_{0.2}^{(5)}(\widetilde u)$ is the predicted
five-checkpoint residual-action proxy and
$q_{\max}^{\rm pred}$ is the maximum predicted retained-band tail. Sixteen
gradient starts take 31 optimisation steps and one final scoring pass, which
matches 512 forward evaluations per family. A repulsion penalty and the
minimum pair distance enforce candidate diversity. After gradient
optimisation, the 16 endpoint fields are ranked by predicted action alone.
The 512-member random and Sobol pools are likewise ranked by the fitted PINO
action prediction before six members are sent to the solver. These two
comparators therefore measure alternative proposal mechanisms under common
PINO screening, not unguided random or Sobol solver search. All methods
use the same number of candidate evaluations and solver rechecks. Gradient
search also requires backward evaluations, so the comparison fixes the
candidate budget rather than wall-clock cost.

The additional threshold-swept experiment in Figure~\ref{fig:pinomain}(c)
uses the action proxy alone. For $n=212$ OOD fields, selection fraction $p$, and
$k=\lceil pn\rceil$, let $T_k$ contain the $k$ largest Galerkin proxies and
$P_k$ contain the $k$ largest operator-predicted proxies. The retrieval recall
is
\begin{align}
 &\operatorname{Recall}(p)=\frac{|T_k\cap P_k|}{k}.
 \label{eq:pinorecall}
\end{align}
The sweep evaluates Equation~\eqref{eq:pinorecall} from $p=0.05$ to $0.30$
without refitting either model. The independently specified composite score in
Equation~\eqref{eq:pinosearchscore} gives top-decile recalls $0.955$ for FNO
and $0.864$ for PINO, with a paired 95\% resampling interval
$[-0.227,0.045]$ for the gain. This composite endpoint was specified
independently of the later action-only sweep and does not show an improvement
for PINO. The sweep therefore provides an exploratory diagnostic rather
than a replacement success criterion.

The paired uncertainty calculation resamples the 212 aligned OOD indices with
replacement 1,000 times. For each resample it recomputes the median of the
casewise FNO-minus-PINO metric difference. The displayed 95\% interval uses
the 2.5th and 97.5th percentiles of those 1,000 medians. A fixed seed makes the
calculation deterministic. It measures sensitivity to the fixed OOD cases for
one fitted training seed, not variability across independently trained
networks. All seven lower endpoints are positive. The same
paired resampling applied to the composite top-decile recall gives the interval
$[-0.227,0.045]$ stated above.

The dense action used in Section~\ref{sec:pinoresults} evaluates the same
integrand at every RK4 state. For $t_q=q\Delta t$, $q=0,1,\ldots,m$,
and $m\Delta t=0.2$, its quadrature is
\begin{align}
 J_{0.2}^{\rm dense}
 &=\Delta t\left[\frac12R(u^0)^{2/3}
 +\sum_{q=1}^{m-1}R(u^q)^{2/3}+\frac12R(u^m)^{2/3}\right].
 \label{eq:densequadraturedetail}
\end{align}
All 48 selected inputs use $(n,K,\Delta t)=(16,3,0.005)$,
$(20,4,0.0025)$, $(24,5,0.00125)$, and $(24,5,0.000625)$.
Each relative difference divides the absolute quadrature difference
by the finest dense value. In contrast, the original five-checkpoint comparison
refines the velocity trajectories while holding its quadrature nodes fixed.
This distinction explains the percent-scale proxy discrepancy alongside much
smaller dense-refinement differences.

\clearpage
\begin{figure}[!ht]
\centering
\includegraphics[width=\linewidth,height=0.40\textheight,keepaspectratio]{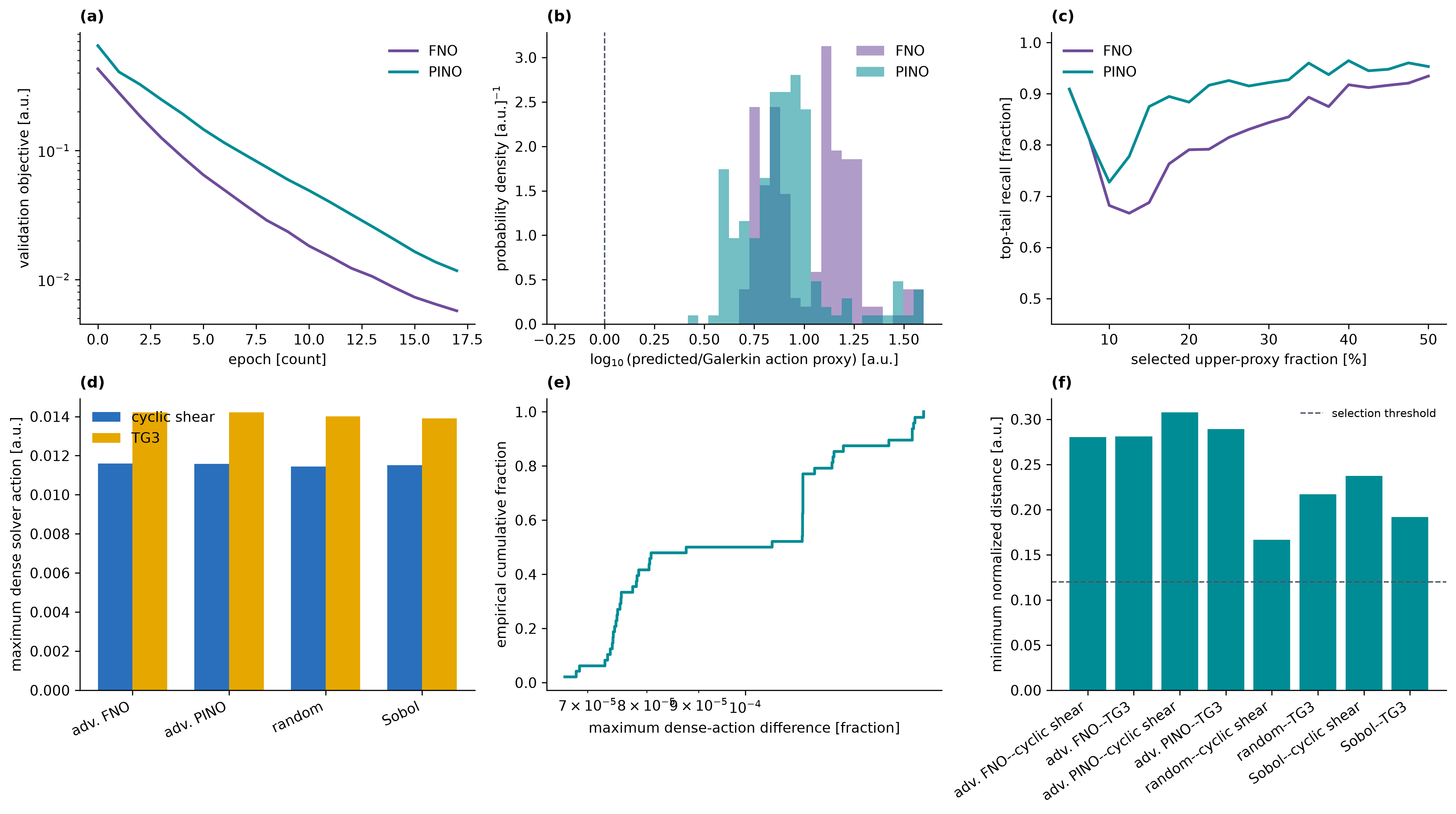}
\caption{Extended neural-operator diagnostics. (a) Validation objective across
18 training epochs for FNO and PINO. (b) OOD distributions of the logarithmic
ratio between predicted and Galerkin five-checkpoint action proxy. (c) Action-proxy upper-tail recall
from 5\% to 50\% selection. (d) Maximum $24^3$ Galerkin action among six
selected candidates for each acquisition method and initial-data family,
using dense quadrature at $\Delta t=0.000625$.
(e) Empirical cumulative distribution of the maximum relative dense-action
difference across three grid--time-step comparisons for 48 candidates. (f) Minimum
normalised pair distance within each six-candidate cohort, with the dashed line
marking the predeclared diversity threshold. Nondimensional objectives,
actions, time steps, ratios, and distances use arbitrary units (a.u.).}
\label{fig:pinoappendix}
\end{figure}

Figure~\ref{fig:pinoappendix}(a) records stable optimisation for both models.
Panel (b) shows that PINO shifts the action-proxy ratio distribution towards
ideal calibration while retaining an upward bias. Panel (c) extends the main-text
recall sweep to half of the OOD set. Panel (d) shows the matched candidate-count
comparison across both families after dense quadrature. Panel (e) demonstrates
that every selected dense action remains tightly clustered across the three
solver comparisons. Panel (f)
shows that all eight cohorts exceed the $0.12$ diversity threshold. Together,
these diagnostics document an operator-guided search stage whose outputs
remain subject to direct Galerkin and exact mathematical
validation.

\section{Expanded algebraic, analytic and discrete transitions}
\label{app:expandedtransitions}

This appendix supplies intermediate calculations for the mathematical
methodology, field diagnostics and neural-operator methodology.
Sections~\ref{app:problemexpansion} and \ref{app:diniexpansion} derive
projection and error comparison. Sections~\ref{app:pathexpansion} and
\ref{app:momentexpansion} expand the signed path and spectral balance, and
Section~\ref{app:smoothnessexpansion} proves the all-order continuum passage.
Sections~\ref{app:tensorexpansion} and \ref{app:heatexpansion} connect tensor
geometry to the heat asymptotic. Sections~\ref{app:discreteexpansion},
\ref{app:neuralexpansion}, and \ref{app:statisticsexpansion} give the discrete,
network and statistical operations. Spatial integrals use
the normalised measure of Section~\ref{sec:mathmethods}. Repeated component
indices in this appendix run from one to three and denote a sum. A Fourier
wave vector $k$ is an integer triple, whereas $|k|$ and $|k|_\infty$ denote
its Euclidean and maximum-component norms, respectively.

\subsection{Pressure, projection and the initial-value problem}
\label{app:problemexpansion}

Pressure elimination converts \eqref{eq:physicalproblem} into
Equations~\eqref{eq:galerkin}--\eqref{eq:fourierB} as follows. Write $\partial_j=\partial/\partial x_j$.
For a solenoidal velocity,
\begin{align}
 u_j\partial_ju_i
 &=\partial_j(u_ju_i)-u_i\partial_ju_j
 =\partial_j(u_ju_i),\\
 \partial_i(u_j\partial_ju_i)
 &=(\partial_i u_j)(\partial_ju_i)
   +u_j\partial_j(\partial_i u_i)
 =(\partial_i u_j)(\partial_ju_i).
 \label{eq:expandpressurediv}
\end{align}
Taking the divergence of the momentum equation gives
$\Delta p=-\partial_i\partial_j(u_i u_j)$. Its Fourier coefficients are
\begin{align}
 -|k|^2\widehat p(k)&=k_i k_j\widehat{u_i u_j}(k),
 &\widehat p(k)&=-\frac{k_i k_j}{|k|^2}\widehat{u_i u_j}(k)quad(k\ne0),
 &\widehat p(0)&=0.
 \label{eq:expandpressurefourier}
\end{align}
The last choice fixes the gauge, rather than imposing another condition on
velocity. Integrating the conservative momentum equation over the torus
gives $\partial_t\int u_i=0$, since every periodic derivative has zero
integral. Thus zero initial mean persists.

For $k\ne0$, put $Q_k=kk^\top/|k|^2$. Direct multiplication gives
$Q_k^\top=Q_k$ and $Q_k^2=k(k^\top k)k^\top/|k|^4=Q_k$. Hence
\begin{align}
 \mathbb P_k^\top&=\mathbb P_k,
 &\mathbb P_k^2&=(I-Q_k)^2=\mathbb P_k,
 &k^\top\mathbb P_k&=0,
 &\mathbb P_k(ik\widehat p)&=0.
 \label{eq:expandprojectionmatrix}
\end{align}
These identities hold for complex vectors because $k$ is real.
The complex pairing $\langle f,g\rangle=\int f\cdot\overline g$ is linear
in its first argument, and energy identities use its real part. Orthogonality gives
$|z|^2=|\mathbb P_kz|^2+|Q_kz|^2$, so $\mathbb P_k$ is contractive.
Its symbol commutes with $A$ and $P_N$.
Applying $P_N\mathbb P$ to the momentum equation gives
\begin{align}
 \partial_t\widehat u_N(k)+\nu|k|^2\widehat u_N(k)
 &=-i\mathbb P_k
 \sum_{\substack{p+q=k\\0<|p|_\infty,|q|_\infty\leq N}}
       (\widehat u_N(p)\cdot q)\widehat u_N(q),
 \quad 0<|k|_\infty\leq N.
 \label{eq:expandfiniteode}
\end{align}
Both maximum-component norms under the sum are positive and at most $N$.
At $k=0$, every summand has
$\widehat u_N(p)\cdot(-p)=0$, so no division by $|k|^2$ is needed.
Complex conjugation and $(p,q)\mapsto(-p,-q)$ preserve reality.
This is a finite real polynomial ODE on the conjugate-symmetric solenoidal
subspace. A polynomial vector field is locally Lipschitz and therefore has
a unique local solution. Appendix~\ref{app:foundations} bounds its energy
and hence every coordinate. On the resulting bounded set the vector field
is bounded and Lipschitz, so the solution extends past every finite
endpoint. This continuation proves fixed-$N$ global existence, whereas the continuum
argument additionally requires cutoff-uniform bounds.

For norm bookkeeping, $|k|\geq1$ implies
$|k|^{2s}\leq(1+|k|^2)^s\leq2^s|k|^{2s}$ for $s\geq0$.
The symbols $A_0(t),A_1(t)$ are scalar path bounds, not powers of $A$.
The moment $E=\|u\|_2^2$ is twice kinetic energy. The field
$E=B(U,H)+B(H,U)$ occurs only in the path calculation.
The global spectral residual $R$ differs from the local third
velocity-gradient invariant also denoted $R$ in the public-DNS figures.
The notation $\mathcal R_J$ below distinguishes the latter. Likewise $\rho_k$ is a
spectral probability, whereas $\rho=\|h\|_{\mathcal A^2}$ is a radius.

\subsection{Complex norms, dissipation and the scalar comparison}
\label{app:diniexpansion}

Differentiating complex coefficient norms and retaining dissipation converts
\eqref{eq:detail017} into \eqref{eq:detail015}. For a continuous real function $q$, its upper right
Dini derivative is $D^+q(t)=\limsup_{h\downarrow0}[q(t+h)-q(t)]/h$.
For a differentiable complex vector $z$ satisfying
$z'+\nu\lambda z=f$, $\lambda\geq0$, and $z\ne0$,
\begin{align}
 \frac{\dd}{\dd t}|z|
 &=\frac{\operatorname{Re}(z'\cdot\overline z)}{|z|}
 =-\nu\lambda|z|+
   \frac{\operatorname{Re}(f\cdot\overline z)}{|z|}
 \leq-\nu\lambda|z|+|f|.
 \label{eq:expandcomplexnorm}
\end{align}
At $z(t)=0$, $z(t+h)=hf(t)+o(h)$ gives $D^+|z|(t)=|f(t)|$,
which is the same inequality. Thus zero Fourier coefficients cause no
division problem. For $Y_j=\sum_k|k|^j|\widehat w_N(k)|$ and $Y=Y_0$,
\begin{align}
 \|B(v,w_N)\|_{\mathcal A^0}&\leq A_0Y_1,
 &\|B(w_N,v)\|_{\mathcal A^0}&\leq A_1Y,
 &\|B(w_N,w_N)\|_{\mathcal A^0}&\leq YY_1,\\
 D^+Y+\nu Y_2&\leq d+A_1Y+(A_0+Y)Y_1.
 \label{eq:expanderrornorm}
\end{align}
The derivative falls on the second argument of $B$, as the factor $q$ in
\eqref{eq:fourierB} shows. Cauchy--Schwarz applied to
$|\widehat w(k)|^{1/2}$ and $|k||\widehat w(k)|^{1/2}$ gives
$Y_1^2\leq YY_2$. Also $Y_2\geq Y$ because $|k|^2\geq1$.
For $b=A_0+Y\geq0$, completing a square gives
\begin{align}
 b\sqrt{YY_2}-\frac\nu2Y_2
 &=\frac{b^2}{2\nu}Y-
   \frac12\left(\sqrt{\nu Y_2}-\frac b{\sqrt\nu}\sqrt Y\right)^2
 \leq\frac{b^2}{2\nu}Y,\\
 -\nu Y_2+bY_1
 &\leq-\frac\nu2Y+\frac{(A_0+Y)^2}{2\nu}Y,\\
 (A_0+Y)^2Y&=A_0^2Y+2A_0Y^2+Y^3.
 \label{eq:expandsquarecomparison}
\end{align}
Substitution proves every coefficient of \eqref{eq:detail015}, including
$-\nu/2$. The coefficient-norm argument above also proves the inequalities at $Y=0$.

Let $F(t,Y)$ denote the resulting cubic right-hand side. On a bounded
interval of nonnegative $Y$, its bounded coefficients make $F$ uniformly
Lipschitz in $Y$, with constant $L$. For an absolutely continuous
supersolution $y$, the positive part $z=(Y-y)_+$ satisfies $z'\leq Lz$
almost everywhere. Then $(e^{-Lt}z)'\leq0$ and $z(0)=0$ imply $z(t)=0$.
This justifies the comparison after the exact slab inequalities in
Appendix~\ref{app:slabalgorithm}. The linear coefficient of the cubic
need not be positive.

\subsection{All ordered products in the signed path}
\label{app:pathexpansion}

Expanding every ordered product in \eqref{eq:path} recovers
\eqref{eq:residual} and the norm table in Appendix~\ref{app:algebra}. Put $c=\alpha-\gamma/2$, so
$v=cU-\beta V+\delta H$. Bilinearity gives nine terms before grouping:
\begin{align}
 B(v,v)={}&c^2B(U,U)-c\beta B(U,V)+c\delta B(U,H)
 -\beta cB(V,U)+\beta^2B(V,V)\notag\\
 &-\beta\delta B(V,H)+\delta cB(H,U)
 -\delta\beta B(H,V)+\delta^2B(H,H)\notag\\
 ={}&c^2V-c\beta C+c\delta E+\beta^2D-\beta\delta F+\delta^2G.
 \label{eq:expandnineproducts}
\end{align}
The scalar functions satisfy $\alpha'+\alpha=0$,
$\beta'+2\beta=\alpha^2$, $\gamma'+\gamma=\alpha\beta$, and
$\delta'+5\delta=\alpha\beta$. Since $AU=U$, $AV=2V$, $AH=5H$ and
$C=H-U/2$, their linear contribution is
\begin{align}
 v'+Av&=-\tfrac12\alpha\beta U-\alpha^2V+\alpha\beta H
       =-\alpha^2V+\alpha\beta C,\\
 c^2-\alpha^2&=-\alpha\gamma+\gamma^2/4,
 &\alpha\beta-c\beta&=\beta\gamma/2,
 &c\delta&=\alpha\delta-\gamma\delta/2.
 \label{eq:expandpathlinear}
\end{align}
Adding the linear and bilinear expansions gives \eqref{eq:residual} while
retaining every higher-degree term.

For $k=(1,0,1)$ in $V=B(U,U)$, the pair $p=(1,0,0)$, $q=(0,0,1)$ has
$\widehat U(p)=(0,0,-i/2)$ and $\widehat U(q)=(0,-i/2,0)$. Thus
\begin{align}
 i(\widehat U(p)\cdot q)\widehat U(q)&=(0,-i/4,0),
 &\widehat U(q)\cdot p&=0,\\
 \mathbb P_{(1,0,1)}(0,-i/4,0)&=(0,-i/4,0).
 \label{eq:expandmodeexample}
\end{align}
The reversed pair contributes zero and this coefficient has squared norm
$1/16$. The four exponential terms in each of the three components of $V$
give twelve coefficients of magnitude $1/4$. Differentiation also gives
\begin{align}
 (V\cdot\nabla)U
 &=(\sin x\cos y\cos z,\cos x\sin y\cos z,
                      \cos x\cos y\sin z)\notag\\
 &=\nabla(-\cos x\cos y\cos z),
 \label{eq:expandvanishingtree}
\end{align}
so \eqref{eq:expandprojectionmatrix} proves $B(V,U)=0$.

The table gives all coefficient-norm multiplicities after signed vectors at
equal output modes have been added. A row lists $\lambda=|k|^2$, squared
coefficient norm $q$, and multiplicity $m$. Unlisted coefficients vanish.
The seed coefficients in \eqref{eq:detail171}--\eqref{eq:detail172} and
the ordered rule \eqref{eq:fourierB} determine their vectors.
\begin{center}\small
\begin{tabular}{@{}rrrr|rrrr@{}}
\toprule
Field&$\lambda$&$q$&$m$&Field&$\lambda$&$q$&$m$\\
\midrule
$U$&1&$1/4$&6&$E$&2&$1/256$&12\\
$V$&2&$1/16$&12&$E$&6&$7/768$&24\\
$C$&1&$1/16$&6&$E$&10&$1/256$&12\\
$C$&5&$1/64$&12&$F$&5&$9/1280$&12\\
$H$&5&$1/64$&12&$F$&9&$29/9216$&24\\
$D$&6&$1/192$&24&$F$&11&$9/5632$&24\\
 &&&&$G$&6&$7/12288$&24\\
 &&&&$G$&14&$27/28672$&24\\
\bottomrule
\end{tabular}
\end{center}
Each field has $\|f\|_{\mathcal A^j}=\sum m\lambda^{j/2}\sqrt q$.
In particular,
\begin{align}
 \|D\|_{\mathcal A^0}&=24/\sqrt{192}=\sqrt3,\\
 \|E\|_{\mathcal A^0}&=24/16+24\sqrt{7/768}=(3+\sqrt{21})/2,\\
 \|F\|_{\mathcal A^0}&=12\sqrt{9/1280}+24\sqrt{29/9216}
                      +24\sqrt{9/5632}\notag\\
                    &=9\sqrt5/20+\sqrt{29}/4+9\sqrt{22}/44,\\
 \|G\|_{\mathcal A^0}&=24\sqrt{7/12288}+24\sqrt{27/28672}
                      =2\sqrt{21}/7.
 \label{eq:expandradicals}
\end{align}
Taking absolute values before adding equal-mode vectors would give a
different bound. Applying the triangle inequality only to the six grouped
fields in \eqref{eq:residual}, with $\alpha,\beta,\gamma,\delta\geq0$,
gives \eqref{eq:r}. Its coefficient $3/2$ multiplying $\beta\gamma$ is
$\|C\|_{\mathcal A^0}/2$.

\subsection{Differentiating the centred spectral residual}
\label{app:momentexpansion}

Differentiating the three moments in \eqref{eq:Rmoments} gives the source
balance used in the physical diagnostics. For $E=\|u\|_2^2>0$, let $X=\langle Au,u\rangle$,
$Z=\|Au\|_2^2$, $W=\|A^{3/2}u\|_2^2$, $\mu=X/E$, and
$b=-P_NB(u,u)$. Define
$\Pi_1=\operatorname{Re}\langle Au,b\rangle$ and
$\Pi_2=\operatorname{Re}\langle A^2u,b\rangle$.
From $u'=-\nu Au+b$ and self-adjointness,
\begin{align}
 E'&=-2\nu X,\qquad X'=-2\nu Z+2\Pi_1,\qquad Z'=-2\nu W+2\Pi_2,\\
 \mu'&=X'/E-XE'/E^2,\qquad
 (X^2/E)'=2XX'/E-X^2E'/E^2,\\
 \tfrac12R'
 &=-\nu W+\Pi_2-\mu(-2\nu Z+2\Pi_1)-\nu\mu^2X\notag\\
 &=-\nu(W-2\mu Z+\mu^2X)+(\Pi_2-2\mu\Pi_1).
 \label{eq:expandmomentderivative}
\end{align}
Consequently
\begin{align}
 D_R&=\|A^{1/2}(A-\mu)u\|_2^2
     =\sum_k|k|^2(|k|^2-\mu)^2|\widehat u(k)|^2\notag\\
    &=W-2\mu Z+\mu^2X,\\
 \mathcal F&=\Pi_2-2\mu\Pi_1,\qquad
 \tfrac12R'+\nu D_R=\mathcal F.
 \label{eq:expandresidualbalance}
\end{align}
The derivative of the minimizer introduces no extra uncancelled term,
because its stationarity condition is
$\langle(A-\mu)u,u\rangle=X-\mu E=0$.
The centred sum avoids subtracting nearly equal $Z$ and $X^2/E$ in
floating-point arithmetic. At zero set $R=D_R=0$. The bound $0\leq R\leq Z$
proves continuity there in $H^2$, not in an unspecified weaker topology.

\subsection{All-order regularity, compactness and uniqueness}
\label{app:smoothnessexpansion}

Weighted convolution estimates propagate every Sobolev order and justify
the compactness step in \eqref{eq:detail058}--\eqref{eq:detail059} and
Appendix~\ref{app:pdepassage}. Comparison supplies a cutoff-independent
function $M(t)$ with $\|u_N(t)\|_{\mathcal A^0}\leq M(t)$.
Before acceptance use $M=A_0+y$, and afterwards use the Wiener barrier
$M\leq\eta<\nu$. Thus $\int_0^T M^2<\infty$ for every finite $T$.

The triangle inequality in $\ell^2$ proves the convolution inequality:
\begin{align}
 \left(\sum_k\left|\sum_p a_p b_{k-p}\right|^2\right)^{1/2}
 &\leq\sum_p|a_p|\left(\sum_k|b_{k-p}|^2\right)^{1/2}
 =\|a\|_{\ell^1}\|b\|_{\ell^2}.
 \label{eq:expandyoungsequence}
\end{align}
For integer $m\geq1$, convexity gives
$|p+q|^m\leq2^{m-1}(|p|^m+|q|^m)$.
Also $|z\otimes w|_F^2=\sum_{i,j}|z_i|^2|w_j|^2=|z|^2|w|^2$.
Applying \eqref{eq:expandyoungsequence} to the two weighted convolutions gives
\begin{align}
 \|u\otimes u\|_{\dot H^m}
 &\leq2^{m-1}\bigl(\|u\|_{\dot H^m}\|u\|_{\mathcal A^0}
                 +\|u\|_{\mathcal A^0}\|u\|_{\dot H^m}\bigr)
 =2^m\|u\|_{\mathcal A^0}\|u\|_{\dot H^m}.
 \label{eq:expandweightedproduct}
\end{align}
Set $H_m=\|u_N\|_{\dot H^m}^2$. Moving the Fourier derivative in
$\nabla\cdot(u_N\otimes u_N)$ to the energy-test factor gives
\begin{align}
 \tfrac12H_m'+\nu H_{m+1}
 &\leq H_{m+1}^{1/2}\|u_N\otimes u_N\|_{\dot H^m}
 \leq2^m M H_m^{1/2}H_{m+1}^{1/2},\\
 2^m M H_m^{1/2}H_{m+1}^{1/2}
 &\leq\tfrac\nu2H_{m+1}+\frac{4^m}{2\nu}M^2H_m,\\
 H_m'+\nu H_{m+1}&\leq\frac{4^m}{\nu}M^2H_m,\\
 H_m(t)&\leq H_m(0)
 \exp\!\left(\frac{4^m}{\nu}\int_0^tM(s)^2\dd s\right).
 \label{eq:expandallorders}
\end{align}
For the last step differentiate
$H_m(t)\exp[-(4^m/\nu)\int_0^tM^2]$ and integrate its nonpositive
derivative. These bounds are uniform in $N$ for each fixed smooth datum.
They do not bound every high Sobolev norm uniformly over the whole
$\mathcal A^2$ perturbation ball. Smoothness makes each $H_m(0)$ finite,
which is the quantifier needed by the theorem.

The Fourier Cauchy--Schwarz inequality gives the embedding explicitly:
\begin{align}
 \|f\|_\infty
 &\leq\sum_k|\widehat f(k)|
 \leq\left[\sum_k(1+|k|^2)^{-2}\right]^{1/2}\|f\|_{H^2}.
 \label{eq:expandembedding}
\end{align}
The cubical ring $|k|_\infty=r$ contains
$(2r+1)^3-(2r-1)^3=24r^2+2$ points, each contributing at most
$(1+r^2)^{-2}$. Comparison with $\sum_{r\geq1}r^{-2}$ proves finiteness.
The main-text $H^1$ and $L^2_tH^2$ bounds now imply
\begin{align}
 \|\partial_tu_N\|_{L^2_tL^2_x}
 &\leq\nu\|Au_N\|_{L^2_tL^2_x}
 +C\|u_N\|_{L^2_tH^2_x}\|u_N\|_{L^\infty_tH^1_x}<C_T,\\
 |\widehat u_N(k,t)-\widehat u_N(k,s)|
 &\leq|t-s|^{1/2}\|\partial_tu_N\|_{L^2_tL^2_x}.
 \label{eq:expandtimecompactness}
\end{align}
For finitely many $k$, first choose subsequences on successive finite
dyadic time grids, using boundedness of their finitely many real and
imaginary coordinates. Take the diagonal subsequence. For a given
$\epsilon>0$, choose a grid spacing for which the common modulus is below
$\epsilon/3$. Convergence at the finitely many grid points then makes
two subsequence elements differ by less than $\epsilon/3$ there. Comparing
an arbitrary time to its nearest grid point bounds their difference by
three such terms, hence by $\epsilon$. The subsequence is uniformly Cauchy
on $[0,T]$ and therefore uniformly convergent. Select nested
subsequences for larger cubes and take their diagonal. The tail satisfies
\begin{align}
 \int_0^T\|(I-P_K)u_N\|_{H^1}^2\dd t
 &\leq(1+K^2)^{-1}\int_0^T\|u_N\|_{H^2}^2\dd t.
 \label{eq:expandcompacttail}
\end{align}
Finite-mode convergence followed by $K\to\infty$ proves strong
$L^2_tH^1_x$ convergence. The quadratic term passes to the limit because
\begin{align}
 u_N\otimes u_N-u\otimes u
 &=(u_N-u)\otimes u_N+u\otimes(u_N-u),\\
 \|u_N\otimes u_N-u\otimes u\|_{L^1_tL^1_x}
 &\leq\|u_N-u\|_{L^2_tL^2_x}
      (\|u_N\|_{L^2_tL^2_x}+\|u\|_{L^2_tL^2_x})\longrightarrow0.
 \label{eq:expandquadraticlimit}
\end{align}
Testing its divergence against a smooth field moves the derivative onto
that field. Uniform coefficient convergence includes $0$ and $T_*$.
At $T_*$, first pass to the limit in
$\sum_{|k|_\infty\leq K}|\widehat u_N(k,T_*)|$, then increase $K$.
Each finite sum is at most $\eta$, so its nonnegative increasing limit
is at most $\eta$. This proves the endpoint Wiener bound.

The same argument using \eqref{eq:expandallorders} at order $m+1$ gives
convergence in $C([0,T],H^m)$ at every integer $m$. Apply
\eqref{eq:expandembedding} to spatial derivatives to obtain their continuity.
The momentum equation and Fourier pressure formula give time derivatives
inductively. This supplies the asserted global smooth solution.

For uniqueness let $w=u-v$. The nonlinear difference is
$(u\cdot\nabla)w+(w\cdot\nabla)v$. The first term integrates to zero
against $w$. Integrating the second term by parts gives
\begin{align}
 \tfrac12\frac{\dd}{\dd t}\|w\|_2^2+\nu\|\nabla w\|_2^2
 &=-\int w_iw_j\partial_jv_i=\int v_iw_j\partial_jw_i\notag\\
 &\leq\|v\|_\infty\|w\|_2\|\nabla w\|_2\notag\\
 &\leq\tfrac\nu2\|\nabla w\|_2^2
       +\frac{\|v\|_\infty^2}{2\nu}\|w\|_2^2.
 \label{eq:expanduniqueness}
\end{align}
The integrating factor and $w(0)=0$ give $w=0$. This calculation displays
the viscosity spent by Young's inequality. Uniqueness among the smooth
solutions constructed here already removes subsequence dependence.

\subsection{Velocity gradients, curl and physical production}
\label{app:tensorexpansion}

The curl equation and matrix contractions below give the local production
and invariant identities in Equations~\eqref{eq:omegaU}--\eqref{eq:nonlinearU},
\eqref{eq:jhtdbgeometry} and \eqref{eq:detail176}. Here
$J_{ij}=\partial_j u_i$, $S=(J+J^\top)/2$ and
$\Omega=(J-J^\top)/2$. The alternating symbol $\epsilon_{ijk}$ is zero
when indices repeat, $+1$ for even permutations of $(1,2,3)$ and $-1$
for odd permutations. Contracting
$\epsilon_{ijk}\epsilon_{k\ell m}
=\delta_{i\ell}\delta_{jm}-\delta_{im}\delta_{j\ell}$ gives
\begin{align}
 (u\times\omega)_i
 &=\epsilon_{ijk}u_j\epsilon_{k\ell m}\partial_\ell u_m
 =u_j\partial_i u_j-u_j\partial_j u_i\notag\\
 &=\partial_i(|u|^2/2)-(u\cdot\nabla)u_i.
 \label{eq:expandvectoridentity}
\end{align}
This proves the vector identity used for the Beltrami centre, including its
sign. A second contraction gives
\begin{align}
 [\nabla\times(u\times\omega)]_i
 &=\partial_j(u_i\omega_j-u_j\omega_i)
 =\omega_j\partial_ju_i-u_j\partial_j\omega_i,
 \label{eq:expandcurltransport}
\end{align}
since $\partial_j u_j=0$ and
$\partial_j\omega_j=\epsilon_{j\ell m}\partial_j\partial_\ell u_m=0$.
Curling \eqref{eq:physicalproblem} gives
$\partial_t\omega+(u\cdot\nabla)\omega=(\omega\cdot\nabla)u+\nu\Delta\omega$.
Antisymmetry implies
$\omega^\top\Omega\omega=-\omega^\top\Omega\omega=0$.
Dotting the vorticity equation with $\omega$ therefore gives
\begin{align}
 (\partial_t+u\cdot\nabla)\frac{|\omega|^2}{2}
 &=\omega\cdot S\omega+\nu\Delta\frac{|\omega|^2}{2}
   -\nu\sum_{i,j}(\partial_j\omega_i)^2,\\
 \frac12\frac{\dd}{\dd t}\int|\omega|^2
 &=\int\omega\cdot S\omega-\nu\int|\nabla\omega|^2.
 \label{eq:expandvorticitybalance}
\end{align}
The diffusion term follows by expanding
$\omega_i\Delta\omega_i
=\partial_j(\omega_i\partial_j\omega_i)-(\partial_j\omega_i)^2$.
For solenoidal Fourier coefficients,
$|ik\times\widehat u|^2=|k|^2|\widehat u|^2-|k\cdot\widehat u|^2
=|k|^2|\widehat u|^2$, so $\int|\omega|^2=X$.
Positive stretching is local enstrophy production, not positive total
production after viscosity, and not positivity of the spectral source
$\mathcal F$.

The matrix $\Omega$ acts by $\Omega z=\omega\times z/2$.
The vector triple product gives
\begin{align}
 \Omega^2&=(\omega\omega^\top-|\omega|^2I)/4,
 &\operatorname{tr}\Omega^2&=-|\omega|^2/2,\\
 \operatorname{tr}J^2&=\operatorname{tr}S^2+\operatorname{tr}\Omega^2
                    =|S|_F^2-|\omega|^2/2,\\
 \operatorname{tr}J^3
 &=\operatorname{tr}S^3+3\operatorname{tr}(S\Omega^2)\notag\\
 &=\operatorname{tr}S^3+\tfrac34\omega\cdot S\omega
                 -\tfrac34|\omega|^2\operatorname{tr}S.
 \label{eq:expandgradienttraces}
\end{align}
Cyclicity makes the three cubic terms with two $\Omega$ factors equal.
The terms with one $\Omega$ have trace zero because $S^2$ is symmetric
and $\Omega$ is skew. Also $\operatorname{tr}\Omega^3=0$.
Incompressibility makes $\operatorname{tr}S=0$.
For trace-zero eigenvalues, expanding
$(\lambda_1+\lambda_2+\lambda_3)^2=0$ gives
$\sum_{i<j}\lambda_i\lambda_j=-\frac12\sum_i\lambda_i^2$.
The factorisation
$a^3+b^3+c^3-3abc=(a+b+c)(a^2+b^2+c^2-ab-bc-ca)$ gives
$\sum_i\lambda_i^3=3\lambda_1\lambda_2\lambda_3$.
Thus the characteristic polynomial is $\lambda^3+Q_J\lambda+\mathcal R_J$,
where $Q_J=-\operatorname{tr}J^2/2$ and
$\mathcal R_J=-\operatorname{tr}J^3/3=-\det J$.
For a measured gradient with nonzero trace $d_J$, its actual second and
third characteristic coefficients instead are
$(d_J^2-\operatorname{tr}J^2)/2$ and
$-[d_J^3-3d_J\operatorname{tr}J^2+2\operatorname{tr}J^3]/6$.
This distinction matters when interpreting discrete divergence defects.

For the cyclic centre the three stretching contributions are
\begin{align}
 2S_{12}\omega_1\omega_2&=a\cos y\,(-a\cos z)(-a\cos x)
                        =a^3\cos x\cos y\cos z,\\
 2S_{13}\omega_1\omega_3&=a\cos x\,(-a\cos z)(-a\cos y)
                        =a^3\cos x\cos y\cos z,\\
 2S_{23}\omega_2\omega_3&=a\cos z\,(-a\cos x)(-a\cos y)
                        =a^3\cos x\cos y\cos z.
 \label{eq:expandstretchpairs}
\end{align}
Their sum gives \eqref{eq:detail149}. Rowwise Cauchy--Schwarz gives
$|Sz|^2\leq\sum_i(\sum_jS_{ij}^2)(\sum_j|z_j|^2)=|S|_F^2|z|^2$,
which proves $\|S\|_{\rm op}\leq|S|_F$ in \eqref{eq:detail153}.
Applying the arithmetic--geometric mean inequality to the nonnegative
$r_i^2$ proves the bound in \eqref{eq:detail154}.

If $S=E\operatorname{diag}(\lambda_i)E^\top$, vorticity coordinates in
that basis are $E^\top\omega$, not $E\omega$. Component $i$ equals
$\sum_jE_{ji}\omega_j$. For $\omega\ne0$,
$\sum_i|\widehat\omega\cdot e_i|^2=|E^\top\widehat\omega|^2=1$.
For a repeated eigenvalue with index set $I$, the basis-independent
quantity is $\sum_{i\in I}|\widehat\omega\cdot e_i|^2$.
This explains the column contraction and its degeneracy qualification in
Appendix~\ref{app:physicalidentities}.

\subsection{Uniform heat expansion and nonlinear remainder}
\label{app:heatexpansion}

Heat smoothing controls the nonlinear remainder in
Equations~\eqref{eq:weightedheatbilinear}--\eqref{eq:secondweightedremainder}
uniformly in time, which justifies integration of the asymptotic action.
Integer dilation $n\geq1$ preserves periodicity. Setting
$u_n(t,x)=v(\tau,y)$ with $y=nx$, $\tau=\nu n^2t$ gives
$\partial_tu_n=\nu n^2\partial_\tau v$, $\Delta_xu_n=n^2\Delta_yv$,
and $(u_n\cdot\nabla_x)u_n=n(v\cdot\nabla_y)v$.
Division by $\nu n^2$ gives $\varepsilon=(\nu n)^{-1}$.
Normalised integrals are unchanged under integer periodic dilation.

Using
$(1+|p+q|^2)^{s/2}\leq C_s[(1+|p|^2)^{s/2}+(1+|q|^2)^{s/2}]$
in the proof of \eqref{eq:expandweightedproduct}, and using the finite sum
$\sum_k(1+|k|^2)^{-s}$ for $s>3/2$, gives
$\|f\otimes g\|_{H^s}\leq C_s\|f\|_{H^s}\|g\|_{H^s}$.
The remaining divergence derivative costs $|k|e^{-r|k|^2}$.
Write $|k|^2=1+z$, use $\sqrt{1+z}\leq1+\sqrt z$, and maximise
$\sqrt z e^{-rz}$ at $z=1/(2r)$. This proves
\begin{align}
 |k|e^{-r|k|^2}&\leq e^{-r}[1+(2er)^{-1/2}],\\
 \|e^{-rA}B(f,g)\|_{H^s}
 &\leq C_s e^{-r}(1+r^{-1/2})\|f\|_{H^s}\|g\|_{H^s}.
 \label{eq:expandheatsmoothing}
\end{align}
The output has zero mean because the nonlinearity is a divergence.
The time integral in \eqref{eq:weightedduhamel} gives a finite constant
$B_s$ such that the bilinear Duhamel operator $\mathcal K$ satisfies
$\|\mathcal K(f,g)\|_{X_s}\leq B_s\|f\|_{X_s}\|g\|_{X_s}$.
Put $a_s=\|U\|_{H^s}$. For $\|v\|_{X_s},\|w\|_{X_s}\leq2a_s$,
\begin{align}
 \|v_0-\varepsilon\mathcal K(v,v)\|_{X_s}
 &\leq a_s+4|\varepsilon|B_sa_s^2,\\
 \|\mathcal K(v,v)-\mathcal K(w,w)\|_{X_s}
 &\leq B_s(\|v\|_{X_s}+\|w\|_{X_s})\|v-w\|_{X_s}.
 \label{eq:expandcontractionradius}
\end{align}
Here $v_0=e^{-\tau}U$. Choosing
$|\varepsilon|\leq(8B_sa_s)^{-1}$ keeps the image inside the ball and
makes the Lipschitz constant at most $1/2$. Successive iterates have
geometrically decreasing differences, so completeness gives a fixed point.
With $d=v-v_0$ and $v_1=-\mathcal K(v_0,v_0)$,
\begin{align}
 \|d\|_{X_s}&\leq4B_sa_s^2|\varepsilon|,\\
 v-v_0-\varepsilon v_1
 &=-\varepsilon\{\mathcal K(d,v)+\mathcal K(v_0,d)\},\\
 \|v-v_0-\varepsilon v_1\|_{X_s}
 &\leq12B_s^2a_s^3\varepsilon^2.
 \label{eq:expandsecondremainder}
\end{align}
This proves a uniform remainder, not merely a finite-time formal series.

Set $z=e^\tau v=U+\varepsilon w_1+\varepsilon^2w_2$, where
$w_1=-\tau e^{-\tau}V$ and $w_2$ is uniformly bounded in $H^s$.
Writing $q=(A-1)z$, exact minimisation gives
\begin{align}
 R(z)&=\|q\|_2^2-\frac{\langle q,z\rangle^2}{\|z\|_2^2},
 &q&=\varepsilon w_1+\varepsilon^2(A-1)w_2,\\
 \langle q,z\rangle
 &=\varepsilon\langle w_1,U\rangle+O(\varepsilon^2)=O(\varepsilon^2),\\
 \|q\|_2^2
 &=\varepsilon^2\|w_1\|_2^2+
 2\varepsilon^3\langle w_1,(A-1)w_2\rangle
 +\varepsilon^4\|(A-1)w_2\|_2^2,\\
 R(v)&=e^{-2\tau}R(z)
 =\tfrac34\varepsilon^2\tau^2e^{-4\tau}
      +O(|\varepsilon|^3e^{-2\tau}).
 \label{eq:expandvarianceasymptotic}
\end{align}
Here $AU=U$, $AV=2V$, $\langle U,V\rangle=0$ and
$\|V\|_2^2=12/16=3/4$. For sufficiently small $|\varepsilon|$,
$\|z\|_2\geq\|U\|_2/2$ justifies the denominator uniformly.
Also $\mu(z)=1+O(\varepsilon^2)$, so
$A^{1/2}(A-\mu)z=\varepsilon A^{1/2}w_1+O_{L^2}(\varepsilon^2)$.
Because $\|A^{1/2}w_1\|_2^2=2\|w_1\|_2^2$, the leading coefficient
of $D_R(v)$ is $3\varepsilon^2\tau^2e^{-4\tau}/2$.

The fractional-power remainder cannot be obtained by dividing by a leading
term that vanishes at $\tau=0$. For $0<p<1$, $x\geq0$, and $y>0$,
the derivative of $(x+y)^p-y^p$ is
$p[(x+y)^{p-1}-y^{p-1}]\leq0$. Its limit as $y\downarrow0$ is $x^p$,
so $(x+y)^p\leq x^p+y^p$. Applying this to the difference of two ordered
nonnegative numbers proves $|a^p-b^p|\leq|a-b|^p$.
At $p=2/3$, the integrated error is bounded by
$C\varepsilon^2\int_0^\infty e^{-4\tau/3}\dd\tau$. The leading integral is
\begin{align}
 \left(\tfrac34\right)^{2/3}|\varepsilon|^{4/3}
 \int_0^\infty\tau^{4/3}e^{-8\tau/3}\dd\tau
 &=\left(\tfrac34\right)^{2/3}|\varepsilon|^{4/3}
   \left(\tfrac38\right)^{7/3}\Gamma(7/3).
 \label{eq:expandgammachange}
\end{align}
Here $\Gamma(s)=\int_0^\infty r^{s-1}e^{-r}\dd r$ for $s>0$.
The substitution is $r=8\tau/3$, $\dd\tau=3\dd r/8$, and the relative
error is $O(|\varepsilon|^{2/3})$. The relations $R(u_n)=n^4R(v)$ and
$\dd t=\dd\tau/(\nu n^2)$ give $\nu^{-7/3}n^{-2/3}$ in
\eqref{eq:actionasymptotic}.
For the signed source, $R(v(0))=R(U)=0$ and $R(v(\tau))\to0$.
Integrate \eqref{eq:expandresidualbalance} and use
$\int_0^\infty\tau^2e^{-4\tau}\dd\tau=2/4^3=1/32$ to obtain
$3\varepsilon^2/64$ in scaled variables. Since $D_R(u_n)=n^6D_R(v)$,
multiplication by $\nu\dd t$ gives $n^4$, hence the coefficient
$3n^2/(64\nu^2)$ in \eqref{eq:fluxasymptotic}.

\subsection{Discrete transforms and time integration}
\label{app:discreteexpansion}

Discrete orthogonality and the RK4 expansion determine the transform
normalisation and time update in the numerical Methods and
Equations~\eqref{eq:detail222}--\eqref{eq:detail223}. Let
$G_n=\{2\pi j/n:j\in\{0,\ldots,n-1\}^3\}$ and let $\Lambda_n$
contain one integer representative per frequency residue class. For even
$n$ use $-n/2,\ldots,n/2-1$ on each axis. Then
\begin{align}
 \widehat f_k&=n^{-3}\sum_j f_j e^{-2\pi i k\cdot j/n},
 &f_j&=\sum_{k\in\Lambda_n}\widehat f_k e^{2\pi i k\cdot j/n},\\
 n^{-3}\sum_j e^{2\pi i(k-\ell)\cdot j/n}
 &=\prod_{d=1}^3\left[n^{-1}\sum_{j_d=0}^{n-1}
                  e^{2\pi i(k_d-\ell_d)j_d/n}\right]
 =\mathbb I_{\{k=\ell\}},\\
 n^{-3}\sum_j|f_j|^2&=\sum_{k\in\Lambda_n}|\widehat f_k|^2.
 \label{eq:expanddiscreteparseval}
\end{align}
Each one-dimensional sum is a geometric progression with ratio different
from one unless its frequency is zero modulo $n$. Expanding a pointwise
product gives
\begin{align}
 \widehat{fg}_k
 &=\sum_{p+q\equiv k\ ({\rm mod}\ n)}\widehat f_p\widehat g_q.
 \label{eq:expandalias}
\end{align}
If every factor has $|k|_\infty\leq K$, four factors have total frequency
norm at most $4K$. For $4K<n$, no nonzero such frequency equals a multiple
of $n$ in every coordinate. Its discrete mean equals its continuous mean.
This argument applies to cutoff trajectories, not unrestricted network
outputs.

For the projected ODE put $f=\mathcal L(u)$, $J_f=D\mathcal L(u)$ and
$H_f=D^2\mathcal L(u)$. These are derivatives of a finite-dimensional
vector field, not the physical gradient tensor. Their actions are
\begin{align}
 J_f z&=-\nu Az-P_K\{B(u,z)+B(z,u)\},\\
 H_f(z,w)&=-P_K\{B(z,w)+B(w,z)\},
 &D^3\mathcal L&=0.
 \label{eq:expandrkderivatives}
\end{align}
Taylor expansion is exact to second order in the argument of $\mathcal L$.
Writing $h=\Delta t$, the stages expand as
\begin{align}
 k_1={}&f,\\
 k_2={}&f+\tfrac h2J_ff+\tfrac{h^2}{8}H_f(f,f),\\
 k_3={}&f+\tfrac h2J_ff
 +h^2[\tfrac14J_f^2f+\tfrac18H_f(f,f)]\notag\\
 &+h^3[\tfrac1{16}J_fH_f(f,f)+\tfrac18H_f(f,J_ff)]+O(h^4),\\
 k_4={}&f+hJ_ff+h^2[\tfrac12J_f^2f+\tfrac12H_f(f,f)]\notag\\
 &+h^3[\tfrac14J_f^3f+\tfrac18J_fH_f(f,f)
                   +\tfrac12H_f(f,J_ff)]+O(h^4).
 \label{eq:expandrkstages}
\end{align}
For example $k_3=f+\frac h2J_fk_2+\frac{h^2}{8}H_f(k_2,k_2)$
produces each displayed coefficient. Exact time derivatives are
$u'=f$, $u''=J_ff$, $u'''=J_f^2f+H_f(f,f)$ and
$u''''=J_f^3f+J_fH_f(f,f)+3H_f(f,J_ff)$.
Substitution into the weights $1,2,2,1$ gives
\begin{align}
 u^{n+1}={}&u+hf+\tfrac{h^2}{2}J_ff
 +\tfrac{h^3}{6}[J_f^2f+H_f(f,f)]\notag\\
 &+\tfrac{h^4}{24}[J_f^3f+J_fH_f(f,f)+3H_f(f,J_ff)]+O(h^5).
 \label{eq:expandrkorder}
\end{align}
Thus the local error has order five for a fixed resolved ODE. This does not
give fourth-order convergence of a separately sampled action quadrature,
or cutoff-independent error constants.

The fixed-datum convergence example in Section~\ref{sec:physics} has the
explicit non-Beltrami datum
\begin{align}
 u_0(x,y,z)=\tfrac32\bigl(&-\cos(x+y)-\sin(2x+y),\notag\\
                         &\cos x+\cos(x+y)+2\sin(2x+y),\notag\\
                         &\cos(x+y)+\sin(2x+y)\bigr).
 \label{eq:expandfixedtriad}
\end{align}
A coefficient $a+ib$ at positive $k$ and its conjugate at $-k$ contribute
$2a\cos(k\cdot x)-2b\sin(k\cdot x)$. Applying this to
$(1,0,0),(1,1,0),(2,1,0)$ gives the displayed field.
Its first two component derivatives cancel and its third component has
zero $z$ derivative. It is solenoidal and two-dimensional in spatial
dependence with three velocity components, not a second genuinely
three-dimensional centre theorem.
For the vector-field visualisations, let $L=2\pi$, reduce a point $x$
modulo $L$, and put $b_j=\lfloor nx_j/L\rfloor$ and
$\xi_j=nx_j/L-b_j$. The periodic trilinear interpolant of grid values $v$
is
\begin{align}
 v_{\rm int}(x)
 &=\sum_{\delta\in\{0,1\}^3}
 \prod_{j=1}^3\bigl[(1-\xi_j)^{1-\delta_j}\xi_j^{\delta_j}\bigr]
 v_{(b+\delta)\bmod n}.
 \label{eq:expandtrilinear}
\end{align}
The weights are nonnegative and sum to
$\prod_j[(1-\xi_j)+\xi_j]=1$, so interpolation is a convex combination
of the eight cell corners. The curve equation is
$\dd x/\dd s=\pm v_{\rm int}(x)/|v_{\rm int}(x)|$, with zero direction
if the denominator is below $10^{-12}$. The parameter $s$ is arclength,
not physical evolution time. The renderer uses the RK4 stages above with
step $0.045$ for at most 105 steps in each direction.

Seed candidates lie on a stride-three subgrid. Its 70th-percentile magnitude
threshold defines the admissible set. After choosing the largest-magnitude
candidate, each next candidate maximises
$d_{\min}^2(0.35+0.65m/m_{\max})$, where $m$ is its magnitude and
$d_{\min}$ is its distance to the nearest chosen candidate in the periodic
metric. Each coordinate difference $d$ is replaced by $\min(|d|,L-|d|)$
before squaring and summing. The selected seeds are shifted by half a grid
cell in each coordinate. This specifies the deterministic placement rather
than inferring it from the picture. The centre comparison uses $n=64$,
whereas the selected dynamic-field comparison uses $n=48$.

\subsection{Network indices, loss factors and acquisition}
\label{app:neuralexpansion}

Indexed contractions and averaging factors make the network and objective
in Section~\ref{sec:pinomethod} and Appendix~\ref{app:neuralobjective}
explicit. For layer input $z_{i,x}$, the learned
branch computes $\widehat b_{o,k}=\sum_i W_{io,k}\widehat z_{i,k}$ on its
specified block and zero elsewhere, followed by the real inverse transform.
The pointwise branch is affine, $p_{o,x}=\sum_iP_{oi}z_{i,x}+b_o$.
The lift and output maps also have trainable biases. The layer is
$\operatorname{GELU}(\operatorname{GN}(b+p))$. For a four-channel group
$I_g$, the mean and variance in \eqref{eq:networkoperations} are
\begin{align}
 \mu_g&=(4n^3)^{-1}\sum_{c\in I_g,x}z_{c,x},
 &v_g&=(4n^3)^{-1}\sum_{c\in I_g,x}(z_{c,x}-\mu_g)^2.
 \label{eq:expandgroupnorm}
\end{align}
The denominator is the number of group entries, not that number minus one.
It differs from the training-set sample standard deviation at input.

The spectral branch uses an unnormalised forward transform, whereas
physical losses use \eqref{eq:expanddiscreteparseval}. For a field difference
$e$, Parseval gives
\begin{align}
 \mathfrak M_3(e)&=\tfrac13\sum_{c,k}|\widehat e_{c,k}|^2,\\
 \frac15\sum_{j=0}^4\mathfrak M_3(\nabla_{\!F} e_j)
 &=\frac1{15}\sum_{j,c,k}|k|^2|\widehat e_{j,c,k}|^2,\\
 \mathcal L_{\rm spectral}
 &=n^{-3}\left[\frac15\sum_{j=0}^4\mathfrak M_3(\nabla_{\!F} e_j)\right].
 \label{eq:expandlossnormalisation}
\end{align}
Here $\nabla_{\!F}e$ has components $\mathcal F_n^{-1}(ik_j\widehat e)$
without taking the real part. Its norm sums complex squared magnitudes
over derivative directions while dividing by three velocity components.
It differs from the real-part differentiator below at Nyquist modes.
The extra $n^{-3}$ comes from averaging an
already normalised Fourier array and must remain when reproducing the
trained objective. A unit check is $e_1=\sin x$, $e_2=e_3=0$ at all five
checkpoints. Two coefficients of magnitude $1/2$ give
$\mathcal L_{\rm field}=1/6$ and $\mathcal L_{\rm spectral}=1/(6n^3)$.

The native-grid differentiators are
$D_jv=\operatorname{Re}\mathcal F_n^{-1}(ik_j\widehat v)$ and
$\Delta_nv=\operatorname{Re}\mathcal F_n^{-1}(-|k|^2\widehat v)$.
Taking the real part fixes the even-grid Nyquist convention used by the
implementation. The nonlinear product is $\sum_jv_jD_jv$ before masking
and Leray projection. At $k=0$, the neural generator uses $\mathbb P_0=I$,
whereas the analytic mean-zero Galerkin problem removes that mode.
Neural outputs need not be solenoidal or mean zero.

Expanding and cancelling cross terms gives the midpoint identity
\begin{align}
 \frac{\|b\|_2^2-\|a\|_2^2}{2\delta t}
 &=\left\langle\frac{b-a}{\delta t},\frac{a+b}{2}\right\rangle.
 \label{eq:expandmidpointenergy}
\end{align}
It explains the energy penalty and its factor one half, but does not make
that penalty an exact balance for arbitrary aliased, non-solenoidal outputs.

The predicted action floors its denominator:
$\widetilde\mu=X/\max(E,10^{-20})$ and
$\widetilde R=\sum_{c,k}(|k|^2-\widetilde\mu)^2|\widehat v_{c,k}|^2$.
For $E\geq10^{-20}$ this is the exact centred grid residual. Below that
threshold it is the implemented regularised observable, not the exact
minimisation quotient. The proxy integrates
$\max(\widetilde R,0)^{2/3}$ at the five checkpoints.
For a tail unit check take $v_1=\sin(2x)+\sin(3x)$, $v_2=v_3=0$, $K=2$.
The numerator is $2^4/2=8$ and the denominator is
$(2^4+3^4)/2=97/2$, so $\mathcal T_2=16/97$, not one.

For acquisition let $r_i\in[0,1]^{10}$ be the 16 parameter vectors and
$d_{ij}=\|r_i-r_j\|_2/\sqrt{10}$. The minimised objective is
\begin{align}
 L_{\rm acquisition}
 &=-\frac1{16}\sum_{i=1}^{16}s_{\rm pred}(r_i)
 +\frac{0.08}{16\cdot15}\sum_{i\ne j}
       \exp[-(d_{ij}/0.12)^2].
 \label{eq:expandacquisition}
\end{align}
Adam takes 31 steps of size $0.03$, clamping each parameter to $[0,1]$
after each step. It does not differentiate through Galerkin rechecks or
exact comparison. Its componentwise moment updates are
$m_t=\beta_1m_{t-1}+(1-\beta_1)g_t$ and
$v_t=\beta_2v_{t-1}+(1-\beta_2)g_t^2$, with $m_0=v_0=0$.
Bias correction gives $\widehat m_t=m_t/(1-\beta_1^t)$,
$\widehat v_t=v_t/(1-\beta_2^t)$, followed by parameter increment
$-\eta\widehat m_t/(\sqrt{\widehat v_t}+10^{-8})$.
Both training and acquisition use $(\beta_1,\beta_2)=(0.9,0.999)$.
Training also clips the gradient as specified in
Appendix~\ref{app:networkdetail}. Fractional-part parameter maps are
differentiable away from wrap points. Their piecewise derivatives do not
establish global optimality of the search.

\subsection{Quadrature, quantiles and numerical claims}
\label{app:statisticsexpansion}

Interpolation and resampling yield the estimators in
Appendix~\ref{app:statistics} and the quadratures in
Equations~\eqref{eq:pinofivepointproxy}--\eqref{eq:densequadraturedetail}.
The linear interpolant between integrand values $f_j,f_{j+1}$ on an
interval of width $h_j$ is $f_j+(f_{j+1}-f_j)s/h_j$. Its integral from
$0$ to $h_j$ is $h_j(f_j+f_{j+1})/2$. Summing proves the trapezoidal
formulas. Refining the ODE states at fixed five output times does not refine
this quadrature. Also $R^{2/3}$ need not have two bounded derivatives at a
zero of $R$, so the smooth-integrand order estimate cannot be assumed there.

For sorted observations $x_{(0)}\leq\cdots\leq x_{(n-1)}$, a linearly
interpolated quantile at probability $p$ uses $r=(n-1)p$,
$j=\lfloor r\rfloor$, $a=r-j$, and returns
$(1-a)x_{(j)}+ax_{(j+1)}$, with endpoints handled directly.
For histogram counts $n_b$ and widths $h_b$, probability per bin is
$n_b/n$, whereas density is $n_b/(nh_b)$. Only the latter integrates to
one after multiplication by bin widths. A logarithmic display does not
change those normalisations. Empty bins do not become observations through
display interpolation.

For paired errors $(x_i,y_i)$, each bootstrap draw samples indices
$i_1,\ldots,i_n$ with replacement to compute
$\operatorname{median}_j(x_{i_j}-y_{i_j})$. Percentiles over repeated draws
give intervals for the median casewise difference, not for independently
resampled medians or a percentage reduction. For recall, rerank both target
and predicted scores within each resample before evaluating
\eqref{eq:pinorecall}, because retaining the original rankings changes the
statistic.

\end{document}